\documentclass[11pt,letterpaper]{article}
\usepackage{epsfig,rotating,setspace,latexsym,amsbsy,amsmath,epsf,amssymb,bm}
\usepackage{cite,graphicx,authblk,color,subfigure}
\usepackage{multirow}
\usepackage{ctable}

\newlength{\Oldarrayrulewidth}
\newcommand{\Cline}[2]{%
  \noalign{\global\setlength{\Oldarrayrulewidth}{\arrayrulewidth}}%
  \noalign{\global\setlength{\arrayrulewidth}{#1}}\cline{#2}%
  \noalign{\global\setlength{\arrayrulewidth}{\Oldarrayrulewidth}}}
  
\title{The Capacity of Private Information Retrieval from Uncoded Storage Constrained Databases
\author{Mohamed Adel Attia \qquad Deepak Kumar \qquad Ravi Tandon}
\affil{Department of Electrical and Computer Engineering\\
University of Arizona, Tucson, AZ, USA.\\
E-mail: {\{\textit{madel, deepakkumar, tandonr}\}}@email.arizona.edu}}

\newcommand{\A}{{\mathbf{A}}}
\newcommand{\Q}{{\mathbf{Q}}}
\newcommand{\W}{{\mathbf{W}}}
\newcommand{\Z}{{\mathbf{Z}}}
\newcommand{\bsigma}{{\boldsymbol{\sigma}}}
\newcommand{\bpi}{{\boldsymbol{\pi}}}
\newcommand{\bdelta}{{\boldsymbol{\delta}}}

\newtheorem{theorem}{Theorem}
\newtheorem{lemma}{Lemma}
\newtheorem{claim}{Claim}
\newtheorem{remark}{Remark}
\newtheorem{example}{Example}

\newenvironment{proof}[1]{\medskip\par\noindent
{\bf Proof:\,}\,#1}{{\mbox{\,$\blacksquare$}\par}}
\allowdisplaybreaks

\begin{document}
\maketitle
\newcommand\blfootnote[1]{%
  \begingroup
  \renewcommand\thefootnote{}\footnote{#1}%
  \addtocounter{footnote}{-1}%
  \endgroup
}

\blfootnote{This work was supported by the NSF grants  CAREER-1651492 and CNS-1715947. This paper was presented in parts at 2018 IEEE International Conference on Communications (ICC), and 2018 IEEE International Symposium on Information Theory (ISIT).}

\thispagestyle{empty}

\begin{abstract}
 Private information retrieval (PIR) allows a user to retrieve a desired message from a set of databases without revealing the identity of the desired message. The replicated databases scenario was considered by Sun and Jafar in \cite{SunAndJaffar1}, where $N$ databases can store the same $K$ messages completely. A PIR scheme was developed to achieve the optimal download cost given by  $\left(1+ \frac{1}{N}+ \frac{1}{N^{2}}+ \cdots + \frac{1}{N^{K-1}}\right)$.
In this work, we consider the problem of PIR from \textit{storage constrained} databases. Each database has a storage capacity of $\mu KL$ bits, where $L$ is the size of each message in bits, and $\mu \in [1/N, 1]$ is the normalized storage.  
On one extreme, $\mu=1$ is the replicated databases case considered in \cite{SunAndJaffar1}. On the other hand, when $\mu= 1/N$, then in order to retrieve a message privately, the user has to download all the messages from the databases achieving a download cost of $1/K$. We aim to characterize the optimal download cost versus storage trade-off for any storage capacity in the range $\mu \in [1/N, 1]$.

In the storage constrained PIR problem, there are two key challenges: 
a) construction of communication efficient schemes through storage content design at each database that allow download efficient PIR; and 
b) characterizing the optimal download cost via information-theoretic lower bounds.
The novel aspect of this work is to characterize the optimum download cost of PIR from uncoded storage constrained databases for any value of storage.
 In particular, for any $(N,K)$,  we show that the optimal trade-off between storage, $\mu$, and the download cost, $D(\mu)$, is given by the lower convex hull of the $N$ pairs $\left(\frac{t}{N}, \left(1+ \frac{1}{t}+ \frac{1}{t^{2}}+ \cdots + \frac{1}{t^{K-1}}\right)\right)$ for $t=1,2,\ldots, N$. 
 To prove this result, we first present the storage constrained PIR scheme for any $(N,K)$. 
 We next obtain a general lower bound on the download cost for PIR, which is valid for the following storage scenarios: replicated or storage constrained, coded or uncoded, and fixed or optimized.
We then specialize this bound using the uncoded storage assumption to obtain the converse proof, i.e., obtaining lower bounds on the download cost for PIR from uncoded storage constrained databases as a function of the available storage.  Given the system constraints, we express the lower bound as a linear program (LP). We then solve the LP to obtain the best lower bounds on the download cost for different regimes of storage, which matches the proposed storage constrained PIR scheme.
\end{abstract}

\section{Introduction}
\label{sec:introduction}

With the paradigm-shifting developments towards distributed storage systems (DSS), assuring privacy while retrieving information from public databases has become a crucial need for users. This problem, also referred to as private information retrieval (PIR) has direct practical applications in cloud storage or social networking, with examples in many areas such as: protecting investors seeking stock market records or bank loans, customers searching for online reservations or insurance policies with potential fare increase due to frequent searching, or even activists seeking files that might be considered anti-regime. 
The original formulation  of the PIR problem involves $N$ non-colluding databases, each database identically replicates the storage of $K$ messages. Upon receiving queries from the legitimate user, the databases answer truthfully with the required information, which means they are curious but honest. Successful PIR must satisfy two properties: first, each of the $N$ message queries sent from the user to the N databases must reveal nothing about the identity of the message being requested; and second, the user must be able to correctly decode the message of interest from the answers received from the $N$ databases.

A trivial solution to PIR is to download all the messages from the databases. Indeed this approach will make the databases indifferent towards  the identity of the message being requested, but it is highly impractical especially when the number of messages $K$ is too large. The goal is to design an efficient protocol, which is characterized by the total upload/download cost the user has to pay in order to download a message privately.
The PIR problem has been studied extensively within the computer science community \cite{chor1995private,yekhanin2012locally,gasarch2004survey,song2000practical,ostrovsky2007survey}. 
In the pioneering work by Chor et al. \cite{chor1995private}, 
the authors considered one bit length messages. The privacy was based on cryptographic assumptions, which require NP hard computations to break. The privacy cost is calculated based on the total amount of communication between the user and the databases, i.e., the upload cost which is the size of the $N$ queries, and the download cost which is the size of the $N$ answers.

The Shannon theoretic approach for this problem is to allow the size of the messages to be arbitrary large, and therefore the upload cost is considered negligible with respect to the download cost, which is the more practical scenario.
Based on the Shannon theoretic formulation, the rate of a PIR scheme is the ratio between the number of desired vs downloaded bits, and PIR capacity is then defined as the maximum achievable rate.  
Under these assumptions, the first achievable PIR rate was found by Shah et al. in \cite{shah2014one} equal to $1-\frac{1}{N}$ for $N$ replicated non-colluding databases and any number of messages, which is a close approximation to the PIR capacity, characterized later in \cite{SunAndJaffar1}, for large number of databases. 
%It also shows that the privacy comes for free for infinite  number of messages at the databases.
In a recent interesting work by Sun and Jafar \cite{SunAndJaffar1}, the exact capacity of PIR (or the inverse of download cost) for any $(N,K)$ was characterized  as $(1 + \frac{1}{N}+\frac{1}{N^2} + \cdots+ \frac{1}{N^{K-1}})^{-1}$. 
%This expression converges to the previous result in \cite{} when $K\rightarrow \infty$, and has the same convergence as $N\rightarrow \infty$.

Since the appearance of \cite{SunAndJaffar1}, significant recent progress has been made on a variety of variations of the basic PIR problem.  The case of $T$-colluding PIR (or TPIR in short) was investigated by Sun and Jafar in \cite{SunAndJaffar2}, where any $T$ databases out of $N$ are able to collude by sharing their received queries. 
Robust PIR, in which any subset of $N$ databases out of $M$ fail to respond was also investigated in \cite{SunAndJaffar2}.
The exact capacity of robust TPIR for any $(T,M,N,K)$ and $M\geq N\geq T$ was characterized  as $(1 + \frac{T}{N}+\frac{T^2}{N^2} + \cdots+ \frac{T^{K-1}}{N^{K-1}})^{-1}$.
  The problem of PIR with databases storing coded messages, using $(N,M)$ MDS codes, was considered by Tajeddine and El-Rouayheb in \cite{Salim} and the capacity was subsequently characterized by Banawan and Ulukus in \cite{BanawanAndUlukus} to take the value $(1 + \frac{M}{N}+\frac{M^2}{N^2} + \cdots+ \frac{M^{K-1}}{N^{K-1}})^{-1}$.
This setting was further investigated for the scenario where any $T$ out of $N$ databases can collude, also referred to as MDS-TPIR \cite{FreijGnilkeHollantiKarpuk, SunAndJafar3} although its capacity remains open for general set of parameters.

As opposed to the original PIR setting where the user privacy is the only concern, the problem of symmetric  PIR (SPIR) was first studied by Sun and Jafar in \cite{SunAndJafar4}. In this setting, privacy is enforced in both directions, i.e.,  user must be able to retrieve the message of interest  privately, while at the same time the databases must avoid any information leakage to the user about the remaining messages. The capacity of SPIR was characterized as $1-\frac{1}{N}$ regardless the number of messages, $K$.
Later, the exact capacity for SPIR problem and $(M,N)$ MDS-coded   messages, or MDS-SPIR, was characterized by Wang and Skoglund in \cite{WangAndSkoglund} as $1-\frac{M}{N}$.
The capacity of cache aided PIR was recently characterized in \cite{TandonCachePIR}, where the user has a local cache of limited storage $0 \leq S\leq K$ and contents known to the databases. 
It was shown that memory sharing between full storage and no-cache-aided PIR schemes is information-theoretically optimal, i.e., the download cost is given by $(1-\frac{S}{K})(1 + \frac{1}{N}+\frac{1}{N^2} + \cdots+ \frac{1}{N^{K-1}})^{-1}$.
Later in \cite{Ulukus-NewCache, Kadhe-NewCache}, the capacity for cache aided PIR was studied when the side information at the user is unknown to the databases for both coded and uncoded storage scenarios. The capacity of PIR with private side information, or PIR-PSI, was characterized in \cite{Jafar-NewCache} to take the value $(1 + \frac{1}{N}+\frac{1}{N^2} + \cdots +\frac{1}{N^{K-M-1}})^{-1}$, where $M$ is the number of messages known as side information to the user.

The case of multi-message PIR, or MPIR, was investigated in \cite{BanawanAndUlukus3, ZhangAndGennian}, in which the user wants to privately retrieve $P\geq 1$ out of $K$ messages in one communication round. The capacity of multi-round PIR, where the queries in each round are allowed to depend on the answers received in previous rounds, was considered by Sun and Jafar in \cite{sun2018multiround}. Although no advantage in terms of capacity of having multi-rounds as opposed to the single round case considered in \cite{SunAndJaffar1}, it was shown that the multi-round queries can help in reducing the storage overhead at the databases.
 In a recent work by Banawan and Ulukus \cite{BanawanAndUlukus2}, the capacity of PIR with Byzantine databases (or BPIR) was characterized, i.e., a scenario in which any subset of databases are adversarial and they can respond with incorrect answers. Recently, Tajeddine et al. studied in \cite{tajeddine2018robust} the PIR problem in the presence of colluding, non-responsive and Byzantine databases.
The PIR through wiretap channel, or PIR-WTC, was considered by Banawan and Ulukus in \cite{banawan2018private}, where the user wants to retrieve a single message in the presence of an external eavesdropper that is adversely tries to decode the contents being sent. 
Adding the constraint of  asymmetric traffic for the databases was recently considered by Banawan and Ulukus in \cite{banawan2018asymmetry}. The results show strict loss in PIR capacity due to the asymmetric traffic constraints compared with the symmetric case in \cite{SunAndJaffar1}.

Majority of above works, however, assume the presence of replicated databases, each storing all the $K$ messages. Indeed, exceptions to this statement include the works on the case when database store MDS coded data, where the databases must satisfy the $M$-out-of-$N$ recoverability constraint. 
Furthermore, \cite{sun2018multiround} also investigated the problem of limited storage PIR for the special case of $K=2$ messages and $N=2$ databases. The authors presented interesting lower and upper bounds on the capacity for this special case, and show the optimality of the proposed scheme for the case of linear schemes.
More recently in \cite{tian2018shannon}, the authors proposed a non-linear scheme for the canonical case $K=2$ and $N=2$, showing that the proposed  non-linear scheme uses less storage than the optimal linear scheme when the retrieval rate is kept optimal.

\textbf{Summary of Contributions--} In this work, we consider the problem of PIR from \textit{uncoded storage constrained} databases. Each database has a storage capacity of $\mu KL$ bits, where $K$ is the number of messages, $L$ is the size of each message in bits, and $\mu \in [1/N, 1]$ is the normalized storage.  
On one extreme, $\mu = 1/N$ is the minimum storage at databases so that the user can retrieve any required message. If the user is interested in retrieving a message, he has to download all the $K$ messages from databases to achieve privacy. On the other hand, $\mu = 1$ is the replicated databases case settled in \cite{SunAndJaffar1}, where the download cost is minimal. Thus, we aim to characterize this trade-off for any value of $\mu$ between these two extremes.
In this work, we present an achievable scheme for the storage constrained PIR problem which works for all $(N, K)$. The achieved PIR download cost is given by the lower convex hull of the $N$ storage-cost pairs $\left(\frac{t}{N},1 + \frac{t}{N}+\frac{t^2}{N^2} + \cdots+ \frac{t^{K-1}}{N^{K-1}})^{-1}\right)$, where $t\in [1:N]$. The achievablility proof, which was presented in parts in \cite{ICC2017}, works as follows: For the discrete storage values $\mu =\frac{t}{N}$ where $\mu \in[1:N]$, the storage design at the databases is inspired by storage design schemes in the caching literature \cite{Maddah}, where the users prefetch popular content into memories to reduce peak traffic rates when downloading from a single server.
As opposed to caching schemes, the storage placement for storage constrained PIR scheme occurs at the databases which should span  the whole set of files.
Our storage design allows dividing the PIR scheme into blocks where only a subset of databases of size $t$ is involved. 
The storage points in between the discrete storage values can be achieved by a memory sharing argument, which is given by the lower convex hull of the achieved rate-storage pairs.

 As a first step in understanding the fundamental limits, we proved in \cite{ISIT2018} the optimality of our storage constrained scheme for the special case of $K=3$ messages, $N=3$ databases, and any storage value at the databases, under \textit{uncoded and symmetric assumptions} for the storage placement.
Our second main contribution of this paper is to show that the proposed scheme is information-theoretically optimal for any $(N,K,\mu)$, under \textit{uncoded storage placement} assumption. 
The key technical challenge in proving the lower bounds is dealing with all possible components of storage at the databases limited by the storage and the message size constraints, which significantly go beyond the techniques introduced in \cite{SunAndJaffar1}. 
To this end, we tailor the mutual information of the key components used in \cite{SunAndJaffar1} for the full storage setting to the case of limited storage. The main differences, however, is that we retain the terms of information theoretic capacity which was discarded in \cite{SunAndJaffar1}.
We factorize these terms to arrive to the first universal lower bound on the download cost, which is valid for both \textit{coded and uncoded placement} strategies.
Next, we specialize the obtained lower bound to uncoded placement strategies inspired by the methodology recently proposed in \cite{wan2016optimality} for uncoded caching systems, and later applied in our previous work on coded data shuffling \cite{attia2018near,Attia-ISIT2018-2}, and uncoded caching systems with secure delivery \cite{bahrami2017towards}. Applying these ideas help in obtaining a linear program (LP) subject to the storage and message size constraints, which can be solved for different regimes of storage to provide a set of lower bounds on the download cost, and show that these bounds exactly match the storage-constrained PIR scheme.

\subsection{Notation}
The notation ${\left[n_1:n_2\right]}$ for $n_1<n_2$, and $n_1,n_2\in \mathbb{N}$ represents the set of all integers between $n_1$, and $n_2$, i.e., ${\{n_1,n_1+1,\ldots, n_2\}}$. 
The combination coefficient ${{{n}\choose {k}}=\frac{n!}{(n-k)! k!}}$ equals zero for ${k>n}$, or $k<0$.
Elements of ordered sets are placed between round brackets $()$, while for regular sets we use curly brackets $\{\}$. 
We use  bold letters to represent ordered sets, and calligraphy letters for regular sets.
In order to describe subsets of ordered sets, we use the subscript to give the indexes of the elements being chosen from the set, e.g., for the ordered set $\bpi=(\pi_1\,\ldots,\pi_n)$, $\bpi_{[1:4]} = (\pi_1,\pi_2,\pi_3,\pi_4)$.
We denote Random Variables (RVs) by capital letters, and ordered sets of RVs by capital bold letters.
The set in the subscript of a set of ordered RVs  is used for short notation of a subset of the set of RVs, e.g.,  for an ordered set of RVs $\Z =(Z_1,\ldots,Z_n)$, we use $\Z_{\mathcal{W}}$ to denote the all the random variables $Z_i$ where ${i\in\mathcal{W}}$.
We use $W_k$ to denote a message of index $k$, and $w_k^j$ to represent the bit number $j$ of the message $W_k$.

\begin{figure}[t]
  \begin{center}
  \includegraphics[width=0.7\columnwidth]{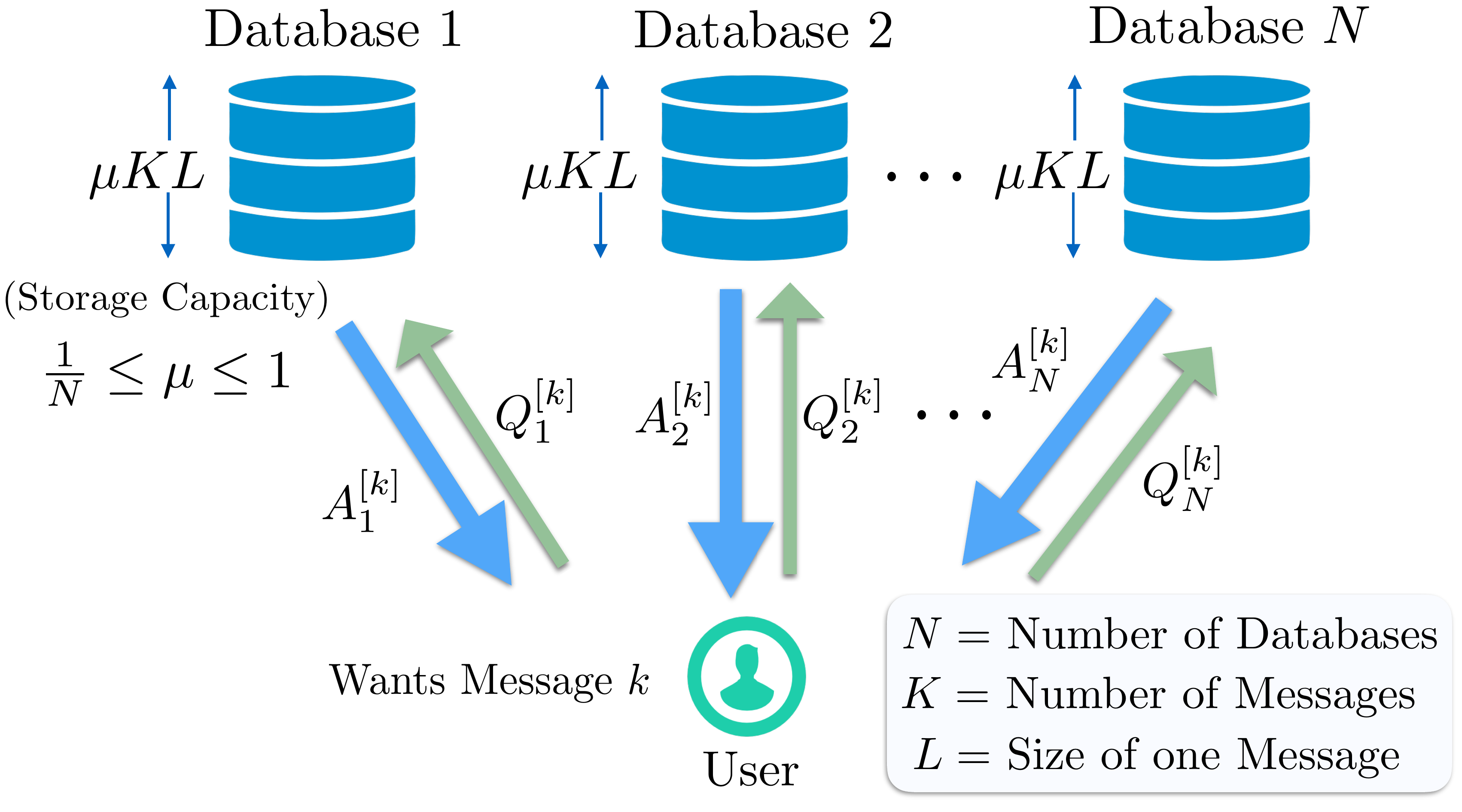}
\caption{Storage Constrained Private Information Retrieval. \label{Figmodel}}
\vspace{-10pt}
  \end{center}
\end{figure}

\section{Storage Constrained PIR: Problem Statement}
\label{sec:problem}

We consider the PIR problem with $N$ non-colluding databases, labeled as $\text{DB}_1,\text{DB}_2,\ldots,\text{DB}_N$,  and $K$ independent messages, labeled as $W_1,W_2,\ldots,W_K$, where each message is of size $L$ bits, i.e., 
\begin{equation}
H(W_1)=H(W_2)=\cdots=H(W_K)=L.
\end{equation}
We assume that each database has a storage capacity of $\mu KL$ bits. If we denote $Z_1, Z_2, \ldots, Z_N$ as the contents stored across the databases, where $Z_n$ is the storage content of DB$_{n}$, then we have the following \textit{storage constraint} for each database:
\begin{equation}
\label{eq:storage-const0}
H(Z_1)=H(Z_2)=\cdots=H(Z_N)=\mu K L.
\end{equation}
%We allow the user to design what contents can be stored at each database subject to the storage constraint. 
We assume that the storage strategy employed by the user is completely public, i.e., each database knows which contents are stored at all the other databases. The normalized storage parameter $\mu$ can take values in the range $1/N \leq \mu \leq 1$. The case when $\mu=1$ is the setting of replicated databases, with each database storing all the $K$ messages. The lower bound $\mu \geq 1/N$ is in fact a necessary condition for reliable decoding.

The storage-constrained PIR model is shown in Figure~\ref{Figmodel}. To request a message, a user privately selects a number $k$ between $1$ and $K$ corresponding to the desired message $W_{k}$. Then the user generates $N$ queries $Q_1^{[k]}, Q_2^{[k]},\ldots, Q_N^{[k]}$, where $Q_n^{[k]}$ is sent to the $n^{{th}}$ database (DB$_n$), and the queries are independent of the messages, i.e.,
\begin{align} 
 I(W_1,\ldots,W_K;Q_1^{[k]},\ldots,Q_N^{[k]})=0, \quad \forall k\in[1:K].\label{cons3}
\end{align}
Upon receiving the query $Q_n^{[k]}$, DB$_n$ returns an answer $A_n^{[k]}$ to the user, which  is a function of the corresponding query and the data stored in the DB$_n$, i.e., 
\begin{align}
H(A_n^{[k]}|Q_n^{[k]},Z_n)=0,\quad \forall k\in[1:K],\;\forall n\in[1:N].\label{cons4}
\end{align}

From all of the answers from databases, the user must be able to correctly decode the desired message $W_k$ with a small probability of error, i.e., the following \textit{correctness constraint} must be satisfied
\begin{equation}
\label{eq:decoding-const}
 H(W_k|A_1^{[k]},\ldots,A_N^{[k]},Q_1^{[k]},\ldots,Q_N^{[k]})=o(L), \quad \forall k\in[1:K], 
\end{equation}
where $o(L)$ represents a function of $L$ such that $o(L)/L$ approaches $0$ as $L$ approaches infinity. 
In order to prevent the database DB$_n$ from learning the identity of requested message, privacy must be guaranteed through the following statistical equivalent constraint for all $k_1 \neq k_2 \in [1:K]$:
\begin{align}
\label{eq:privacy-const0}
&(Q_n^{[k_1]},A_n^{[k_1]},W_1,\ldots,W_K, Z_1, \ldots, Z_N)\ \sim \ (Q_n^{[k_2]},A_n^{[k_2]},W_1,\ldots,W_K, Z_1, \ldots, Z_N).
\end{align}
In other words, the index of the desired message $k$ must be hidden from the query and answer of DB$_n$ as well as all the messages and the storage content of other databases, i.e.,
\begin{align}
\label{eq:privacy-const}
&I(k;Q_n^{[k]},A_n^{[k]},W_1,\ldots,W_K, Z_1, \ldots, Z_N) =0.
\end{align}

For a storage parameter $\mu$, we say that a pair $(D,L)$ is achievable if there exists a PIR scheme with storage, querying, and decoding functions, which satisfy the storage, correctness and privacy constraints in \eqref{eq:storage-const0}, \eqref{eq:decoding-const} and \eqref{eq:privacy-const0}, respectively. The performance of a PIR scheme is characterized by the number of bits of desired information per downloaded bit. In particular, if $D$ is the total number of downloaded bits, and $L$ is the size of the desired message, then the normalized downloaded cost is $D/L$. In other words, the PIR rate is $L/D$. 
The goal is to characterize the optimal normalized download cost as a function of the database storage parameter $\mu$:
\begin{align}
\label{optimal-download-cost}
D^{*}(\mu)= \text{min }\{D/L: (D,L) \text{ is achievable}\}.
\end{align}
The storage-constrained capacity of PIR is the inverse of the normalized download cost,
\begin{align}
C^{*}(\mu)= \text{max }\{L/D: (D,L) \text{ is achievable}\}.
\end{align}

We next present a claim which shows that the optimal download cost $D^{*}(\mu)$ (or the inverse of capacity $1/C^{*}(\mu)$) is a convex function of the normalized storage $\mu$. 

\begin{claim}\label{claim1}
The optimal download cost $D^{*}(\mu)$ is a convex function of $\mu$. In other words, for any $(\mu_1, \mu_2)$, and $\alpha \in [0,1]$, the following inequality is true:
\begin{align}
D^{*}(\alpha \mu_1 + (1-\alpha)\mu_2) \leq \alpha D^{*}(\mu_1) + (1-\alpha)D^{*}(\mu_2).
\end{align}
\end{claim}

\begin{proof}
Claim~\ref{claim1} follows from a simple memory sharing argument. Consider any two storage parameters $\mu_1$, and $\mu_2$, with optimal download costs $D^{*}(\mu_1)$, and $D^{*}(\mu_2)$, respectively, then for any storage parameter $\bar{\mu} = \alpha \mu_1 + (1-\alpha)\mu_2$, $\alpha\in[0,1]$, there exists a PIR scheme which achieves a download cost of $\bar{D}(\bar{\mu})=\alpha D^{*}(\mu_1) + (1-\alpha)D^{*}(\mu_2)$.
This is done as follows: first, we divide each message $W_k$ into two partitions $W_k= \left\{W_k^{(1)}, W_k^{(2)}\right\}$, where $W_{k}^{(1)}$ and $W_{k}^{(2)}$ are of size $\alpha L$ and $(1-\alpha)L$, respectively. 
Likewise, the storage of each database $Z_n$ is divided into two partitions $Z_n= \left\{Z_n^{(1)}, Z_n^{(2)}\right\}$, where $Z_n^{(1)}$ and $Z_n^{(2)}$ are of size $\alpha\mu_1KL $ and $(1-\alpha)\mu_2KL $, respectively. Now, for messages partitions denoted by $W_k^{(1)}$ for $k\in[1:K]$ and databases partitions denoted by $Z_n^{(1)}$ for $n\in[1:N]$, we can apply the PIR scheme which achieves a download cost of $\alpha D^{*}(\mu_1)$, while for messages partitions denoted by $W_k^{(2)}$ for $k\in[1:K]$ and databases partitions denoted by $Z_n^{(2)}$ for $n\in[1:N]$, we can achieve a download cost of $(1-\alpha) D^{*}(\mu_2)$, which gives a total download cost of $\bar{D}(\bar{\mu})=\alpha D^{*}(\mu_1) + (1-\alpha)D^{*}(\mu_2)$.
Since $D^{*}(\alpha \mu_1 + (1-\alpha)\mu_2)$ by definition is optimal download cost for the storage parameter $\bar{\mu}$, it cannot be larger than the download cost of the memory sharing scheme, i.e., $D^{*}(\alpha \mu_1 + (1-\alpha)\mu_2)\leq \bar{D}(\bar{\mu})=\alpha D^{*}(\mu_1) + (1-\alpha)D^{*}(\mu_2)$, which completes the proof of Claim~\ref{claim1}. 
\end{proof}

\subsection{Storage Constrained PIR: Uncoded Storage Assumption}
\label{re:uncoded-assump}

%For the scope of this paper, we focus on uncoded storage, where each database can store uncoded bits of each message subject to the storage constraint.

%\begin{remark}[\textbf{Uncoded Storage Assumption}]
Now, we specialize the above system model for the storage constrained PIR using uncoded storage assumption, where the databases only store uncoded functions of the $K$ messages subject to the storage constraint. We consider a generic uncoded placement strategy such that if we consider a message $W_k$, we denote $W_{k,\mathcal{S}}$ as the set of bits of $W_k$ that are fully stored at the databases in the set $\mathcal{S}\subseteq [1:N]$, where  $\vert \mathcal{S}\vert \geq 1$, and are not stored at any of the other databases in the set $[1:K]\setminus \mathcal{S}$. 
As a result, we can write the content of DB$_n$, $Z_n$ as
\begin{align}
\label{eq:storage-content2}
Z_n= \underset{k\in [1:K]}{\cup} \underset{\substack{\mathcal{S}\subseteq [1:N]\\ n\in \mathcal{S}}}{\cup} W_{k,\mathcal{S}}.
\end{align}
% \end{remark}
Furthermore, the message $W_k$ consists of  $2^{N}-1$ partitions, $W_{k,\mathcal{S}}$, for $\mathcal{S}\in \mathcal{P}([1:N])$, where $\mathcal{P}([1:N])$ is the power set of all possible subsets of the set $[1:N]$ not including the empty set.  Therefore, the message $W_k$ can also be equivalently expressed  as
\begin{align}
\label{eq:def_W_k}
W_k =\underset{\substack{\mathcal{S}\subseteq [1:N] \\  \vert \mathcal{S}\vert \geq 1}}{\cup} W_{k,\mathcal{S}}.
\end{align}
Now, let us consider $W_{k,\mathcal{S}}$ as a random variable with entropy $H(W_{k,\mathcal{S}})= \vert W_{k,\mathcal{S}}\vert L$, and size $\vert W_{k,\mathcal{S}} \vert$  normalized by the message size $L$. Therefore, the following two constraints are obtained:

\noindent$\bullet$\hspace{5pt}\textbf{Message size constraint:} The first constraint is related to the total size of all the messages, $W_k$ and $k\in[1:K]$, given by $KL$ bits,
\begin{align}
1 &= \frac{1}{KL} H(\W_{[1:K]}) = \frac{1}{KL} H(W_1,W_2,\cdots,W_K)\overset{(a)}{=}\frac{1}{KL} \sum_{k=1}^K H(W_k)\overset{(b)}{=} \frac{1}{KL}\sum_{k=1}^K \sum_{\substack{\mathcal{S}\subseteq [1:N]\\ \vert \mathcal{S}\vert \geq 1}} H(W_{k,\mathcal{S}})\nonumber\\
&=  \sum_{\ell=1}^N \frac{1}{K} \sum_{k=1}^K \sum_{\substack{\mathcal{S}\subseteq [1:N]\\ \vert \mathcal{S}\vert =\ell}} \vert W_{k,\mathcal{S}} \vert = \sum_{\ell=1}^N \binom{N}{\ell} \underbrace{\frac{1}{K\binom{N}{\ell}} \sum_{k=1}^K \sum_{\substack{\mathcal{S}\subseteq [1:N]\\ \vert \mathcal{S}\vert =\ell}} \vert W_{k,\mathcal{S}} \vert}_{\triangleq x_{\ell}} = \sum_{\ell =1}^N \binom{N}{\ell} x_{\ell},\label{eq:size-const}
\end{align}
where $(a)$ follows since the messages are independent, $(b)$ follows from \eqref{eq:def_W_k}, and $x_{\ell}\geq 0$ is defined as
\begin{align}
 x_{\ell}\: \overset{\Delta}{=}\:  \frac{1}{K\binom{N}{\ell}} \sum_{k=1}^K \sum_{\substack{\mathcal{S}\subseteq [1:N]\\ \vert \mathcal{S}\vert =\ell}} \vert W_{k,\mathcal{S}} \vert, \quad \ell\in[1:N].\label{eq:def-x_ell}
\end{align}

\noindent$\bullet$\hspace{5pt}\textbf{Storage constraint:} The second constraint is related to the total  storage of all the $N$ databases, which cannot exceed $\mu NKL$ bits for $\mu \in [\frac{1}{K},1]$,
\begin{align}
\mu N &\geq \frac{1}{KL} \sum_{n=1}^N H(Z_n) \overset{(a)}{=} \frac{1}{KL} \sum_{n=1}^N \sum_{k=1}^K \sum_{\substack{\mathcal{S}\subseteq [1:N]\\ n\in \mathcal{S}}} H(W_{k,\mathcal{S}}) 
\overset{(b)}{=} \frac{1}{K} \sum_{k=1}^K \sum_{\substack{\mathcal{S}\subseteq [1:N]\\ \vert \mathcal{S}\vert \geq 1}} \vert \mathcal{S}\vert \: \vert W_{k,\mathcal{S}}\vert
\nonumber \\
& = \sum_{\ell=1}^N  \frac{\ell}{K} \sum_{k=1}^K \sum_{\substack{\mathcal{S}\subseteq [1:N]\\ \vert  \mathcal{S}\vert =\ell}}  \vert W_{k,\mathcal{S}}\vert \overset{(c)}{=} \sum_{\ell=1}^K \ell\binom{N}{\ell} x_{\ell},\label{eq:storage-const}
\end{align}
where $(a)$ follows from \eqref{eq:storage-content2}, $(b)$ is true because when we sum up the contents of the storage at all the databases, the chunk $W_{k,\mathcal{S}}$ is counted $\vert \mathcal{S}\vert$ number of times, which is the number of databases storing this chunk, and $(c)$ follows from the definition of $x_{\ell}$ in \eqref{eq:def-x_ell}.
The message size and storage constraints defined in this sub-section will be used in the converse proofs for PIR from uncoded storage constrained databases.

\section{Main Result and Discussions}

Our first result is a general information theoretic lower bound on the download cost of the PIR problem from both coded or uncoded storage constrained databases.
\begin{theorem} \label{theorem1}
For the storage constrained PIR problem with $N$ databases, $K$ messages (of size $L$ bits each), and arbitrary storage $Z_1,Z_2,\ldots,Z_N$ at the $N$ databases, the optimal download cost is lower bounded as follows,
\begin{align}
\label{eq:theorem1}
D^*(\mu)&\geq  1 + \sum_{n_1=1}^{N} \frac{\lambda_{(N-n_1,1)} }{n_1}+ \sum_{n_1=1}^{N}\sum_{n_2=n_1}^{N} \frac{\lambda_{(N-n_1,2)}}{n_1n_2}  +\cdots+ \sum_{n_1=1}^{N}\ldots\hspace{-6pt}\sum_{n_{K-1}=n_{K-2}}^{N}\hspace{-4pt} \frac{\lambda_{(N-n_1,K-1)}}{n_1 \times\cdots\times n_{K-1}},
\end{align}
where $\lambda_{(n,k)}$ is defined as follows,
\begin{align}
\label{eq:def-L_nk}
\lambda_{(n,k)}\ \overset{\Delta}{=}\ \frac{1}{KL\binom{K-1}{k}\binom{N}{n}} \sum_{\mathcal{K}\subseteq[1:K] \atop \left|\mathcal{K}\right|=k} \sum_{\mathcal{N}\subseteq[1:N]\atop\left|\mathcal{N}\right|=n}\sum_{j\in [1:K]\backslash \mathcal{K}} H(W_j|\Z_{\mathcal{N}},\W_{\mathcal{K}}),
\end{align}
for $n\in[0:N]$ and $k\in[0:K]$.
\end{theorem}

\begin{remark}
The term $\lambda_{(n,k)}$ signifies the normalized average remaining entropy in a message after recognizing $k$ other messages and the storage from $n$ databases.
Henceforth, the result in Theorem~\ref{theorem1} gives a general information theoretic lower bound on the download cost of the PIR problem, which is valid  to cover the following scenarios: a) storage constrained databases where the content of each database can be designed for optimal PIR download cost; b) fixed storage content at the databases, where each database possesses different set of files; and c) both coded and uncoded storage content at the databases. 
\end{remark}

\noindent \textbf{\underline{Boundary Conditions on the function $\lambda_{(n,k)}$}:}
\begin{itemize}
\item We notice that when $n=N$ or $k=K$, then we get all the messages in the conditioning of the entropy terms of $\lambda_{(n,k)}$ in \eqref{eq:def-L_nk}, and therefore we get the following boundary conditions on $\lambda_{(n,k)}$:
\begin{align}
\label{eq:BC-L-nk}
\lambda_{(n=N,k)} = 0,\  \forall k\in[0:K], \qquad \lambda_{(n,k=K)} = 0,\ \forall n\in[0:N].
\end{align}

\item We further notice that for $n=0$ and all $k\in[0:K]$, then we only have messages in the conditioning of the entropy terms in \eqref{eq:def-L_nk} which are i.i.d., therefore, we get another set  of boundary conditions on $\lambda_{(n,k)}$:
\begin{align}
\label{eq:BC-L-nk2}
\lambda_{(n=0,k)} &= \frac{1}{KL\binom{K-1}{k}\binom{N}{n}} \sum_{\mathcal{K}\subseteq[1:K] \atop \left|\mathcal{K}\right|=k} \sum_{\mathcal{N}\subseteq[1:N]\atop\left|\mathcal{N}\right|=n}\sum_{j\in [1:K]\backslash \mathcal{K}} H(W_j)\nonumber\\
&=\frac{1}{KL\binom{K-1}{k}\binom{N}{n}} \sum_{j=1}^K \sum_{\mathcal{K}\subseteq[1:K]\backslash j \atop \left|\mathcal{K}\right|=k} \sum_{\mathcal{N}\subseteq[1:N]\atop\left|\mathcal{N}\right|=n} L \ = \ 1
,\quad  \forall k\in[0:K].
\end{align}

\item For the replicated databases case, considered in \cite{SunAndJaffar1} where every database stores all the files, then for the function $\lambda_{(n,k)}$ where $n\in[1:N]$, we retain all the messages in the conditioning of the entropy terms in \eqref{eq:def-L_nk}, which gives the following conditions over $\lambda_{(n,k)}$ only valid for the replicated databases case,
\begin{align}
\label{eq:BC-L-nk-replicated}
\lambda_{(n,k)} = 0,\  \forall n\in[1:N], \ \forall k\in[0:K].
\end{align}
\end{itemize}

\begin{remark}
We notice that for the replicated databases case considered in \cite{SunAndJaffar1}, by applying the boundary conditions in 
\eqref{eq:BC-L-nk2} and \eqref{eq:BC-L-nk-replicated} to the general  lower bound  in Theorem~\ref{theorem1}, we get the lower bound previously obtained in \cite{SunAndJaffar1} as follows,
\begin{align}
D^*(\mu = 1)&\geq  1 + \sum_{n_1=1}^{N} \frac{\lambda_{(N-n_1,1)} }{n_1}+ \sum_{n_1=1}^{N}\sum_{n_2=n_1}^{N} \frac{\lambda_{(N-n_1,2)}}{n_1n_2}  +\cdots+ \sum_{n_1=1}^{N}\ldots\hspace{-6pt}\sum_{n_{K-1}=n_{K-2}}^{N}\hspace{-4pt} \frac{\lambda_{(N-n_1,K-1)}}{n_1 \times\cdots\times n_{K-1}}\nonumber\\
&\overset{(a)}{=}1 + \sum_{n_1=N}^{N} \frac{\lambda_{(N-n_1,1)} }{n_1}+ \sum_{n_1=N}^{N}\sum_{n_2=n_1}^{N} \frac{\lambda_{(N-n_1,2)}}{n_1n_2}  +\cdots+ \sum_{n_1=N}^{N}\ldots\hspace{-6pt}\sum_{n_{K-1}=n_{K-2}}^{N}\hspace{-4pt} \frac{\lambda_{(N-n_1,K-1)}}{n_1 \times\cdots\times n_{K-1}}\nonumber\\
&=1 + \frac{\lambda_{(0,1)} }{N}+  \frac{\lambda_{(0,2)}}{N^2}  +\cdots+ \frac{\lambda_{(0,K-1)}}{N^{K-1}} \overset{(b)}{=}1+\frac{1}{N}+\frac{1}{N^2}+\cdots+\frac{1}{N^{K-1}},
\end{align}
where $(a)$ follows from the boundary conditions on $\lambda_{(n,k)}$ from replicated databases in \eqref{eq:BC-L-nk-replicated}, and $(b)$ follows from the boundary condition in \eqref{eq:BC-L-nk2}.
\end{remark}

The complete proof of Theorem \ref{theorem1} is presented in Section~\ref{sec:LB}. Here, we describe the intuition behind the proof through an example of $N=3$ databases and $K=3$ messages.  

\begin{example}
\label{ex:thm1}
\normalfont

We start by proving the following bound on $D$ which was first found in \cite[Lemma~1]{SunAndJaffar1} for $N=K=3$:
\begin{align}\label{upperbound}
D-L+o(L)\geq I(&\W_{[2:3]};\Q^{[1]}_{[1:3]},\A^{[1]}_{[1:3]}|W_1).
\end{align}
The above bound has an interesting interpretation that given message $W_1$ is requested, then the privacy penalty $D-L$ is bounded by the amount of information the queries and answers tell about the remaining messages $\W_{[2:3]}$ after successfully decoding message $W_1$.
To prove the above bound, we start by bounding the right hand side term as follows:
\begin{align} \label{upperbound-0}
 I(&\W_{[2:3]};\Q^{[1]}_{[1:3]},\A^{[1]}_{[1:3]}|W_1) \nonumber\\
 &=  I(\W_{[2:3]};\Q^{[1]}_{[1:3]},\A^{[1]}_{[1:3]},W_1)-  I(\W_{[2:3]};W_1)\nonumber\\
&\overset{(a)}{=} I(\W_{[2:3]};\Q^{[1]}_{[1:3]})+I(\W_{[2:3]};\A^{[1]}_{[1:3]}|\Q^{[1]}_{[1:3]})+I(\W_{[2:3]};W_1|\Q^{[1]}_{[1:3]},\A^{[1]}_{[1:3]})\nonumber\\
&\overset{(b)}{=} H(\A^{[1]}_{[1:3]}|\Q^{[1]}_{[1:3]})-H(\A^{[1]}_{[1:3]}|\Q^{[1]}_{[1:3]},\W_{[2:3]}) +I(\W_{[2:3]};W_1|\Q^{[1]}_{[1:3]},\A^{[1]}_{[1:3]})\nonumber\\
&\overset{}{\leq} D-H(\A^{[1]}_{[1:3]},W_1|\Q^{[1]}_{[1:3]},\W_{[2:3]})+H(W_1|\A^{[1]}_{[1:3]},\Q^{[1]}_{[1:3]},\W_{[2:3]})+I(\W_{[2:3]};W_1|\Q^{[1]}_{[1:3]},\A^{[1]}_{[1:3]})\nonumber\\
&\overset{}{=} D-H(W_1|\Q^{[1]}_{[1:3]},\W_{[2:3]})-H(\A^{[1]}_{[1:3]}|\Q^{[1]}_{[1:3]},\W_{[1:3]})+H(W_1|\Q^{[1]}_{[1:3]},\A^{[1]}_{[1:3]})\nonumber\\
&\overset{(c)}{=} D-L+o(L),
\end{align}
where $(a)$ follows from the chain rule of mutual information and from the assumption that messages are independent,
$(b)$ follows from \eqref{cons3} where queries are independent from messages, 
and $(c)$ follows from \eqref{cons4} where the answers $\A^{[1]}_{[1:3]}$ are functions of the queries $\Q^{[1]}_{[1:3]}$ and all messages and from the decodability constraint in \eqref{eq:decoding-const} where  message $W_1$ must be decoded from the queries $\Q^{[1]}_{[1:3]}$ and the answers $\A^{[1]}_{[1:3]}$.  Later, we prove a more general form of this bound in Lemma~\ref{lem2}. 
Using the chain rule for mutual information in all possible orders for a permutation ${\bsigma}:(1,2,3)\rightarrow (\sigma_1,\sigma_2,\sigma_3)$, we expand the RHS of the bound in \eqref{upperbound} as,
\begin{align}
I(\W_{[2:3]};\Q^{[1]}_{[1:3]},\A^{[1]}_{[1:3]}|W_1)&=I(\W_{[2:3]};Q^{[1]}_{\sigma_1},A^{[1]}_{\sigma_1}|W_1)  +I(\W_{[2:3]};Q^{[1]}_{\sigma_2},A^{[1]}_{\sigma_2}|W_1,Q^{[1]}_{\sigma_1},A^{[1]}_{\sigma_1})\nonumber\\ 
& \hspace{15pt}+I(\W_{[2:3]};Q^{[1]}_{\sigma_3},A^{[1]}_{\sigma_3}|W_1,Q^{[1]}_{\sigma_1},A^{[1]}_{\sigma_1},Q^{[1]}_{\sigma_2},A^{[1]}_{\sigma_2})\label{lemma1_2_1}\\
&\overset{(a)}{\geq} I(\W_{[2:3]};Q^{[1]}_{\sigma_1},A^{[1]}_{\sigma_1}|W_1)  +I(\W_{[2:3]};Q^{[1]}_{\sigma_2},A^{[1]}_{\sigma_2}|W_1,Z_{\sigma_1})\nonumber\\ 
& \hspace{15pt}+I(\W_{[2:3]};Q^{[1]}_{\sigma_3},A^{[1]}_{\sigma_3}|W_1,Z_{\sigma_1},Z_{\sigma_2}), \label{lemma1_2_2}
\end{align}
where $(a)$ follows by bounding the second and the third terms in \eqref{lemma1_2_1} separately as follows:
\begin{align}
I(\W_{[2:3]};Q^{[1]}_{\sigma_2},&A^{[1]}_{\sigma_2}|W_1,Q^{[1]}_{\sigma_1},A^{[1]}_{\sigma_1})\nonumber\\
&=H(Q^{[1]}_{\sigma_2},A^{[1]}_{\sigma_2}|W_1,Q^{[1]}_{\sigma_1},A^{[1]}_{\sigma_1}) -H(Q^{[1]}_{\sigma_2},A^{[1]}_{\sigma_2}|\W_{[1:3]}, Q^{[1]}_{\sigma_1},A^{[1]}_{\sigma_1})\nonumber\\
&\overset{(a)}{\geq} H(Q^{[1]}_{\sigma_2},A^{[1]}_{\sigma_2}|W_1,Z_{\sigma_1},Q^{[1]}_{\sigma_1},A^{[1]}_{\sigma_1}) -H(Q^{[1]}_{\sigma_2},A^{[1]}_{\sigma_2}|\W_{[1:3]},Z_{\sigma_1},Q^{[1]}_{\sigma_1},A^{[1]}_{\sigma_1})\nonumber\\
&\overset{(b)}{=} H(Q^{[1]}_{\sigma_2},A^{[1]}_{\sigma_2}|W_1,Z_{\sigma_1},Q^{[1]}_{\sigma_1}) -H(Q^{[1]}_{\sigma_2},A^{[1]}_{\sigma_2}|\W_{[1:3]},Z_{\sigma_1},Q^{[1]}_{\sigma_1})\nonumber\\ 
&=I(\W_{[2:3]};Q^{[1]}_{\sigma_2},A^{[1]}_{\sigma_2}|W_1,Z_{\sigma_1},Q^{[1]}_{\sigma_1})\nonumber\\
&\overset{(c)}{=}I(\W_{[2:3]};Q^{[1]}_{\sigma_1},Q^{[1]}_{\sigma_2},A^{[1]}_{\sigma_2}|W_1,Z_{\sigma_1})\nonumber\\
&\geq I(\W_{[2:3]};Q^{[1]}_{\sigma_2},A^{[1]}_{\sigma_2}|W_1,Z_{\sigma_1})\label{lemma1_3},
\end{align}
where $(a)$ follows from the fact that conditioning reduces entropy (which allows us to introduce $Z_{\sigma_1}$ in the first term), and the fact that $Z_{\sigma_1}$ is a function of $\W_{[1:3]}$ (hence it can be introduced in the second term); $(b)$ follows from the fact that  $A_{\sigma_1}^{[1]}$ is a function of $(Z_{\sigma_1}, Q_{\sigma_1}^{[1]})$; and step $(c)$ follows from (\ref{cons3}) where the queries are independent from the messages.

 In similar steps to \eqref{lemma1_3}, we bound the third term in \eqref{lemma1_2_1} as follows:
\begin{align}
I(\W_{[2:3]};Q^{[1]}_{\sigma_3},&A^{[1]}_{\sigma_3}|W_1,\Q^{[1]}_{\bsigma_{[1:2]}},\A^{[1]}_{\bsigma_{[1:2]}}) \nonumber \\
&=H(Q^{[1]}_{\sigma_3},A^{[1]}_{\sigma_3}|W_1,\Q^{[1]}_{\bsigma_{[1:2]}},\A^{[1]}_{\bsigma_{[1:2]}}) -H(Q^{[1]}_{\sigma_3},A^{[1]}_{\sigma_3}|\W_{[1:3]}, \Q^{[1]}_{\bsigma_{[1:2]}},\A^{[1]}_{\bsigma_{[1:2]}}) \nonumber \\
&{\geq} H(Q^{[1]}_{\sigma_3},A^{[1]}_{\sigma_3}|W_1,\Z_{\bsigma_{[1:2]}},\Q^{[1]}_{\bsigma_{[1:2]}},\A^{[1]}_{\bsigma_{[1:2]}}) -H(Q^{[1]}_{\sigma_3},A^{[1]}_{\sigma_3}|\W_{[1:3]},\Z_{\bsigma_{[1:2]}},\Q^{[1]}_{\bsigma_{[1:2]}},\A^{[1]}_{\bsigma_{[1:2]}})\nonumber \\
&{=} H(Q^{[1]}_{\sigma_3},A^{[1]}_{\sigma_3}|W_1,\Z_{\bsigma_{[1:2]}},\Q^{[1]}_{\bsigma_{[1:2]}}) -H(Q^{[1]}_{\sigma_3},A^{[1]}_{\sigma_3}|\W_{[1:3]},\Z_{\bsigma_{[1:2]}},\Q^{[1]}_{\bsigma_{[1:2]}})\nonumber \\
&=I(\W_{[2:3]};Q^{[1]}_{\sigma_3},A^{[1]}_{\sigma_3}|W_1,\Z_{\bsigma_{[1:2]}},\Q^{[1]}_{\bsigma_{[1:2]}})\nonumber\\
&{=}I(\W_{[2:3]};\Q^{[1]}_{\bsigma_{[1:3]}},\A^{[1]}_{\bsigma_3}|W_1,\Z_{\bsigma_{[1:2]}})\nonumber\\
&\geq I(\W_{[2:3]};Q^{[1]}_{\sigma_3},A^{[1]}_{\sigma_3}|W_1,\Z_{\bsigma_{[1:2]}})\label{lemma1_4},
\end{align}
which completes the proof of step $(a)$ in \eqref{lemma1_2_2}.
Later in Lemma~\ref{lem1}, we generally prove that a mutual information with queries and answers of some databases in the conditioning can be lower bounded by replacing the queries and the answers with the corresponding databases storage random variables.

Next, we sum \eqref{lemma1_2_2} over all possible permutations $\mathbf{\bsigma}\in [3!]$ to get the following bound:
\begin{align}
&I(\W_{[2:3]};\Q^{[1]}_{[1:3]},\A^{[1]}_{[1:3]}|W_1)\nonumber\\
&\geq  \frac{1}{3} \sum_{i=1}^3 I(\W_{[2:3]};Q^{[1]}_{i},A^{[1]}_{i}|W_1)  +\frac{1}{6}\sum_{i=1}^3 \sum_{j\in[1:3]\setminus i} I(\W_{[2:3]};Q^{[1]}_{j},A^{[1]}_{j}|W_1,Z_{i})\nonumber\\ 
& \hspace{15pt}+\frac{1}{3} \sum_{i=1}^3 I(\W_{[2:3]};Q^{[1]}_{i},A^{[1]}_{i}|W_1,\Z_{[1:3]\setminus i})\nonumber\\ 
&\overset{(a)}{=}  \frac{1}{3} \sum_{i=1}^3 I(\W_{[2:3]};Q^{[2]}_{i},A^{[2]}_{i}|W_1)  +\frac{1}{6}\sum_{i=1}^3 \sum_{j\in[1:3]\setminus i} I(\W_{[2:3]};Q^{[2]}_{j},A^{[2]}_{j}|W_1,Z_{i})\nonumber\\ 
& \hspace{15pt}+\frac{1}{3} \sum_{i=1}^3 I(\W_{[2:3]};Q^{[2]}_{i},A^{[2]}_{i}|W_1,\Z_{[1:3]\setminus i})\nonumber\\ 
&\geq  \frac{1}{3} \sum_{i=1}^3 I(\W_{[2:3]};A^{[2]}_{i}|W_1,Q^{[2]}_{i})  +\frac{1}{6}\sum_{i=1}^3 \sum_{j\in[1:3]\setminus i} I(\W_{[2:3]};A^{[2]}_{j}|W_1,Z_{i},Q^{[2]}_{j})\nonumber\\ 
& \hspace{15pt}+\frac{1}{3} \sum_{i=1}^3 I(\W_{[2:3]};A^{[2]}_{i}|W_1,\Z_{[1:3]\setminus i},Q^{[2]}_{i})\nonumber\\ 
&\overset{(b)}{=}   \frac{1}{3} \sum_{i=1}^3 H(A^{[2]}_{i}|W_1,Q^{[2]}_{i})  +\frac{1}{6}\sum_{i=1}^3 \sum_{j\in[1:3]\setminus i} H(A^{[2]}_{j}|W_1,Z_{i},Q^{[2]}_{j})+\frac{1}{3} \sum_{i=1}^3 H(A^{[2]}_{i}|W_1,\Z_{[1:3]\setminus i},Q^{[2]}_{i})\nonumber\\ 
&\overset{(c)}{\geq}  \frac{1}{3}  H(\A^{[2]}_{[1:3]}|W_1,\Q^{[2]}_{[1:3]})  +\frac{1}{6}\sum_{i=1}^3  H(\A^{[2]}_{[1:3]\setminus i}|W_1,Z_{i},\Q^{[2]}_{[1:3]})+\frac{1}{3} \sum_{i=1}^3 H(A^{[2]}_{i}|W_1,\Z_{[1:3]\setminus i},\Q^{[2]}_{[1:3]})\nonumber\\ 
&\overset{(d)}{=} \frac{1}{3}  I(\W_{[2:3]};\A^{[2]}_{[1:3]}|W_1,\Q^{[2]}_{[1:3]})  +\frac{1}{6}\sum_{i=1}^3  I(\W_{[2:3]};\A^{[2]}_{[1:3]\setminus i}|W_1,Z_{i},\Q^{[2]}_{[1:3]})\nonumber\\ 
& \hspace{15pt}+\frac{1}{3} \sum_{i=1}^3 I(\W_{[2:3]};A^{[2]}_{i}|W_1,\Z_{[1:3]\setminus i},\Q^{[2]}_{[1:3]})\nonumber\\
&\overset{(e)}{=} \frac{1}{3}  I(\W_{[2:3]};\Q^{[2]}_{[1:3]},\A^{[2]}_{[1:3]}|W_1)  +\frac{1}{6}\sum_{i=1}^3  I(\W_{[2:3]};\Q^{[2]}_{[1:3]},\A^{[2]}_{[1:3]\setminus i}|W_1,Z_{i})\nonumber\\ 
& \hspace{15pt}+\frac{1}{3} \sum_{i=1}^3 I(\W_{[2:3]};\Q^{[2]}_{[1:3]},A^{[2]}_{i}|W_1,\Z_{[1:3]\setminus i})\nonumber\\ 
&\overset{(f)}{=} \frac{1}{3}  I(\W_{[2:3]};W_2,\Q^{[2]}_{[1:3]},\A^{[2]}_{[1:3]}|W_1)  +\frac{1}{6}\sum_{i=1}^3  I(\W_{[2:3]};W_2,\Q^{[2]}_{[1:3]},\A^{[2]}_{[1:3]\setminus i}|W_1,Z_{i})\nonumber\\ 
& \hspace{15pt}+\frac{1}{3} \sum_{i=1}^3 I(\W_{[2:3]};W_2,\Q^{[2]}_{[1:3]},A^{[2]}_{i}|W_1,\Z_{[1:3]\setminus i})+o(L)\nonumber\\ 
&= \frac{1}{3}  I(\W_{[2:3]};W_2|W_1)  +\frac{1}{6}\sum_{i=1}^3  I(\W_{[2:3]};W_2|W_1,Z_{i})+\frac{1}{3} \sum_{i=1}^3 I(\W_{[2:3]};W_2|W_1,\Z_{[1:3]\setminus i})\nonumber\\ 
&\hspace{15pt}+ \frac{1}{3}  I(W_{3};\Q^{[2]}_{[1:3]},\A^{[2]}_{[1:3]}|\W_{[1:2]})  +\frac{1}{6}\sum_{i=1}^3  I(W_{3};\Q^{[2]}_{[1:3]},\A^{[2]}_{[1:3]\setminus i}|\W_{[1:2]},Z_{i})\nonumber\\ 
& \hspace{15pt}+\frac{1}{3} \sum_{i=1}^3 I(W_{3};\Q^{[2]}_{[1:3]},A^{[2]}_{i}|\W_{[1:2]},\Z_{[1:3]\setminus i})+o(L)
\nonumber\\ 
&\overset{(g)}{=} \frac{1}{3}  H(W_2|W_1) +\frac{1}{6}\sum_{i=1}^3  H(W_2|W_1,Z_{i})+\frac{1}{3} \sum_{i=1}^3 H(W_2|W_1,\Z_{[1:3]\setminus i})+ \frac{1}{3}  \underbrace{I(W_{3};\Q^{[2]}_{[1:3]},\A^{[2]}_{[1:3]}|\W_{[1:2]})}_{\triangleq\text{Term}_{1}}  \nonumber\\ 
&\hspace{15pt}+\frac{1}{6}\sum_{i=1}^3  \underbrace{I(W_{3};\Q^{[2]}_{[1:3]\setminus i},\A^{[2]}_{[1:3]\setminus i}|\W_{[1:2]},Z_{i})}_{\triangleq\text{Term}_{2}}+\frac{1}{3} \sum_{i=1}^3 \underbrace{I(W_{3};Q^{[2]}_{i},A^{[2]}_{i}|\W_{[1:2]},\Z_{[1:3]\setminus i})}_{\triangleq\text{Term}_{3}}+o(L),
 \label{lemma1_5}
\end{align}
where $(a)$ follows from the privacy constraint in \eqref{eq:privacy-const} where the individual queries and answers are invariant with respect to the requested message index; $(b)$ and $(d)$ follow since any answer $A^{[2]}_i$ is a function of the messages $\W_{[1:K]}$ and the query $Q^{[2]}_i$; $(c)$ follows from the union bound and since conditioning reduces entropy; $(e)$ and $(g)$ follow from the fact that queries are independent from the messages; and $(f)$ follows from the decoding constraint in \eqref{optimal-download-cost} where $W_{2}$ is decodable from $\Q_{[1:3]}^{[2]}$, $\A_{[1:3]\setminus \mathcal{N}}^{[2]}$ and $\Z_{\mathcal{N}}$ for any $\mathcal{N}\subseteq[1:3]$.
Next, we are going to lower bound the three terms, Term$_1$, Term$_2$ and Term$_3$, in \eqref{lemma1_5} separately. Bounding Term$_1$ follows similar steps to \eqref{lemma1_5} as,
\begin{align}
\text{Term}_{1} &= I(W_{3};\Q^{[2]}_{[1:3]},\A^{[2]}_{[1:3]}|\W_{[1:2]})\nonumber\\
&\overset{(a)}{\geq} \frac{1}{3}  I(W_{3};\Q^{[3]}_{[1:3]},\A^{[3]}_{[1:3]}|\W_{[1:2]})  +\frac{1}{6}\sum_{i=1}^3  I(W_{3};\Q^{[3]}_{[1:3]},\A^{[3]}_{[1:3]\setminus i}|\W_{[1:2]},Z_{i})\nonumber\\ 
& \hspace{15pt}+\frac{1}{3} \sum_{i=1}^3 I(W_{3};\Q^{[3]}_{[1:3]},A^{[3]}_{i}|\W_{[1:2]},\Z_{[1:3]\setminus i})\nonumber\\ 
&\overset{(b)}{\geq} \frac{1}{3}  I(W_{3};W_{3},\Q^{[3]}_{[1:3]},\A^{[3]}_{[1:3]}|\W_{[1:2]})  +\frac{1}{6}\sum_{i=1}^3  I(W_{3};W_{3},\Q^{[3]}_{[1:3]},\A^{[3]}_{[1:3]\setminus i}|\W_{[1:2]},Z_{i})\nonumber\\ 
& \hspace{15pt}+\frac{1}{3} \sum_{i=1}^3 I(W_{3};W_{3},\Q^{[3]}_{[1:3]},A^{[3]}_{i}|\W_{[1:2]},\Z_{[1:3]\setminus i})+o(L)\nonumber\\ 
&\overset{}{=} \frac{1}{3}  H(W_3|\W_{[1:2]})  +\frac{1}{6}\sum_{i=1}^3  H(W_3|\W_{[1:2]},Z_{i})+\frac{1}{3} \sum_{i=1}^3 H(W_3|\W_{[1:2]},\Z_{[1:3]\setminus i})+o(L),
 \label{lemma1_6}
\end{align}
where $(a)$ follows in very similar steps to $(a)\rightarrow (e)$ in \eqref{lemma1_5}; and $(b)$ follows from the decoding constraint in \eqref{eq:decoding-const} where $W_{3}$ is decodable from $\Q_{[1:3]}^{[3]}$, $\A_{[1:3]\setminus \mathcal{N}}^{[3]}$ and $\Z_{\mathcal{N}}$ for any $\mathcal{N}\subseteq[1:3]$. We bound Term$_2$ as follows:
\begin{align}
\text{Term}_{2} &=  I(W_{3};\Q^{[2]}_{[1:3]\setminus i},\A^{[2]}_{[1:3]\setminus i}|\W_{[1:2]},Z_{i})\nonumber\\
&\overset{(a)}{=}\frac{1}{2}\sum_{j\in[1:3]\setminus i} \left[I(W_{3};Q^{[2]}_{j},A^{[2]}_{j}|\W_{[1:2]},Z_{i})+ I(W_{3};Q^{[2]}_{j},A^{[2]}_{j}|\W_{[1:2]},Z_{i},\Q^{[2]}_{[1:3]\setminus\{i,j\}},\A^{[2]}_{[1:3]\setminus\{i,j\}})\right]\nonumber\\
&\overset{(b)}{\geq}\frac{1}{2}\sum_{j\in[1:3]\setminus i}\left[ I(W_{3};Q^{[3]}_{j},A^{[3]}_{j}|\W_{[1:2]},Z_{i})+I(W_{3};Q^{[3]}_{j},A^{[3]}_{j}|\W_{[1:2]},\Z_{[1:3]\setminus j})\right]\nonumber\\
&\geq\frac{1}{2}\sum_{j\in[1:3]\setminus i}\left[ I(W_{3};A^{[3]}_{j}|\W_{[1:2]},Z_{i},Q^{[3]}_{j})+I(W_{3};A^{[3]}_{j}|\W_{[1:2]},\Z_{[1:3]\setminus j},Q^{[3]}_{j})\right]\nonumber\\
&\geq\frac{1}{2}\sum_{j\in[1:3]\setminus i}\left[ H(A^{[3]}_{j}|\W_{[1:2]},Z_{i},Q^{[3]}_{j})+H(A^{[3]}_{j}|\W_{[1:2]},\Z_{[1:3]\setminus j},Q^{[3]}_{j})\right]\nonumber\\
&\geq\frac{1}{2} H(\A^{[3]}_{[1:3]\setminus i}|\W_{[1:2]},Z_{i},\Q^{[3]}_{[1:3]})+\frac{1}{2}\sum_{j\in[1:3]\setminus i} H(A^{[3]}_{j}|\W_{[1:2]},\Z_{[1:3]\setminus j},\Q^{[3]}_{[1:3]})\nonumber\\
&\overset{(c)}{=}\frac{1}{2} I(W_3;\A^{[3]}_{[1:3]\setminus i}|\W_{[1:2]},Z_{i},\Q^{[3]}_{[1:3]})+\frac{1}{2}\sum_{j\in[1:3]\setminus i} I(W_3;A^{[3]}_{j}|\W_{[1:2]},\Z_{[1:3]\setminus j},Q^{[3]}_{[1:3]})\nonumber\\
&\overset{(d)}{=}\frac{1}{2} I(W_3;W_3,\Q^{[3]}_{[1:3]},\A^{[3]}_{[1:3]\setminus i}|\W_{[1:2]},Z_{i})\hspace{-1pt}+\hspace{-1pt}\frac{1}{2}\hspace{-2pt}\sum_{j\in[1:3]\setminus i} \hspace{-8pt}I(W_3;W_3,\Q^{[3]}_{[1:3]},A^{[3]}_{j}|\W_{[1:2]},\Z_{[1:3]\setminus j})
\hspace{-1pt}+\hspace{-1pt}o(L)\nonumber\\
&=\frac{1}{2} H(W_3|\W_{[1:2]},Z_{i})+\frac{1}{2}\sum_{j\in[1:3]\setminus i} H(W_3|\W_{[1:2]},\Z_{[1:3]\setminus j})
+o(L),
 \label{lemma1_7}
\end{align}
where $(a)$ follows from the chain rule of mutual information; $(b)$ follows from steps similar to \eqref{lemma1_3} and \eqref{lemma1_4} which we will formally prove later in Lemma~\ref{lem1}; $(c)$ follows since any answer $A^{[3]}_i$ is a function of the messages $\W_{[1:K]}$ and the query $Q^{[3]}_i$; and $(d)$ follows from the decoding constraint in \eqref{eq:decoding-const}.
Then, we bound Term$_3$ as follows:
\begin{align}
\text{Term}_{3} &= I(W_{3};Q^{[2]}_{i},A^{[2]}_{i}|\W_{[1:2]},\Z_{[1:3]\setminus i})\nonumber\\
&\geq I(W_{3};A^{[2]}_{i}|\W_{[1:2]},\Z_{[1:3]\setminus i},Q^{[2]}_{i})\overset{(a)}{=} H(A^{[2]}_{i}|\W_{[1:2]},\Z_{[1:3]\setminus i},Q^{[2]}_{i})\nonumber\\
& \geq  H(A^{[2]}_{i}|\W_{[1:2]},\Z_{[1:3]\setminus i},\Q^{[2]}_{[1:3]}) \overset{(b)}{=} I(W_{3};A^{[2]}_{i}|\W_{[1:2]},\Z_{[1:3]\setminus i},\Q^{[2]}_{[1:3]})\nonumber\\
&\overset{(c)}{=} I(W_{3};W_3,\Q^{[2]}_{[1:3]},A^{[2]}_{i}|\W_{[1:2]},\Z_{[1:3]\setminus i})+o(L)\nonumber\\
&=  H(W_3|\W_{[1:2]},\Z_{[1:3]\setminus i})+o(L),
 \label{lemma1_8}
\end{align}
where $(a)$ and $(b)$ follows since any answer $A^{[2]}_i$ is a function of the messages $\W_{[1:K]}$ and the query $Q^{[2]}_i$; and $(c)$ follows from the decoding constraint in \eqref{eq:decoding-const}.
	
Using the lower bounds \eqref{lemma1_6}, \eqref{lemma1_7} and \eqref{lemma1_8} in  \eqref{lemma1_5} and hence use the resultant bound in \eqref{upperbound}, we arrive to the following bound:
\begin{align}
D\geq& L+\frac{1}{3}  H(W_2|W_1)+\frac{1}{9}  H(W_3|\W_{[1:2]})+\frac{1}{6}\sum_{i=1}^3  H(W_2|W_1,Z_{i})+\frac{1}{3} \sum_{i=1}^3 H(W_2|W_1,\Z_{[1:3]\setminus i})\nonumber\\
&+\frac{5}{36}\sum_{i=1}^3  H(W_3|\W_{[1:2]},Z_{i})+\frac{11}{18}\sum_{i=1}^3  H(W_3|\W_{[1:2]},\Z_{[1:3]\setminus i}) +o(L).
 \label{lemma1_9}
\end{align}

By repeating the bounding procedure in \eqref{upperbound} with any permutation of the messages indexes $\mathbf{\bpi}:(1,2,3)\rightarrow (\pi_1,\pi_2,\pi_3)$, we get a more general bound on $D$:
\begin{align}
D\geq& L+\frac{1}{3}  H(W_{\pi_2}|W_{\pi_1})+\frac{1}{9}  H(W_{\pi_3}|\W_{\bpi_{[1:2]}})+\frac{1}{6}\sum_{i=1}^3  H(W_{\pi_2}|W_{\pi_1},Z_{i})+\frac{1}{3} \sum_{i=1}^3 H(W_{\pi_2}|W_{\pi_1},\Z_{[1:3]\setminus i})\nonumber\\
&+\frac{5}{36}\sum_{i=1}^3  H(W_{\pi_3}|\W_{\bpi_{[1:2]}},Z_{i})+\frac{11}{18}\sum_{i=1}^3  H(W_{\pi_3}|\W_{\bpi_{[1:2]}},\Z_{[1:3]\setminus i}) +o(L).
 \label{lemma1_10}
\end{align}
Now, if we sum up (\ref{lemma1_10}) over all permutations $\mathbf{\bpi}\in[3!]$, we obtain,
	\begin{align}
D&\geq L+\frac{1}{6}\sum_{\mathcal{K}\subseteq[1:3] \atop \left|\mathcal{K}\right|=1}\sum_{k\in [1:3]\backslash \mathcal{K}}\left(\frac{1}{3}  H(W_{k}|\W_{\mathcal{K}})+\frac{1}{6}\sum_{i=1}^3  H(W_{k}|\W_{\mathcal{K}},Z_{i})+\frac{1}{3} \sum_{i=1}^3 H(W_{k}|\W_{\mathcal{K}},\Z_{[1:3]\setminus i})\right) \nonumber\\
&\hspace{5pt}+\hspace{-2pt}\frac{1}{6}\hspace{-3pt}\sum_{\mathcal{K}\subseteq[1:3] \atop \left|\mathcal{K}\right|=2}\sum_{k\in [1:3]\backslash \mathcal{K}}\hspace{-2pt}\left(\frac{2}{9}  H(W_{k}|\W_{\mathcal{K}})+\frac{5}{18}\sum_{i=1}^3  H(W_{k}|\W_{\mathcal{K}},Z_{i})\hspace{-3pt}+\hspace{-3pt}\frac{11}{9} \sum_{i=1}^3 H(W_{k}|\W_{\mathcal{K}},\Z_{[1:3]\setminus i})\right)
+o(L)\nonumber\\
&=L+ \frac{1}{3}\lambda_{(0,1)}L+ \frac{1}{2}\lambda_{(1,1)}L+\lambda_{(2,1)}L+  \frac{2}{18}\lambda_{(0,2)}L+  \frac{5}{12}\lambda_{(1,2)}L+  \frac{11}{6}\lambda_{(2,2)}L+o(L),
 \label{lemma1_coded}
\end{align}
where $\lambda_{(n,k)}$ is defined in \eqref{eq:def-L_nk}.

Since the bound in \eqref{lemma1_10} is valid for any achievable pair $(D,L)$, it is also a valid bound on the optimal normalized download cost, $D^*(\mu)$, as defined in \eqref{optimal-download-cost},  where $\mu \in[\frac{1}{3},1]$. Therefore, by taking the limit $L\rightarrow \infty$, we obtain the following bound on $D^*(\mu)$:
	\begin{align}
D^*(\mu)\geq\frac{D}{L}\geq 1+ \frac{1}{3}\lambda_{(0,1)}+ \frac{1}{2}\lambda_{(1,1)}+\lambda_{(2,1)}+  \frac{2}{18}\lambda_{(0,2)}+  \frac{5}{12}\lambda_{(1,2)}+  \frac{11}{6}\lambda_{(2,2)},
 \label{lemma1_coded2}
\end{align}	
which proves Theorem~\ref{theorem1} for $N=K=3$.

\end{example}

The following Theorem summarizes the second main result of this paper, which characterizes the information theoretically optimal download cost of the PIR problem from uncoded storage constrained databases as a function of the available storage.
\begin{theorem} \label{theorem2}
For the uncoded storage constrained PIR problem with $N$ databases, $K$ messages (of size $L$ bits each), and a per database storage constraint of $\mu KL$ bits, the information-theoretically optimal trade-off between storage and download cost is given by the lower convex hull of the following $(\mu, D^*(\mu))$ pairs, for $t=1, 2, \ldots, N$:
	\begin{align}
	\label{eq:theorem2}
	\left(\mu=\frac{t}{N},\ D^*(\mu)= \tilde{D}(t)\right),
	\end{align}
where $\tilde{D}(t)$ is defined as follows for $t\in[1:N]$:
\begin{align}
\label{eq:def-D_l}
\tilde{D}(t)\ \overset{\Delta}{=} \ \sum_{k=0}^{K-1}\ \frac{1}{t^k}.
\end{align}
\end{theorem}

\begin{figure}[t]
  \begin{center}
  \includegraphics[width=0.8\columnwidth]{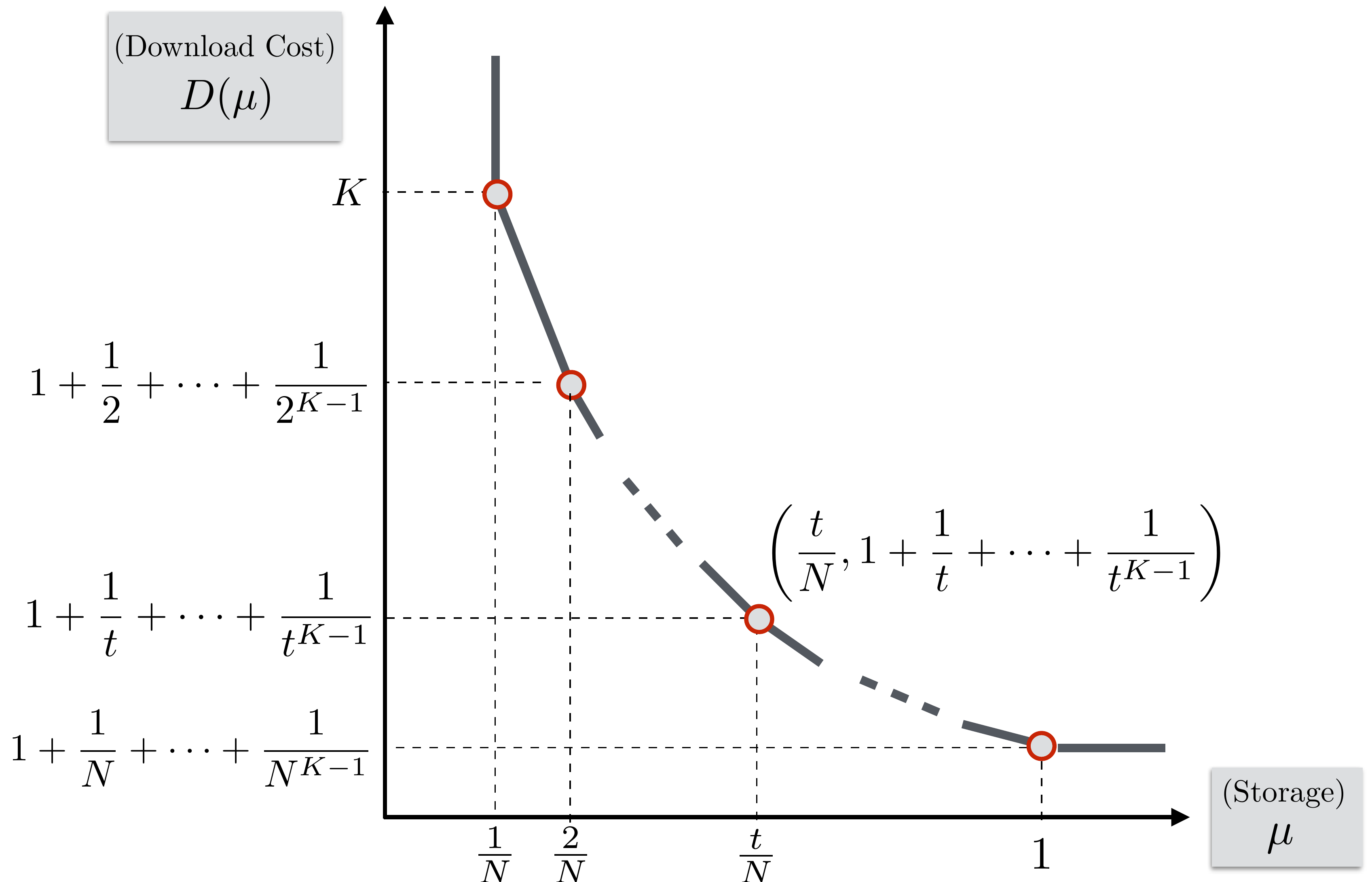}
\caption{The optimal trade-off between storage and download cost for uncoded storage constrained PIR. \label{Figtheorem}}
  \end{center}
\end{figure}

\begin{remark}\label{remark1}\normalfont
The general optimal trade-off resulting from Theorem \ref{theorem2} is illustrated in Figure \ref{Figtheorem}. The smallest value of $\mu=1/N$ corresponds to the parameter $t=1$, for which the optimal download cost is maximal and is equal to $K$, corresponding to download all the messages from the databases. The other extreme value of storage is $\mu=1$, corresponding to  $t=N$, i.e., the setting of full storage in which every database can store all the messages. For this case, the optimal download cost was characterized in \cite{SunAndJaffar1} as $(1+\frac{1}{N}+\frac{1}{N^2}+\cdots+\frac{1}{N^{K-1}})$. The PIR download cost for the storage values in between outperforms 
 memory  sharing between the two extremes, i.e., lower than the line joining between them.
\end{remark}

We briefly describe here the main elements of the achievability and the converse proof of Theorem~\ref{theorem2} through an example of $N=3$ databases and $K=3$ messages.  
The general achievable scheme and the converse proof, for any $(N, K)$ and any $\mu$, are described in Sections \ref{sec:converse}, and \ref{sec:achiev}, respectively. 

\begin{example}
\label{ex:N3-K3}
\normalfont

Figure \ref{fig-example} shows the optimal trade-off between download cost and storage for $N=K=3$.
Following Theorem \ref{theorem2}, the trade-off  has three corner points as shown in Figure \ref{fig-example} (labeled as $P_1$, $P_2$ and $P_3$). 
 We can make the following interesting observations from this figure: 
\begin{enumerate}
\item The storage point $\mu=1$ (or point $P_3$) corresponds to replicated databases, for which the optimal download cost is the same as that in \cite{SunAndJaffar1}.
\item The storage point $\mu=1/N$ (or point $P_1$) corresponds to the minimum value of storage, for which the optimal scheme is to download all messages to ensure privacy.
\item For intermediate values, i.e.,  $1/N < \mu <1$, the optimal trade-off outperforms memory sharing between $P_1$ and $P3$.
\end{enumerate}

\begin{figure}[t]
  \begin{center}
  \includegraphics[width=0.8\columnwidth]{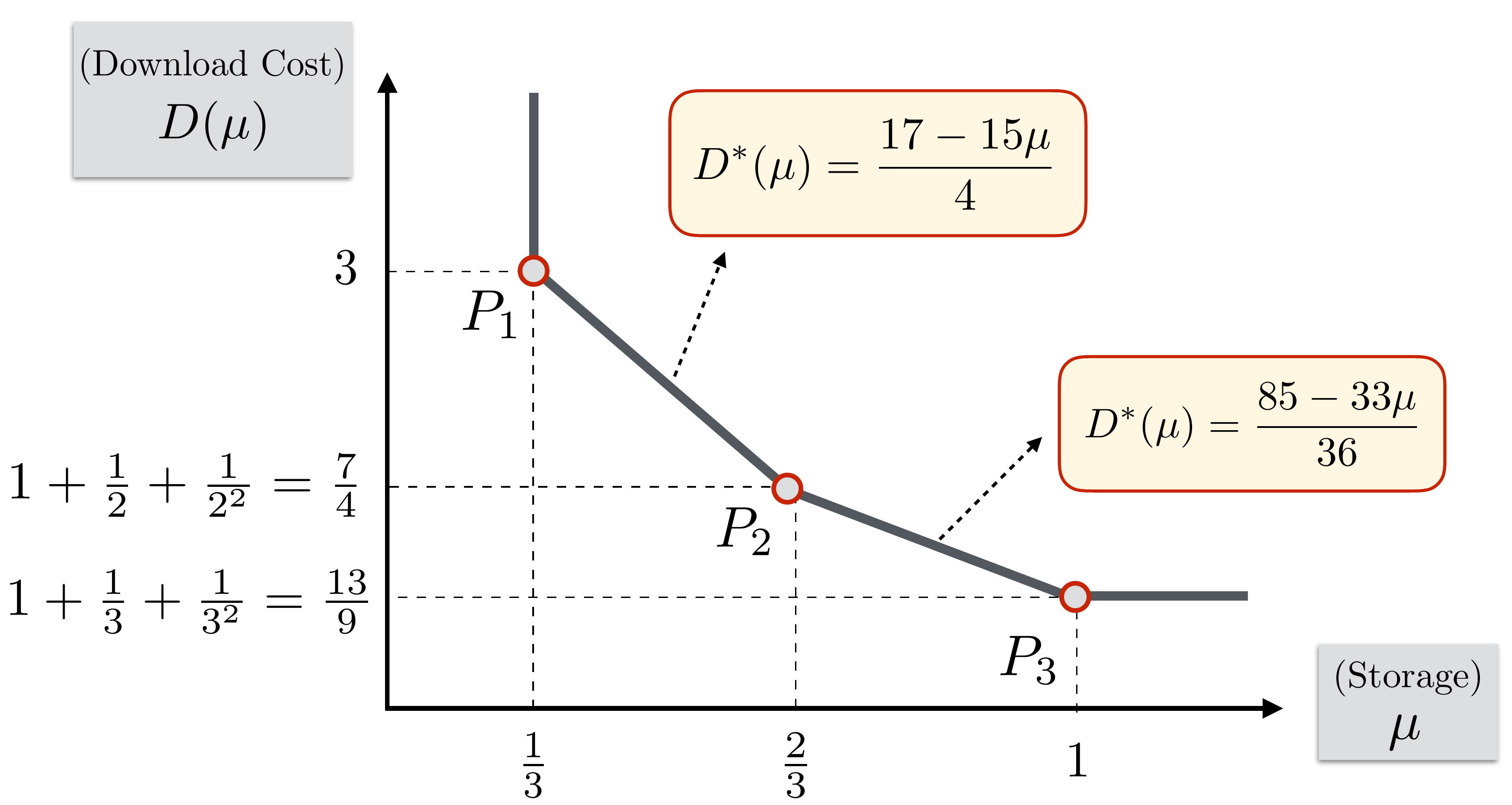}
\caption{Optimal trade-off between download and storage for $(N,K)=(3,3)$. Following Theorem \ref{theorem2}, the trade-off  has three corner points, labeled as $P_1$, $P_2$ and $P_3$.\label{fig-example}}
\vspace{-25pt}
  \end{center}
\end{figure}

\noindent \textbf{{Converse Proof}}-- 
From Figure \ref{fig-example}, it is clear that we need to need to prove the following two lower bounds on the download cost: 
\begin{align}
D^*(\mu)\geq \frac{17-15\mu}{4},\qquad D^*(\mu)\geq \frac{85- 33\mu}{36}. \label{conversegoal}
\end{align} 
To this end, we now specialize the lower bound in (\ref{lemma1_coded2}) for the case of uncoded storage placement as defined in Section~\ref{re:uncoded-assump} to write the bound in (\ref{lemma1_coded}) as follows:
	\begin{align}
D&\geq L+\frac{1}{6}\sum_{\mathcal{K}\subseteq[1:3] \atop \left|\mathcal{K}\right|=1}\sum_{k\in [1:3]\backslash \mathcal{K}}\left(\frac{1}{3}  H(W_{k}|\W_{\mathcal{K}})+\frac{1}{6}\sum_{i=1}^3  H(W_{k}|\W_{\mathcal{K}},Z_{i})+\frac{1}{3} \sum_{i=1}^3 H(W_{k}|\W_{\mathcal{K}},\Z_{[1:3]\setminus i})\right) \nonumber\\
&\hspace{5pt}+\hspace{-2pt}\frac{1}{6}\hspace{-3pt}\sum_{\mathcal{K}\subseteq[1:3] \atop \left|\mathcal{K}\right|=2}\sum_{k\in [1:3]\backslash \mathcal{K}}\hspace{-2pt}\left(\frac{2}{9}  H(W_{k}|\W_{\mathcal{K}})+\frac{5}{18}\sum_{i=1}^3  H(W_{k}|\W_{\mathcal{K}},Z_{i})\hspace{-3pt}+\hspace{-3pt}\frac{11}{9} \sum_{i=1}^3 H(W_{k}|\W_{\mathcal{K}},\Z_{[1:3]\setminus i})\right)
+o(L)	\nonumber\\
&\overset{(a)}{=} L+\frac{1}{6}\sum_{\mathcal{K}\subseteq[1:3] \atop \left|\mathcal{K}\right|=1}\sum_{k\in [1:3]\backslash \mathcal{K}}\left(\frac{1}{3}  H(W_{k})+\frac{1}{6}\sum_{i=1}^3  H(W_{k}|Z_{i})+\frac{1}{3} \sum_{i=1}^3 H(W_{k}|\Z_{[1:3]\setminus i})\right) \nonumber\\
&\hspace{10pt}+\frac{1}{6}\sum_{\mathcal{K}\subseteq[1:3] \atop \left|\mathcal{K}\right|=2}\sum_{k\in [1:3]\backslash \mathcal{K}}\left(\frac{2}{9}  H(W_{k})+\frac{5}{18}\sum_{i=1}^3  H(W_{k}|Z_{i})\hspace{-3pt}+\hspace{-3pt}\frac{11}{9} \sum_{i=1}^3 H(W_{k}|\Z_{[1:3]\setminus i})\right)
+o(L)\nonumber\\
&= L+\frac{1}{6}\sum_{k=1}^3 \left[\sum_{\mathcal{K}\subseteq[1:3]\backslash k \atop \left|\mathcal{K}\right|=1}\left(\frac{1}{3}  H(W_{k})+\frac{1}{6}\sum_{i=1}^3  H(W_{k}|Z_{i})+\frac{1}{3} \sum_{i=1}^3 H(W_{k}|\Z_{[1:3]\setminus i})\right)\right. \nonumber\\
&\hspace{10pt}+\left.\sum_{\mathcal{K}\subseteq[1:3]\backslash k \atop \left|\mathcal{K}\right|=2}\left(\frac{2}{9}  H(W_{k})+\frac{5}{18}\sum_{i=1}^3  H(W_{k}|Z_{i})\hspace{-3pt}+\hspace{-3pt}\frac{11}{9} \sum_{i=1}^3 H(W_{k}|\Z_{[1:3]\setminus i})\right)\right]
+o(L)\nonumber\\
&= L+\frac{4}{27}\sum_{k=1}^3  H(W_{k})+\frac{11}{108}\sum_{i=1}^3\sum_{k=1}^3  H(W_{k}|Z_{i})+\frac{17}{54} \sum_{i=1}^3 \sum_{k=1}^3 H(W_{k}|\Z_{[1:3]\setminus i})+o(L),
 \label{lemma1_uncoded}
\end{align}
where step $(a)$ follows since the messages are i.i.d. and using the uncoded storage content given in \eqref{eq:storage-content2}.
Next, we use the representation of the message $W_k$ for uncoded storage given in \eqref{eq:def_W_k} to write the lower bound in \eqref{lemma1_uncoded} in terms of $x_{\ell}$ defined in \eqref{eq:def-x_ell} as follows,
	\begin{align}
D& \overset{(a)}{\geq} L+\frac{4}{27}\sum_{k=1}^3\sum_{\substack{\mathcal{S}\subseteq [1:3] \\ \vert \mathcal{S}\vert \geq 1}}  \vert W_{k,\mathcal{S}}\vert L+\frac{11}{108}\sum_{i=1}^3\sum_{k=1}^3 \sum_{\substack{\mathcal{S}\subseteq [1:3]\setminus i \\ \vert \mathcal{S}\vert \geq 1}}  \vert W_{k,\mathcal{S}}\vert L+\frac{17}{54} \sum_{i=1}^3 \sum_{k=1}^3 \vert W_{k,\{i\}}\vert L+o(L)\nonumber\\
&{=} L+\frac{2}{3}\sum_{k=1}^3\sum_{\substack{\mathcal{S}\subseteq [1:3] \\ \vert \mathcal{S}\vert = 1}}  \vert W_{k,\mathcal{S}}\vert L +\frac{1}{4}\sum_{k=1}^3 \sum_{\substack{\mathcal{S}\subseteq [1:3]\\ \vert \mathcal{S}\vert =2}}  \vert W_{k,\mathcal{S}}\vert L+\frac{4}{27}\sum_{k=1}^3 \sum_{\substack{\mathcal{S}\subseteq [1:3]\\ \vert \mathcal{S}\vert =3}}  \vert W_{k,\mathcal{S}}\vert L+o(L)\nonumber\\
&\overset{(b)}{=} L+6x_1 L +\frac{9}{4}x_2 L+\frac{4}{9}x_3 L+o(L),
 \label{lemma1_12}
\end{align}	
where $(a)$ follows from \eqref{eq:def_W_k}; and $(b)$ follows from the definition of $x_{\ell}$ in \eqref{eq:def-x_ell}.

Since the bound in \eqref{lemma1_12} is valid for any achievable pair $(D,L)$, it is also a valid bound on the optimal normalized download cost, $D^*(\mu)$, as defined in \eqref{optimal-download-cost},  where $\mu \in[\frac{1}{3},1]$. Therefore, by dividing \eqref{lemma1_12} and taking the limit $L\rightarrow \infty$, we obtain the following bound on $D^*(\mu)$:
	\begin{align}
D^*(\mu)\geq 1+6x_1 +\frac{9}{4}x_2+\frac{4}{9}x_3.
 \label{lemma1_13}
\end{align}	
Moreover, the message size, and the storage constraints for uncoded storage placement for this example $N=K =3$ follow from \eqref{eq:size-const}, and \eqref{eq:storage-const}, respectively. Hence, we obtain the following constraints:
\begin{align}
&3x_1 +3x_2 + x_3 =1, \label{eq:ex-datasize} \\
 &3x_1 +6x_2 + 3x_3 \leq 3\mu.\label{eq:ex-storage} 
\end{align}

We get the first lower bound on $D^*(\mu)$, given in \eqref{conversegoal},  in two steps. First, we use the message size constraint in \eqref{eq:ex-datasize} to remove $x_1$ from \eqref{lemma1_13} and \eqref{eq:ex-datasize} and get two bounds,
\begin{align}
D^*(\mu)&\geq 3 -\frac{15}{4}x_2-\frac{14}{9}x_3, \label{lemma1_14-1} \\
 3\mu&\geq 1 +3x_2 +2x_3.\label{eq:ex-storage-1} 
\end{align}
Then, we use \eqref{eq:ex-storage-1} to bound $x_2$ in \eqref{lemma1_14-1} and get
\begin{align}
\label{ex:bound1}
D^*(\mu)&\geq \frac{17-15\mu}{4} +\frac{17}{18}x_3 \overset{(a)}{\geq} \frac{17-15\mu}{4},
\end{align}
where $(a)$ follows since $x_3\geq 0$.

We get the second lower bound on $D^*(\mu)$, given in \eqref{conversegoal}, by eliminating the variables $(x_2,x_3)$. First, we use the message size constraint in \eqref{eq:ex-datasize} to remove $x_2$ from \eqref{lemma1_13} and \eqref{eq:ex-datasize} and get two bounds,
\begin{align}
D^*(\mu)&\geq \frac{7}{4} +\frac{15}{4}x_1-\frac{11}{36}x_3, \label{lemma1_14-2} \\
 3\mu&\geq 2 -3x_1 +x_3.\label{eq:ex-storage-2} 
\end{align}
Then, we use \eqref{eq:ex-storage-2} to bound $x_3$ in \eqref{lemma1_14-2} and get
\begin{align}
\label{ex:bound2}
D^*(\mu)&\geq \frac{85-33\mu}{36} +\frac{17}{6}x_1 \overset{(a)}{\geq} \frac{85-33\mu}{36},
\end{align}
where $(a)$ follows since $x_1\geq 0$.

Proving the two bounds in \eqref{ex:bound1} and \eqref{ex:bound2} completes the converse proof for $N=3$, and $K=3$.

\noindent \textbf{{Achievable Scheme}}-- 
The optimal trade-off is achieved by memory sharing  between different PIR schemes (see Claim~\ref{claim1}), which are designed for three values of storage $\mu\in\{1/3,2/3,1\}$. Consider without loss of generality, that the user wants to privately retrieve message $W_1$.

\noindent $\bullet$\hspace{5pt} \textbf{Case $\mathbf{P_1}$ ($\mathbf{t=1\textbf{ or }\mu =1/3}$):}

\noindent\underline{Storage Placement:}
For this point, we assume that each message is of size $L= 1^{3}\times \binom{3}{1}= 3$ bits. Hence, the storage constraint for each database is $\mu KL =3$ bits.  In particular, we split each message into $\binom{3}{1}=3$ sub-messages and label each by a unique subset of $[1:3]$ of size $t=1$, i.e., $W_k= \{W_{k,\{1\}},W_{k,\{2\}},W_{k,\{3\}}\}$ for $k\in[1:3]$, and each sub-message is of size $1^{3}=1$ bits.
Subsequently,  DB$_n$ stores those sub-messages (of each message) whose index contains $n$. For instance, DB$_1$ stores $\{W_{1,\{1\}},W_{2,\{1\}},W_{3,\{1\}}\}$. Hence, the  total storage required per database is $3$ bits, which satisfies the storage constraint.

\noindent\underline{PIR Scheme:}
The PIR scheme is trivial for this storage point, where in order to maintain privacy all the messages should be downloaded from the databases for any message request.  Hence, the download cost for this scheme  is given as $D(\mu =\frac{1}{3})= \frac{9}{3}=3$, and point $P_1$ is achieved.

\noindent $\bullet$\hspace{5pt} \textbf{Case $\mathbf{P_2}$ ($\mathbf{t=2\textbf{ or }\mu =2/3}$):}

\noindent\underline{Storage Placement:}
The storage placement here is inspired by the storage placement strategy for caching systems in \cite{Maddah}.
For this storage point, we assume that each message is of size $L= 2^{3}\times \binom{3}{2}= 24$ bits. 
Hence, the storage constraint for each database is $\mu KL =48$ bits.
In particular, we split each message into $\binom{3}{2}=3$ sub-messages and label each by a unique subset of $[1:3]$ of size $t=2$, i.e., $W_k= \{W_{k,\{1,2\}},W_{k,\{1,3\}},W_{k,\{2,3\}}\}$ for $k\in[1:3]$. Each sub-message is of size $2^{3}=8$ bits, i.e., $W_{k,\mathcal{S}} = (w^{1}_{k,\mathcal{S}},\ldots, w^{8}_{k,\mathcal{S}})$ for $k\in[1:3]$ and $\mathcal{S}$ is a subset of $[1:3]$ of size $2$. We can write the bits representation of the $3$ messages as:
\begin{align}
&W_1=\{(w^{1}_{1,\{1,2\}},\ldots, w^{8}_{1,\{1,2\}}),(w^{1}_{1,\{1,3\}},\ldots, w^{8}_{1,\{1,3\}}),(w^{1}_{1,\{2,3\}},\ldots, w^{8}_{1,\{2,3\}})\},\nonumber\\
&W_2=\{(w^{1}_{2,\{1,2\}},\ldots, w^{8}_{2,\{1,2\}}),(w^{1}_{2,\{1,3\}},\ldots, w^{8}_{2,\{1,3\}}),(w^{1}_{2,\{2,3\}},\ldots, w^{8}_{2,\{2,3\}})\},\nonumber\\
&W_3=\{(w^{1}_{3,\{1,2\}},\ldots, w^{8}_{3,\{1,2\}}),(w^{1}_{3,\{1,3\}},\ldots, w^{8}_{3,\{1,3\}}),(w^{1}_{3,\{2,3\}},\ldots, w^{8}_{3,\{2,3\}})\}.
\end{align}	
Subsequently,  DB$_n$ stores those sub-messages (of each message) whose index contains $n$. For instance, DB$_1$ stores $\{W_{1,\{1,2\}},W_{1,\{1,3\}},W_{2,\{1,2\}},W_{2,\{1,3\}},W_{3,\{1,2\}},W_{3,\{1,3\}}\}$. Hence, the  total storage required per database is $6\times 8 = 48$ bits, which satisfies the storage constraint.

\noindent\underline{PIR Scheme:}
 For every $k\in[1:3]$ and $\mathcal{S}\subset [1:3]$ where $\vert\mathcal{S} \vert=2$, let $(u^1_{k,\mathcal{S}},u^2_{k,\mathcal{S}},\ldots,u^8_{k,\mathcal{S}})$ represents a random permutation of the bits of the sub-message $W_{k,\mathcal{S}}$ according to the permutation $\bdelta^{k,\mathcal{S}}: (1,2,\ldots,8)\rightarrow (\delta^{k,\mathcal{S}}_1,\delta^{k,\mathcal{S}}_2,\ldots,\delta^{k,\mathcal{S}}_8)$, i.e., $u^i_{k,\mathcal{S}}= w^{(\delta^{k,\mathcal{S}}_i)}_{k,\mathcal{S}}$ for $i\in[1:8]$.
 The storage constraint PIR scheme is shown  in Table~\ref{Table}.
It works in $3$ blocks, where in every block, only 2 databases are involved given by DB$_i$ and $i\in\mathcal{S}$ for every $\mathcal{S}\subset [1:3]$ where $\vert\mathcal{S} \vert=2$. We can label each block by the set $\mathcal{S}$ including the databases involved.
In the block $\mathcal{S}$, the user only downloads bits from the sub-messages $W_{k,\mathcal{S}}$ for $k\in[1:3]$. Each block is divided into three stages described next for the block $\mathcal{S}=\{1,2\}$. Similar steps can be followed for the two other blocks $\{1,3\}$ and $\{2,3\}$.

	\begin{table} [!h]
		%\hrule \vspace{4pt}
		\renewcommand{\arraystretch} {1.3}
		\caption{Storage Constrained PIR: $(N,K)=(3,3), \mu=\frac{t}{N}=\frac{2}{3}$}
		\centering 
		\begin{tabular}{|c|c|c|c|}
			\specialrule{.15em}{0em}{0em} 
			Blocks &DB$_{1}$ &DB$_{2}$ & DB$_{{3}}$ \\
			\specialrule{.15em}{0em}{0em} 
			\multirow{5}{*}{\shortstack{Block $1$ \\ \\ $\mathcal{S}= \{1,2\}$ }} &$u^1_{1,\{1,2\}}, u^1_{2,\{1,2\}}, u^1_{3,\{1,2\}}$   &$u^2_{1,\{1,2\}}, u^2_{2,\{1,2\}}, u^2_{3,\{1,2\}}$  &   -\\
			\cline{2-4}
			&$u^3_{1,\{1,2\}}+u^2_{2,\{1,2\}}$\  & $u^5_{1,\{1,2\}}+u^1_{2,\{1,2\}}$\  & - \\
			&$u^4_{1,\{1,2\}}+u^2_{3,\{1,2\}}$\  & $u^6_{1,\{1,2\}}+u^1_{3,\{1,2\}}$\  & - \\
			&$u^3_{2,\{1,2\}}+u^3_{3,\{1,2\}}$ & $u^4_{2,\{1,2\}}+u^4_{3,\{1,2\}}$ &  -\\
			\cline{2-4}
			&$u^7_{1,\{1,2\}}\hspace{-1pt}+\hspace{-1pt}u^4_{2,\{1,2\}}\hspace{-1pt}+\hspace{-1pt}u^4_{3,\{1,2\}}$ &$u^8_{1,\{1,2\}}\hspace{-1pt}+\hspace{-1pt}u^3_{2,\{1,2\}}\hspace{-1pt}+\hspace{-1pt}u^3_{3,\{1,2\}}$ &   \\	
			\specialrule{.15em}{0em}{0em} 
\multirow{5}{*}{\shortstack{Block $2$ \\ \\ $\mathcal{S}= \{1,3\}$ }} & $u^1_{1,\{1,3\}},u^1_{2,\{1,3\}},u^1_{3,\{1,3\}}$ & -  & $u^2_{1,\{1,3\}}, u^2_{2,\{1,3\}}, u^2_{3,\{1,3\}}$  \\
\cline{2-4}
&$u^3_{1,\{1,3\}}+u^2_{2,\{1,3\}}$ & - &  $u^5_{1,\{1,3\}}+u^1_{2,\{1,3\}}$\\
		&$u^4_{1,\{1,3\}}+u^2_{3,\{1,3\}}$ &- &  $u^6_{1,\{1,3\}}+u^1_{3,\{1,3\}}$\\
			&$u^3_{2,\{1,3\}}+u^3_{3,\{1,3\}}$ & - &  $u^4_{2,\{1,3\}}+u^4_{3,\{1,3\}}$\\
		 \cline{2-4}
		&$u^7_{1,\{1,3\}}\hspace{-1pt}+\hspace{-1pt}u^4_{2,\{1,3\}}\hspace{-1pt}+\hspace{-1pt}u^4_{3,\{1,3\}}$ & - &$u^8_{1,\{1,3\}}\hspace{-1pt}+\hspace{-1pt}u^3_{2,\{1,3\}}+u^3_{3,\{1,3\}}$\\	
			\specialrule{.15em}{0em}{0em} 
			\multirow{5}{*}{\shortstack{Block $3$ \\ \\ $\mathcal{S}= \{2,3\}$ }}			& -&  $u^1_{1,\{2,3\}}, u^1_{2,\{2,3\}},u^1_{3,\{2,3\}}$ & $u^2_{1,\{2,3\}}, u^2_{2,\{2,3\}}, u^2_{3,\{2,3\}}$ \\
\cline{2-4}
		& - & $u^3_{1,\{2,3\}}+u^2_{2,\{2,3\}}$ &  $u^5_{1,\{2,3\}}+u^1_{2,\{2,3\}}$\\
		& - & $u^4_{1,\{2,3\}}+u^2_{3,\{2,3\}}$ &  $u^6_{1,\{2,3\}}+u^1_{3,\{2,3\}}$\\
		& - & $u^3_{2,\{2,3\}}+u^3_{3,\{2,3\}}$ &  $u^4_{2,\{2,3\}}+u^4_{3,\{2,3\}}$\\
		 \cline{2-4}
			& - &$u^7_{1,\{2,3\}}\hspace{-1pt}+\hspace{-1pt}u^4_{2,\{2,3\}}\hspace{-1pt}+\hspace{-1pt}u^4_{3,\{2,3\}}$ &  $u^8_{1,\{2,3\}}\hspace{-1pt}+\hspace{-1pt}u^3_{2,\{2,3\}}\hspace{-1pt}+\hspace{-1pt}u^3_{3,\{2,3\}}$\\		
			\specialrule{.15em}{0em}{0em} 
		\end{tabular}
		\label{Table}
	\end{table}

\noindent \textbf{Stage $1$}: In first stage of the block $\{1,2\}$, single bits from the messages are downloaded from DB$_1$ and DB$_2$. From DB$_1$, the user downloads $1$ desired bits $u^1_{1,\{1,2\}}$, then $2$ undesired bits $\{u^1_{2,\{1,2\}},\ u^1_{2,\{1,3\}}\}$ to confuse the database on which message is requested by maintaining message symmetry. Similarly, the bits $\{u^2_{1,\{1,2\}},u^2_{2,\{1,2\}},\ u^2_{3,\{1,2\}}\}$ are downloaded from DB$_2$.

\noindent \textbf{Stage $2$}: In the second stage of the block $\{1,2\}$, the undesired bits downloaded from the Stage $1$ can be paired with new desired bits.
It is important to notice that the  undesired bits cannot be repeatedly acquired from the same database, otherwise the symmetry in the queries across messages, and hence the privacy, will be violated.
Therefore, the number of new desired bits to be downloaded from DB$_1$ in this stage is limited to the number of undesired bits downloaded in Stage 1 from DB$_2$, and vice versa.
From DB$_1$, the user can download 2 coded desired symbols 
$\{u^3_{1,\{1,2\}}+u^2_{2,\{1,2\}}$, $u^4_{1,\{1,2\}}+u^2_{3,\{1,2\}}\}$. In order to confuse the database, the user maintains the symmetry in the downloaded message by downloading one more undesired coded symbols  $u^3_{2,\{1,2\}}+u^3_{3,\{1,2\}}$. Similarly, the following coded bits are downloaded from DB$_2$: $\{u^5_{1,\{1,2\}}+u^1_{2,\{1,2\}}$, $u^6_{1,\{1,2\}}+u^1_{3,\{1,2\}}$, $u^4_{2,\{1,2\}}+u^4_{3,\{1,2\}}\}$.

%\vspace{-20pt}

\noindent \textbf{Stage $3$}: In the third stage, the second order undesired symbols downloaded from Stage 2 are used to be paired up with new useful bits of the desired message and download third order coded symbols (triples) from the databases. 
From DB$_1$, since $u^4_{2,\{1,2\}}+u^4_{3,\{1,2\}}$ is an undesired side information bit at the user downloaded from DB$_2$ in Stage $2$, it can be paired with a new desired bit of $W_1$, $u^7_{1,\{1,2\}}$, i.e., download $u^7_{1,\{1,2\}}+u^4_{2,\{1,2\}}+u^4_{3,\{1,2\}}$. Similarly, the user downloads the $u^8_{1,\{1,2\}}+u^3_{2,\{1,2\}}+u^3_{3,\{1,2\}}$ from DB$_2$.

\begin{remark}
Each blocks of the storage constrained PIR scheme for the point $\text{P}_2$, i.e., ($N=3$, $K=3$, $t=2$) can be viewed as doing the original PIR scheme proposed in \cite{SunAndJaffar1} for $t=2$ databases and $N=3$ messages. This is enabled by our storage placement scheme inspired by the the work on coded caching \cite{Maddah}.
\end{remark}

We can readily verify that from the $14$ bits downloaded  in every block $\mathcal{S}$, the user is able to correctly retrieve the $8$ bits of the corresponding desired sub-message $W_{1,\mathcal{S}}$.
In total, the user downloads $3\times 14= 42$ coded bits in the three blocks in order to privately retrieve the $24$ bits of the desired message $W_1$.
 Hence, we can achieve a download cost for this scheme $D(\mu =\frac{2}{3})= \frac{42}{24}=7/4$, and point $P_2$ is achieved.

\noindent $\bullet$\hspace{5pt} \textbf{Case $\mathbf{P_3}$ ($\mathbf{t=3\textbf{ or }\mu =1}$):}

\noindent\underline{Storage Placement:}
This storage point is the replicated databases case considered in \cite{SunAndJaffar1}. We describe it here for the sake of completeness.
The storage placement is trivial in this case, where all the databases can store the 3 messages completely.
Each message is assumed to be of size $L= 3^{3}\times \binom{3}{3}= 27$ bits, e.g., $W_1=(w^1_{1,\{1,2,3\}},w^2_{1,\{1,2,3\}},\ldots,w^{27}_{1,\{1,2,3\}})$.

\noindent\underline{PIR Scheme:}
 For every $k\in[1:3]$, let $(u^1_{k,\{1,2,3\}},u^2_{k,\{1,2,3\}},\ldots,u^{27}_{k,\{1,2,3\}})$ represents a random permutation of the bits of the sub-message $W_{k,\{1,2,3\}}$ according to the permutation $\bdelta^{k,\{1,2,3\}}: (1,2,\ldots,27)\rightarrow (\delta^{k,\{1,2,3\}}_1,\delta^{k,\{1,2,3\}}_2,\ldots,\delta^{k,\{1,2,3\}}_{27})$, i.e., $u^i_{k,\{1,2,3\}}= w^{(\delta^{k,\{1,2,3\}}_i)}_{k,\{1,2,3\}}$ for $i\in[1:27]$.
The PIR scheme works in three stages as shown in Table~\ref{Table2}.

	\begin{table} [!h]
		%\hrule \vspace{4pt}
		\renewcommand{\arraystretch} {1.3}
		\caption{Storage Constrained PIR: $(N,K)=(3,3), \mu=\frac{t}{N}=1$}
		\centering 
		\begin{tabular}{|c|c|c|}
		\specialrule{.15em}{0em}{0em} 
			DB$_{1}$ &DB$_{2}$ & DB$_{3}$ \\
		\specialrule{.15em}{0em}{0em} 
			$u^1_{1,\{1,2,3\}},u^1_{2,\{1,2,3\}},u^1_{3,\{1,2,3\}}$   &$u^2_{1,\{1,2,3\}},u^2_{2,\{1,2,3\}},u^2_{3,\{1,2,3\}}$  &$u^3_{1,\{1,2,3\}},u^3_{2,\{1,2,3\}},u^3_{3,\{1,2,3\}}$   \\
			\hline
			$u^4_{1,\{1,2,3\}}+u^2_{2,\{1,2,3\}}$\  & $u^8_{1,\{1,2,3\}}+u^1_{2,\{1,2,3\}}$ &$u^{12}_{1,\{1,2,3\}}+u^1_{2,\{1,2,3\}}$\\
			$u^5_{1,\{1,2,3\}}+u^2_{3,\{1,2,3\}}$\  & $u^9_{1,\{1,2,3\}}+u^1_{3,\{1,2,3\}}$  &$u^{13}_{1,\{1,2,3\}}+u^1_{3,\{1,2,3\}}$\\
			$u^4_{2,\{1,2,3\}}+u^4_{3,\{1,2,3\}}$\  & $u^{6}_{2,\{1,2,3\}}+u^6_{3,\{1,2,3\}}$  &$u^{8}_{2,\{1,2,3\}}+u^8_{3,\{1,2,3\}}$\\
			$u^6_{1,\{1,2,3\}}+u^3_{2,\{1,2,3\}}$\  & $u^{10}_{1,\{1,2,3\}}+u^3_{2,\{1,2,3\}}$  &$u^{14}_{1,\{1,2,3\}}+u^2_{2,\{1,2,3\}}$\\
			$u^7_{1,\{1,2,3\}}+u^3_{3,\{1,2,3\}}$\  & $u^{11}_{1,\{1,2,3\}}+u^3_{3,\{1,2,3\}}$ &$u^{15}_{1,\{1,2,3\}}+u^2_{3,\{1,2,3\}}$\\
			$u^5_{2,\{1,2,3\}}+u^5_{3,\{1,2,3\}}$\  & $u^{7}_{2,\{1,2,3\}}+u^7_{3,\{1,2,3\}}$  &$u^{9}_{2,\{1,2,3\}}+u^9_{3,\{1,2,3\}}$\\
			\hline
			$u^{16}_{1,\{1,2,3\}}+u^{6}_{2,\{1,2,3\}}+u^6_{3,\{1,2,3\}}$ &$u^{20}_{1,\{1,2,3\}}+u^{4}_{2,\{1,2,3\}}+u^4_{3,\{1,2,3\}}$&  $u^{24}_{1,\{1,2,3\}}+u^{4}_{2,\{1,2,3\}}+u^4_{3,\{1,2,3\}}$\\
			$u^{17}_{1,\{1,2,3\}}+u^{7}_{2,\{1,2,3\}}+u^7_{3,\{1,2,3\}}$ &$u^{21}_{1,\{1,2,3\}}+u^{5}_{2,\{1,2,3\}}+u^5_{3,\{1,2,3\}}$&  $u^{25}_{1,\{1,2,3\}}+u^{5}_{2,\{1,2,3\}}+u^5_{3,\{1,2,3\}}$\\
			$u^{18}_{1,\{1,2,3\}}+u^{8}_{2,\{1,2,3\}}+u^8_{3,\{1,2,3\}}$ &$u^{22}_{1,\{1,2,3\}}+u^{8}_{2,\{1,2,3\}}+u^8_{3,\{1,2,3\}}$&  $u^{26}_{1,\{1,2,3\}}+u^{6}_{2,\{1,2,3\}}+u^6_{3,\{1,2,3\}}$\\
			$u^{19}_{1,\{1,2,3\}}+u^{9}_{2,\{1,2,3\}}+u^9_{3,\{1,2,3\}}$ &$u^{23}_{1,\{1,2,3\}}+u^{9}_{2,\{1,2,3\}}+u^9_{3,\{1,2,3\}}$&  $u^{27}_{1,\{1,2,3\}}+u^{7}_{2,\{1,2,3\}}+u^7_{3,\{1,2,3\}}$\\			
		\specialrule{.15em}{0em}{0em} 
		\end{tabular}
		\label{Table2}
	\end{table}

\noindent \textbf{Stage $1$}: In first stage, single bits from the messages are downloaded. From DB$_1$, the user downloads 1 desired bits $u^1_{1,\{1,2,3\}}$, then 2 undesired bits $\{u^1_{2,\{1,2,3\}}$, $u^1_{3,\{1,2,3\}}\}$ to confuse the database on which message is requested by maintaining message symmetry. Similarly, $\{u^2_{1,\{1,2,3\}}$ ,$u^2_{2,\{1,2,3\}}$, $u^2_{3,\{1,2,3\}}\}$ and $\{u^3_{1,\{1,2,3\}}$, $u^3_{2,\{1,2,3\}}$, $u^3_{3,\{1,2,3\}}\}$ are downloaded from DB$_2$ and DB$_3$ respectively.

\noindent \textbf{Stage $2$}: In the second stage, the undesired bits downloaded from the first stage can be paired with new desired bits.
For instance, the user can download 4 desired bits from DB$_1$ by requesting 4 second order coded symbols (pairs) paired up with the undesired bits downloaded from DB$_2$ and DB$_3$ in Stage 1, as follows:
$\{u^4_{1,\{1,2,3\}}+u^2_{2,\{1,2,3\}}$, $u^5_{1,\{1,2,3\}}+u^2_{3,\{1,2,3\}}$, $u^6_{1,\{1,2,3\}}+u^3_{2,\{1,2,3\}}$, $u^7_{1,\{1,2,3\}}+u^3_{3,\{1,2,3\}}\}$. In order to confuse the database the user maintains the symmetry in the downloaded message by downloading 2 undesired coded symbols  $\{u^4_{2,\{1,2,3\}}+u^4_{3,\{1,2,3\}}$, $u^5_{2,\{1,2,3\}}+u^5_{3,\{1,2,3\}}\}$.

%\vspace{-20pt}

\noindent \textbf{Stage $3$}: In the third stage, the  undesired coded pairs downloaded from Stage 2 are used to be paired up with new useful bits of the desired message and download  coded triples from the databases. 
For instance, the user can download 4 desired bits from DB$_1$, $\{u^{16}_{1,\{1,2,3\}}$, $u^{17}_{1,\{1,2,3\}}$, $u^{18}_{1,\{1,2,3\}}$, $u^{19}_{1,\{1,2,3\}}\}$, by requesting 4  coded triples paired up with the undesired coded bits downloaded from DB$_2$ and DB$_3$ in Stage 2, as shown in Table~\ref{Table2}.

We can readily verify that from the $39$ bits downloaded  from all three databases, the user is able to correctly retrieve the $27$ bits of message $W_1$. Hence, we can achieve a download cost for this scheme $D(\mu =1)= \frac{39}{27}=13/9$, and point $P_3$ is achieved.

 Finally, the intermediate values of $\mu$, between the points $P_1$, $P_2$, and $P_3$,  can be achieved by memory-sharing (see Claim~\ref{claim1}), showing that the lower convex hull given in Figure~\ref{fig-example} is achievable, and that the scheme is information-theoretically optimal for $N=3$, and $K=3$.
\end{example}

\section{Proof of Theorem \ref{theorem1}: General Lower Bound on $D^*(\mu)$}
\label{sec:LB}

We start by proving the following Lemma, which provides an information theoretic bound useful in many steps of the general converse proof. 

\begin{lemma}\label{lem1}
For any $\mathcal{N}\subseteq [1:N]$, $\mathcal{K}\subseteq [1:K]$, $i\in [1:N]$, and $j\in [1:K]$ we can write the following lower bound:
\begin{align}
I(\W_{[1:K]\backslash\mathcal{K}};Q_i^{[j]},A_i^{[j]}|\W_{\mathcal{K}},\Q_{\mathcal{N}}^{[j]},\A_{\mathcal{N}}^{[j]})\geq I(\W_{[1:K]\backslash\mathcal{K}};Q_i^{[j]},A_i^{[j]}|\W_{\mathcal{K}},\Z_{\mathcal{N}}).
\end{align}
\end{lemma}

\begin{remark}
Lemma \ref{lem1} lower bounds the mutual information $I(\W_{[1:K]\backslash\mathcal{K}};Q_i^{[j]},A_i^{[j]}|\W_{\mathcal{K}},\Q_{\mathcal{N}}^{[j]},\A_{\mathcal{N}}^{[j]})$ by replacing the queries and answers $\Q_{\mathcal{N}}^{[j]},\A_{\mathcal{N}}^{[j]}$ in the conditioning, with the storage contents of the corresponding subset of databases $\Z_{\mathcal{N}}$. This Lemma is repeatedly used in our converse proof for the download cost. 
\end{remark}

\begin{proof}
\begin{align}
I(\W_{[1:K]\backslash\mathcal{K}};&Q_i^{[j]},A_i^{[j]}|\W_{\mathcal{K}},\Q_{\mathcal{N}}^{[j]},\A_{\mathcal{N}}^{[j]})\nonumber\\
&=H(Q_i^{[j]},A_i^{[j]}|\W_{\mathcal{K}},\Q_{\mathcal{N}}^{[j]},\A_{\mathcal{N}}^{[j]})-H(Q_i^{[j]},A_i^{[j]}|\W_{[1:K]},\Q_{\mathcal{N}}^{[j]},\A_{\mathcal{N}}^{[j]})\nonumber\\
&\overset{(a)}{=}H(Q_i^{[j]},A_i^{[j]}|\W_{\mathcal{K}},\Q_{\mathcal{N}}^{[j]},\A_{\mathcal{N}}^{[j]})-H(Q_i^{[j]},A_i^{[j]}|\W_{[1:K]},\Q_{\mathcal{N}}^{[j]},\A_{\mathcal{N}}^{[j]},\Z_{\mathcal{N}})\nonumber\\
&\overset{(b)}{\geq} H(Q_i^{[j]},A_i^{[j]}|\W_{\mathcal{K}},\Q_{\mathcal{N}}^{[j]},\A_{\mathcal{N}}^{[j]},\Z_{\mathcal{N}})-H(Q_i^{[j]},A_i^{[j]}|\W_{[1:K]},\Q_{\mathcal{N}}^{[j]},\A_{\mathcal{N}}^{[j]},\Z_{\mathcal{N}})\nonumber\\
&\overset{(c)}{=} H(Q_i^{[j]},A_i^{[j]}|\W_{\mathcal{K}},\Q_{\mathcal{N}}^{[j]},\Z_{\mathcal{N}})-H(Q_i^{[j]},A_i^{[j]}|\W_{[1:K]},\Q_{\mathcal{N}}^{[j]},\Z_{\mathcal{N}})\nonumber\\
&=I(Q_i^{[j]},A_i^{[j]};\W_{[1:K]\backslash\mathcal{K}}|\W_{\mathcal{K}},\Q_{\mathcal{N}}^{[j]},\Z_{\mathcal{N}})\nonumber\\
&\overset{(d)}{=}I(\Q_{\mathcal{N}}^{[j]},Q_i^{[j]},A_i^{[j]};\W_{[1:K]\backslash\mathcal{K}}|\W_{\mathcal{K}},\Z_{\mathcal{N}})\nonumber\\
&\geq I(\W_{[1:K]\backslash\mathcal{K}};Q_i^{[j]},A_i^{[j]}|\W_{\mathcal{K}},\Z_{\mathcal{N}}),
\end{align}
where 
$(a)$ follows from the fact that the random variables $\Z_{\mathcal{N}}$ are functions of all the messages $\W_{[1:K]}$;
$(b)$ follows since conditioning reduces entropy; $(c)$ follows  since the answers $\A_{\mathcal{N}}^{[j]}$ are functions of the storage random variables $\Z_{\mathcal{N}}$ and the queries $\Q_{\mathcal{N}}^{[j]}$;
and $(d)$ follows from fact that queries $\Q_{\mathcal{N}}^{[j]}$ are independent from the messages  stored at the databases.
\end{proof}

The following Lemma gives a lower bound on the number of downloaded bits $D$ in terms of the summation of a mutual information term of the form $I(\W_{{[1:K]}\backslash\mathcal{K}};Q^{[j]}_i,A^{[j]}_i|\W_{\mathcal{K}},\Z_{\mathcal{N}})$. Notice that this is the same term appears in right side of the bound in Lemma~\ref{lem1}.

\begin{lemma}\label{lem2}
The download cost $D$ of the storage-constrained PIR is lower bounded as follows:
\begin{align}
D\geq L+\sum_{n=0}^{N-1}T(n,1)+o(L),
\end{align} 
 where $T(n,k)$ for $n\in[0:N]$ and $k\in[0:K]$ is defined as follows:
\begin{align}
\label{eq:def-T_nk}
T(n,k)\overset{\Delta}{=}\frac{1}{NK\binom{K-1}{k}\binom{N-1}{n}}\sum_{\mathcal{K}\subseteq[1:K]\atop \left|\mathcal{K}\right|=k}\sum_{\mathcal{N}\subseteq[1:N]\atop\left|\mathcal{N}\right|=n}\sum_{j\in [1:K]\backslash \mathcal{K}}\sum_{i\in [1:N]\backslash\mathcal{N}}I(\W_{{[1:K]}\backslash\mathcal{K}};Q^{[j]}_i,A^{[j]}_i|\W_{\mathcal{K}},\Z_{\mathcal{N}}).
\end{align}
We notice that when $n=N$ or $k=K$, then we get all the messages in the conditioning of the mutual information term above, and therefore we get the following boundary conditions on $T(n,k)$:
\begin{align}
\label{eq:BC-T-nk}
T(n=N,k) = 0,\  \forall k\in[0:K], \qquad T(n,k=K) = 0,\ \forall n\in[0:N].
\end{align}
\end{lemma}

\begin{proof}

We start by obtaining the following bound for all $k\in[1:K]$:
\begin{align}
\label{eq:lem2-step1}
I(\W_{[1:K]\backslash k}&; \Q_{[1:N]}^{[k]},\A_{[1:N]}^{[k]}|W_{k})\nonumber\\
&\overset{}{=}I(\W_{[1:K]\backslash k}; \Q_{[1:N]}^{[k]},\A_{[1:N]}^{[k]},W_{k})-I(\W_{[1:K]\backslash k}; W_{k})\nonumber\\
&\overset{(a)}{=}I(\W_{[1:K]\backslash k};\Q_{[1:N]}^{[k]})+I(\W_{[1:K]\backslash k};\A_{[1:N]}^{[k]}|\Q_{[1:N]}^{[k]})+I(\W_{[1:K]\backslash k};W_{k}|\Q_{[1:N]}^{[k]},\A_{[1:N]}^{[k]})\nonumber\\
&\overset{(b)}{=}H(\A_{[1:N]}^{[k]}|\Q_{[1:N]}^{[k]})-H(\A_{[1:N]}^{[k]}|\W_{[1:K]\backslash k},\Q_{[1:N]}^{[k]})+I(\W_{[1:K]\backslash k};W_{k}|\Q_{[1:N]}^{[k]},\A_{[1:N]}^{[k]})\nonumber\\
&=H(\A_{[1:N]}^{[k]}|\Q_{[1:N]}^{[k]})-H(\A_{[1:N]}^{[k]},W_{k}|\W_{[1:K]\backslash k},\Q_{[1:N]}^{[k]})+H(W_{k}|\W_{[1:K]\backslash k},\Q_{[1:N]}^{[k]},\A_{[1:N]}^{[k]})\nonumber\\
&\qquad+I(\W_{[1:K]\backslash k};W_{k}|\Q_{[1:N]}^{[k]},\A_{[1:N]}^{[k]})\nonumber\\
&\overset{}{\leq}D-H(W_{k}|\W_{[1:K]\backslash k},\Q_{[1:N]}^{[k]})-H(\A_{[1:N]}^{[k]}|\W_{[1:K]},\Q_{[1:N]}^{[k]})+H(W_{k}|\Q_{[1:N]}^{[k]},\A_{[1:N]}^{[k]})\nonumber\\
&\overset{(c)}{=}D-L+o(L),
\end{align}
where $(a)$ follows from the chain rule of mutual information and from the fact that the messages are i.i.d., $(b)$ follows from \eqref{cons3}
where queries are not functions of the messages; $(c)$ follows from \eqref{cons4} where answers are functions of the messages and the corresponding queries and also from the decodability constraint in \eqref{eq:decoding-const}, where $W_{k}$ is decodable from $\Q_{[1:N]}^{[k]}$ and $\A_{[1:N]}^{[k]}$. Summing up the obtained bound in \eqref{eq:lem2-step1} over $k\in[1:K]$, we get the following bound:
\begin{align}
\label{eq:lem2-step2}
D&\geq L+\frac{1}{K}\sum_{k=1}^K I(\W_{[1:K]\backslash k};\Q_{[1:N]}^{[k]},\A_{[1:N]}^{[k]}|W_{k})+o(L)\nonumber\\
&= L+\frac{1}{K}\frac{1}{N!} \sum_{k=1}^K \sum_{\bsigma\in [N!]} I(\W_{[1:K]\backslash k};\Q_{\bsigma_{[1:N]}}^{[k]},\A_{\bsigma_{[1:N]}}^{[k]}|W_{k})+o(L)\nonumber\\
&\overset{(a)}{=} L\hspace{-1pt} +\hspace{-1pt} \frac{1}{K}\frac{1}{N!} \sum_{n=1}^N   \sum_{k=1}^K \sum_{\bsigma\in [N!]} \hspace{-6pt} I(\W_{[1:K]\backslash k};Q_{\sigma_n}^{[k]},A_{\sigma_n}^{[k]}|\W_{k},\Q_{\bsigma_{[1:n-1]}}^{[k]},\A_{\bsigma_{[1:n-1]}}^{[k]})\hspace{-1pt} +\hspace{-1pt} o(L)\nonumber\\
&\overset{(b)}{\geq} L\hspace{-1pt} +\hspace{-1pt} \frac{1}{K}\frac{1}{N!} \sum_{n=1}^N   \sum_{k=1}^K \sum_{\bsigma\in [N!]} \hspace{-6pt} I(\W_{[1:K]\backslash k};Q_{\sigma_n}^{[k]},A_{\sigma_n}^{[k]}|\W_{k},\Z_{\bsigma_{[1:n-1]}})\hspace{-1pt} +\hspace{-1pt} o(L)\nonumber\\
&\overset{(c)}{=} L\hspace{-1pt} +\hspace{-1pt} \frac{1}{K(K-1)}\frac{1}{N!} \sum_{n=1}^N   \sum_{k=1}^K \sum_{j\in[1:K]\setminus k} \sum_{\bsigma\in [N!]} \hspace{-6pt} I(\W_{[1:K]\backslash k};Q_{\sigma_n}^{[j]},A_{\sigma_n}^{[j]}|\W_{k},\Z_{\bsigma_{[1:n-1]}})\hspace{-1pt} +\hspace{-1pt} o(L)\nonumber\\
&= L\hspace{-1pt} +\hspace{-1pt} \frac{1}{K\binom{K-1}{1}}\frac{1}{N!} \sum_{n=1}^N    \sum_{\substack{\mathcal{K}\subseteq[1:K]\\ \vert \mathcal{K}\vert =1}}\sum_{j\in[1:K]\setminus \mathcal{K}} \sum_{\bsigma\in [N!]} \hspace{-6pt} I(\W_{[1:K]\backslash \mathcal{K}};Q_{\sigma_n}^{[j]},A_{\sigma_n}^{[j]}|\W_{\mathcal{K}},\Z_{\bsigma_{[1:n-1]}})\hspace{-1pt} +\hspace{-1pt} o(L)\nonumber\\
&\overset{(d)}{=}L+ \sum_{n=1}^N   \sum_{\substack{\mathcal{K}\subseteq[1:K]\\ \vert \mathcal{K}\vert =1}} \sum_{j\in[1:K]\setminus \mathcal{K}} \hspace{-5pt}\frac{(N-n)!(n-1)!}{K\binom{K-1}{1}N!}\hspace{-3pt}\sum_{\mathcal{N}\subseteq[1:N]\atop|\mathcal{N}|=n-1}\sum_{i\in [1:N]\backslash \mathcal{N}}\hspace{-5pt} I(\W_{[1:K]\backslash \mathcal{K}};Q_i^{[j]},A_i^{[j]}|\W_{\mathcal{K}},\Z_{\mathcal{N}})\hspace{-1pt} + \hspace{-1pt}o(L)\nonumber\\
&=L+ \sum_{n=0}^{N-1}\frac{1}{KN\binom{N-1}{n}\binom{K-1}{1}} \sum_{\substack{\mathcal{K}\subseteq[1:K]\\ \vert \mathcal{K}\vert =1}} \sum_{j\in[1:K]\setminus \mathcal{K}}\sum_{\mathcal{N}\subseteq[1:N]\atop|\mathcal{N}|=n}\sum_{i\in [1:N]\backslash \mathcal{N}}\hspace{-5pt} I(\W_{[1:K]\backslash \mathcal{K}};Q_i^{[j]},A_i^{[j]}|\W_{\mathcal{K}},\Z_{\mathcal{N}}) \hspace{-1pt}+\hspace{-1pt} o(L)\nonumber\\
&\overset{(e)}{=}L+\sum_{n=0}^{N-1}T(n,1)+o(L),
\end{align}
where $(a)$ follows from chain rule of mutual information; $(b)$ follows from Lemma ~\ref{lem1}; $(c)$ follows from the privacy constraint in \eqref{eq:privacy-const} where the individual queries and answers are invariant with respect to the requested message index; $(d)$ follows from the symmetry with respect to the summation indexes, where for every set $\mathcal{N}\subseteq[1:N]$ of size $(n-1)$ and every $i\in[1:N]\backslash \mathcal{N}$, the number of permutations $\sigma$ that lead to the mutual information $I(\W_{[1:K]\backslash \mathcal{K}};Q_i^{[j]},A_i^{[j]}|\W_{\mathcal{K}},\Z_{\mathcal{N}})$ is $(N-n)!(n-1)!$; and $(e)$ follows from the definition of $T(n,k)$ in \eqref{eq:def-T_nk}. 
\end{proof}

In order to utilize the bound developed in Lemma~\ref{lem2}, we further lower bound the function $T(n,k)$ in the following Lemma. This lower bound on $T(n,k)$ has an interesting recursive structure, which in turn allows us to leverage the boundary conditions \eqref{eq:BC-T-nk} of the function $T(n,k)$ and thus obtain a closed-form lower bound on the download cost. 

\begin{lemma}\label{lem3}
The function $T(n,k)$ is lower bounded as follows:
\begin{align}
T(n,k)\geq \frac{1}{N-n}\left[\sum_{n'=n}^{N-1}T(n',k+1)+\lambda_{(n,k)}L\right] + o(L),
\end{align}
where $\lambda_{(n,k)}$ as defined in \eqref{eq:def-L_nk},
\begin{align}
\label{eq:def-L_nk2}
\lambda_{(n,k)}\ \overset{\Delta}{=}\ \frac{1}{KL\binom{K-1}{k}\binom{N}{n}} \sum_{\mathcal{K}\subseteq[1:K] \atop \left|\mathcal{K}\right|=k} \sum_{\mathcal{N}\subseteq[1:N]\atop\left|\mathcal{N}\right|=n}\sum_{j\in [1:K]\backslash \mathcal{K}} H(W_j|\Z_{\mathcal{N}},\W_{\mathcal{K}}),
\end{align}
for $n\in[0:N]$ and $k\in[0:K]$. 

%We notice that when $n=N$ or $k=K$, then we get all the messages in the conditioning of the entropy term above, and therefore we get the following boundary conditions on $\lambda_{(n,k)}$:
%\begin{align}
%\label{eq:BC-L-nk}
%\lambda_{(n=N,k)} = 0,\  \forall k\in[0:K], \qquad \lambda_{(n,k=K)} = 0,\ \forall n\in[0:N].
%\end{align}
%We further notice that for $n=0$ and all $k\in[0:K]$, we only have messages in the conditioning which are i.i.d., therefore, we get another set  of boundary conditions on $\lambda_{(n,k)}$:
%\begin{align}
%\label{eq:BC-L-nk2}
%\lambda_{(n=0,k)} &= \frac{1}{KL\binom{K-1}{k}\binom{N}{n}} \sum_{\mathcal{K}\subseteq[1:K] \atop \left|\mathcal{K}\right|=k} \sum_{\mathcal{N}\subseteq[1:N]\atop\left|\mathcal{N}\right|=n}\sum_{j\in [1:K]\backslash \mathcal{K}} H(W_j)\nonumber\\
%&=\frac{1}{KL\binom{K-1}{k}\binom{N}{n}} \sum_{j=1}^K \sum_{\mathcal{K}\subseteq[1:K]\backslash j \atop \left|\mathcal{K}\right|=k} \sum_{\mathcal{N}\subseteq[1:N]\atop\left|\mathcal{N}\right|=n} L \ = \ 1
%,\quad  \forall k\in[0:K].
%\end{align}
\end{lemma}

\begin{proof}

We start by bounding $T(n,k)$ defined in \eqref{eq:BC-T-nk} as follows,
\begin{align}
\label{eq:lem3_proof}
&T(n,k)\nonumber\\
&=\frac{1}{NK\binom{K-1}{k}\binom{N-1}{n}}\sum_{\mathcal{K}\subseteq[1:K]\atop \left|\mathcal{K}\right|=k}\sum_{\mathcal{N}\subseteq[1:N]\atop\left|\mathcal{N}\right|=n}\sum_{j\in [1:K]\backslash \mathcal{K}}\sum_{i\in [1:N]\backslash\mathcal{N}}I(\W_{{[1:K]}\backslash\mathcal{K}};Q^{[j]}_i,A^{[j]}_i|\W_{\mathcal{K}},\Z_{\mathcal{N}})\nonumber\\
&\overset{(a)}{=}\frac{1}{NK\binom{K-1}{k}\binom{N-1}{n}}\sum_{\mathcal{K}\subseteq[1:K]\atop \left|\mathcal{K}\right|=k}\sum_{\mathcal{N}\subseteq[1:N]\atop\left|\mathcal{N}\right|=n}\sum_{j\in [1:K]\backslash \mathcal{K}}\sum_{i\in [1:N]\backslash\mathcal{N}}I(\W_{{[1:K]}\backslash\mathcal{K}};A^{[j]}_i|\W_{\mathcal{K}},\Z_{\mathcal{N}},Q^{[j]}_i)\nonumber\\
&\overset{(b)}{=}\frac{1}{NK\binom{K-1}{k}\binom{N-1}{n}}\sum_{\mathcal{K}\subseteq[1:K]\atop \left|\mathcal{K}\right|=k}\sum_{\mathcal{N}\subseteq[1:N]\atop\left|\mathcal{N}\right|=n}\sum_{j\in [1:K]\backslash \mathcal{K}}\sum_{i\in [1:N]\backslash\mathcal{N}} H(A^{[j]}_i|\W_{\mathcal{K}},\Z_{\mathcal{N}},Q^{[j]}_i)\nonumber\\
&\overset{(c)}{\geq} \frac{1}{NK\binom{K-1}{k}\binom{N-1}{n}}\sum_{\mathcal{K}\subseteq[1:K]\atop \left|\mathcal{K}\right|=k}\sum_{\mathcal{N}\subseteq[1:N]\atop\left|\mathcal{N}\right|=n}\sum_{j\in [1:K]\backslash \mathcal{K}} \sum_{i\in [1:N]\backslash\mathcal{N}} H(A^{[j]}_i|\W_{\mathcal{K}},\Z_{\mathcal{N}},\Q^{[j]}_{[1:N]})\nonumber\\
&{\geq} \frac{1}{NK\binom{K-1}{k}\binom{N-1}{n}}\sum_{\mathcal{K}\subseteq[1:K]\atop \left|\mathcal{K}\right|=k}\sum_{\mathcal{N}\subseteq[1:N]\atop\left|\mathcal{N}\right|=n}\sum_{j\in [1:K]\backslash \mathcal{K}} H(\A^{[j]}_{[1:N]\backslash\mathcal{N}}|\W_{\mathcal{K}},\Z_{\mathcal{N}},\Q^{[j]}_{[1:N]})\nonumber\\
&\overset{(d)}{=}\frac{1}{NK\binom{K-1}{k}\binom{N-1}{n}}\sum_{\mathcal{K}\subseteq[1:K]\atop \left|\mathcal{K}\right|=k}\sum_{\mathcal{N}\subseteq[1:N]\atop\left|\mathcal{N}\right|=n}\sum_{j\in [1:K]\backslash \mathcal{K}} I(\W_{[1:K]\backslash\mathcal{K}};\A^{[j]}_{[1:N]\backslash\mathcal{N}}|\W_{\mathcal{K}},\Z_{\mathcal{N}},\Q^{[j]}_{[1:N]})\nonumber\\
&\overset{(e)}{=}\frac{1}{NK\binom{K-1}{k}\binom{N-1}{n}}\sum_{\mathcal{K}\subseteq[1:K]\atop \left|\mathcal{K}\right|=k}\sum_{\mathcal{N}\subseteq[1:N]\atop\left|\mathcal{N}\right|=n}\sum_{j\in [1:K]\backslash \mathcal{K}} I(\W_{[1:K]\backslash\mathcal{K}};\Q^{[j]}_{[1:N]},\A^{[j]}_{[1:N]\backslash\mathcal{N}}|\W_{\mathcal{K}},\Z_{\mathcal{N}})\nonumber\\
&\overset{(f)}{\geq}\frac{1}{NK\binom{K-1}{k}\binom{N-1}{n}}\sum_{\mathcal{K}\subseteq[1:K]\atop \left|\mathcal{K}\right|=k}\sum_{\mathcal{N}\subseteq[1:N]\atop\left|\mathcal{N}\right|=n}\sum_{j\in [1:K]\backslash \mathcal{K}} I(\W_{[1:K]\backslash\mathcal{K}};W_j,\Q^{[j]}_{[1:N]},\A^{[j]}_{[1:N]\backslash\mathcal{N}}|\W_{\mathcal{K}},\Z_{\mathcal{N}}) + o(L)\nonumber\\
%&=\frac{1}{NK\binom{K-1}{k}\binom{N-1}{n}}\sum_{\mathcal{K}\subseteq[1:K]\atop \left|\mathcal{K}\right|=k}\sum_{\mathcal{N}\subseteq[1:N]\atop\left|\mathcal{N}\right|=n}\sum_{j\in [1:K]\backslash \mathcal{K}} I(W_{[1:K]\backslash\mathcal{K}};W_j|W_{\mathcal{K}},Z_{\mathcal{N}})\nonumber\\
%&+\frac{1}{NK\binom{K-1}{k}\binom{N-1}{n}}\sum_{\mathcal{K}\subseteq[1:K]\atop \left|\mathcal{K}\right|=k}\sum_{\mathcal{N}\subseteq[1:N]\atop\left|\mathcal{N}\right|=n}\sum_{j\in [1:K]\backslash \mathcal{K}} I(W_{[1:K]\backslash\mathcal{K}};Q^{[j]}_{[1:N]},A^{[j]}_{[1:N]\backslash\mathcal{N}}|W_{\mathcal{K}},Z_{\mathcal{N}},W_j)\nonumber\\
&=\frac{1}{NK\binom{K-1}{k}\binom{N-1}{n}}\sum_{\mathcal{K}\subseteq[1:K]\atop \left|\mathcal{K}\right|=k}\sum_{\mathcal{N}\subseteq[1:N]\atop\left|\mathcal{N}\right|=n}\sum_{j\in [1:K]\backslash \mathcal{K}} I(\W_{[1:K]\backslash\mathcal{K}};\Q^{[j]}_{[1:N]},\A^{[j]}_{[1:N]\backslash\mathcal{N}}|\W_{\mathcal{K}},\Z_{\mathcal{N}},W_j)\nonumber\\
 &\hspace{15pt}+\frac{1}{NK\binom{K-1}{k}\binom{N-1}{n}}\sum_{\mathcal{K}\subseteq[1:K]\atop \left|\mathcal{K}\right|=k}\sum_{\mathcal{N}\subseteq[1:N]\atop\left|\mathcal{N}\right|=n}\sum_{j\in [1:K]\backslash \mathcal{K}} H(W_j|\W_{\mathcal{K}},\Z_{\mathcal{N}})+o(L)\nonumber\\
 &\overset{(g)}{=}\frac{1}{NK\binom{K-1}{k}\binom{N-1}{n}}\sum_{\mathcal{K}\subseteq[1:K]\atop \left|\mathcal{K}\right|=k}\sum_{\mathcal{N}\subseteq[1:N]\atop\left|\mathcal{N}\right|=n}\sum_{j\in [1:K]\backslash \mathcal{K}} I(\W_{[1:K]\backslash (\mathcal{K} \cup j)};\Q^{[j]}_{[1:N]\backslash\mathcal{N}},\A^{[j]}_{[1:N]\backslash\mathcal{N}}|\W_{(\mathcal{K}\cup j)},\Z_{\mathcal{N}})\nonumber\\
&\hspace{15pt}+\frac{1}{N-n}\lambda_{(n,k)}L+o(L)\nonumber\\
 &\overset{}{=}\underbrace{\frac{1}{NK\binom{K-1}{k}\binom{N-1}{n}}\sum_{\mathcal{K}\subseteq[1:K]\atop \left|\mathcal{K}\right|=k+1}\sum_{\mathcal{N}\subseteq[1:N]\atop\left|\mathcal{N}\right|=n}\sum_{j\in  \mathcal{K}} I(\W_{[1:K]\backslash \mathcal{K} };\Q^{[j]}_{[1:N]\backslash\mathcal{N}},\A^{[j]}_{[1:N]\backslash\mathcal{N}}|\W_{\mathcal{K}},\Z_{\mathcal{N}})}_{\tilde{T}(n,k)}\nonumber\\
&\hspace{15pt}+\frac{1}{N-n}\lambda_{(n,k)}L+o(L),
 \end{align}
where $(a)$ and $(e)$ follow from the fact that queries are independent from the messages; $(b)$ and $(d)$ follow from the fact that answers are functions of all the messages; $(c)$ is because conditioning reduces entropy; $(f)$ follows from the decoding constraint in \eqref{eq:decoding-const} where $W_{j}$ is decodable from $\Q_{[1:N]}^{[j]}$, $\A_{[1:N]\setminus \mathcal{N}}^{[j]}$ and $\Z_{\mathcal{N}}$; and $(g)$ follows from the definition of $\lambda_{(n,k)}$ in \eqref{eq:def-L_nk} and since the queries $\Q^{[j]}_{\mathcal{N}}$ are independent from the messages.
 We further lower bound $T(n,k)$ by bounding the term $\tilde{T}(n,k)$ in \eqref{eq:lem3_proof} as follows,
 \begin{align}
 \label{eq:tilde_T}
 &\tilde{T}(n,k)=\frac{1}{NK\binom{K-1}{k}\binom{N-1}{n}}\sum_{\mathcal{K}\subseteq[1:K]\atop \left|\mathcal{K}\right|=k+1}\sum_{\mathcal{N}\subseteq[1:N]\atop\left|\mathcal{N}\right|=n}\sum_{j\in  \mathcal{K}} I(\W_{[1:K]\backslash \mathcal{K} };\Q^{[j]}_{[1:N]\backslash\mathcal{N}},\A^{[j]}_{[1:N]\backslash\mathcal{N}}|\W_{\mathcal{K}},\Z_{\mathcal{N}})\nonumber\\
 &\overset{(a)}{=} \frac{1}{NK\binom{K-1}{k}\binom{N-1}{n}} \sum_{\mathcal{K}\subseteq[1:K]\atop \left|\mathcal{K}\right|=k+1} \sum_{j\in  \mathcal{K}} \frac{1}{n!(N-n)!}\sum_{\sigma\in [N!]} I(\W_{[1:K]\backslash\mathcal{K}};\Q^{[j]}_{\bsigma_{[n+1:N]}},\A^{[j]}_{\bsigma_{[n+1:N]}}|W_{\mathcal{K}},\Z_{\bsigma_{[1:n]}})\nonumber\\
&\overset{(b)}{=}\hspace{-2pt} \frac{1}{N!K\binom{K-1}{k}(N-n)}\hspace{-5pt} \sum_{\mathcal{K}\subseteq[1:K]\atop \left|\mathcal{K}\right|=k+1}\hspace{-2pt}  \sum_{j\in \mathcal{K}}\hspace{-1pt}  \sum_{\sigma\in [N!]}\hspace{-1pt} \sum_{n'=n}^{N-1} \hspace{-3pt} I(\W_{[1:K]\backslash\mathcal{K}};Q^{[j]}_{\sigma_{n'+1}},A^{[j]}_{\sigma_{n'+1}}|\W_{\mathcal{K}},\Z_{\bsigma_{[1:n]}}, \Q^{[j]}_{\bsigma_{[n+1:n']}},\A^{[j]}_{\bsigma_{[n+1:n']}})\nonumber\\
&\overset{(c)}{\geq} \frac{1}{N!K\binom{K-1}{k}(N-n)} \sum_{\mathcal{K}\subseteq[1:K]\atop \left|\mathcal{K}\right|=k+1} \sum_{j\in \mathcal{K}} \sum_{\sigma\in [N!]} \sum_{n'=n}^{N-1}  I(\W_{[1:K]\backslash\mathcal{K}};Q^{[j]}_{\sigma_{n'+1}},A^{[j]}_{\sigma_{n'+1}}|\W_{\mathcal{K}},\Z_{\bsigma_{[1:n']}})\nonumber\\
&\overset{(d)}{=} \sum_{n'=n}^{N-1}\frac{1}{N!K\binom{K-1}{k}(N-n)}\sum_{\mathcal{K}\subseteq[1:K]\atop \left|\mathcal{K}\right|=k+1}\sum_{j\in \mathcal{K}}  n'!(N-n'-1)!\hspace{-3pt} \sum_{\mathcal{N}\subseteq [1:N]\atop |\mathcal{N}|=n'}\sum_{i\in [1:N]\backslash\mathcal{N}} \hspace{-8pt}I(\W_{[1:K]\backslash\mathcal{K}};Q_i^{[j]},A_i^{[j]}|\W_{\mathcal{K}},\Z_{\mathcal{N}})\nonumber\\
&\overset{(e)}{=} \frac{1}{N-n}\sum_{n'=n}^{N-1}\frac{1}{NK\binom{K-1}{k}\binom{N-1}{n'}}\sum_{\mathcal{N}\subseteq [1:N]\atop |\mathcal{N}|=n'}\sum_{i\in [1:N]\backslash\mathcal{N}}\sum_{\mathcal{K}\subseteq[1:K]\atop \left|\mathcal{K}\right|=k+1}\sum_{j\in \mathcal{K}} \sum_{j'\in [1:K]\backslash \mathcal{K}} \hspace{-8pt}\frac{I(\W_{[1:K]\backslash\mathcal{K}};Q_i^{[j']},A_i^{[j']}|\W_{\mathcal{K}},\Z_{\mathcal{N}})}{K-k-1}\nonumber\\
&{=} \frac{1}{N-n}\sum_{n'=n}^{N-1}\frac{1}{NK\binom{K-1}{k+1}\binom{N-1}{n'}}\sum_{\mathcal{N}\subseteq [1:N]\atop |\mathcal{N}|=n'}\sum_{i\in [1:N]\backslash\mathcal{N}}\sum_{\mathcal{K}\subseteq[1:K]\atop \left|\mathcal{K}\right|=k+1} \sum_{j'\in [1:K]\backslash \mathcal{K}} \hspace{-8pt}{I(\W_{[1:K]\backslash\mathcal{K}};Q_i^{[j']},A_i^{[j']}|\W_{\mathcal{K}},\Z_{\mathcal{N}})}\nonumber\\
&\overset{(f)}{=}\frac{1}{N-n}\sum_{n'=n}^{N-1}T(n',k+1),
\end{align}
where $(a)$ and $(d)$ follow from a similar argument to step $(d)$ in \eqref{eq:lem2-step2}; $(b)$ follows by applying the chain rule of entropy; $(c)$ follows by applying Lemma \ref{lem1};
$(e)$ follows from the privacy constraint in \eqref{eq:privacy-const} where the individual queries and answers are invariant with respect to the requested message index; and $(f)$ follows from the definition of $T(n,k)$ in \eqref{eq:def-T_nk}.
By applying the lower bound on the term $\tilde{T}(n,k)$ in \eqref{eq:tilde_T} to the lower bound on ${T}(n,k)$ in \eqref{eq:lem3_proof} we conclude the proof of Lemma~\ref{lem3}.
\end{proof}

Now, we use the recursive lower bound on $T(n,k)$ given in Lemma~\ref{lem3} to further lower bound the download cost on $D$ obtained in Lemma~\ref{lem2} as follows:
\begin{align}
\label{eq:bound-step1}
D&\geq L+\sum_{n_1=0}^{N-1}T(n_1,1)+o(L)\nonumber\\
&\geq L + \sum_{n_1=0}^{N-1} \frac{1}{N-n_1}\left(\lambda_{(n_1,1)}L + \sum_{n_2=n_1}^{N-1} T(n_2,2)  \right)+o(L)\nonumber\\
&\geq L + \sum_{n_1=0}^{N-1} \frac{\lambda_{(n_1,1)}L}{N-n_1} + \sum_{n_1=0}^{N-1}\sum_{n_2=n_1}^{N-1} \frac{1}{(N-n_1)(N-n_2)} \left(\lambda_{(n_2,2)}L + \sum_{n_3=n_2}^{N-1} T(n_3,3) \right)+o(L)\nonumber\\
&\hspace{7pt}\vdots \nonumber\\
& {\geq} \hspace{-1pt}L\hspace{-2pt} +\hspace{-1pt} \sum_{n_1=0}^{N-1} \hspace{-1pt}\frac{\lambda_{(n_1,1)} L}{N-n_1}\hspace{-1pt}+\hspace{-2pt} \sum_{n_1=0}^{N-1}\sum_{n_2=n_1}^{N-1}\hspace{-1pt} \frac{\lambda_{(n_2,2)}L}{(N-n_1)(N-n_2)}  \hspace{-1pt}+\hspace{-2pt}\sum_{n_1=0}^{N-1}\sum_{n_{2}=n_{1}}^{N-1}\sum_{n_{3}=n_{2}}^{N-1}\hspace{-1pt} \frac{ \lambda_{(n_{3},3)}L}{(N-n_1)(N-n_2)(N-n_3)} \nonumber \\
&\hspace{15pt}+\cdots+ \sum_{n_1=0}^{N-1}\ldots\sum_{n_{K-1}=n_{K-2}}^{N-1} \frac{\lambda_{(n_{K-1},K-1)}L+\sum_{n_{K}=n_{K-1}}^{N-1} {T(n_K,K)}}{(N-n_1) \times\cdots\times (N-n_{K-1})}+o(L)\nonumber\\
& \overset{(a)}{=} \hspace{-1pt}L\hspace{-2pt} +\hspace{-1pt} \sum_{n_1=0}^{N-1} \hspace{-1pt}\frac{\lambda_{(n_1,1)} L}{N-n_1}\hspace{-1pt}+\hspace{-2pt} \sum_{n_1=0}^{N-1}\sum_{n_2=n_1}^{N-1}\hspace{-1pt} \frac{\lambda_{(n_2,2)}L}{(N-n_1)(N-n_2)}  \hspace{-1pt}+\hspace{-2pt}\sum_{n_1=0}^{N-1}\sum_{n_{2}=n_{1}}^{N-1}\sum_{n_{3}=n_{2}}^{N-1}\hspace{-1pt} \frac{ \lambda_{(n_{3},3)}L}{(N-n_1)(N-n_2)(N-n_3)} \nonumber \\
&\hspace{15pt}+\cdots+ \sum_{n_1=0}^{N-1}\ldots\sum_{n_{K-1}=n_{K-2}}^{N-1} \frac{\lambda_{(n_{K-1},K-1)}L}{(N-n_1) \times\cdots\times (N-n_{K-1})}+o(L)\nonumber\\
& \overset{(b)}{=} L + \sum_{n_1=1}^{N} \frac{\lambda_{(N-n_1,1)} L}{n_1}+ \sum_{n_1=1}^{N}\sum_{n_2=1}^{N-n_1} \frac{\lambda_{(N-n_2,2)}L}{n_1n_2}  +\cdots+ \sum_{n_1=1}^{N}\ldots\sum_{n_{K-1}=1}^{N-n_{K-2}} \frac{\lambda_{(N-n_{K-1},K-1)}L}{n_1 \times\cdots\times n_{K-1}}  +o(L)\nonumber\\
& \overset{(c)}{=} L + \sum_{n_1=1}^{N} \frac{\lambda_{(N-n_1,1)}L}{n_1}+ \sum_{n_1=1}^{N}\sum_{n_2=n_1}^{N} \frac{\lambda_{(N-n_1,2)}L}{n_1n_2}  +\cdots+ \sum_{n_1=1}^{N}\ldots\hspace{-6pt}\sum_{n_{K-1}=n_{K-2}}^{N}\hspace{-4pt} \frac{\lambda_{(N-n_1,K-1)}L}{n_1 \times\cdots\times n_{K-1}}  +o(L),
\end{align}
where $(a)$ follows by applying the boundary condition on $T(n,k)$ as given in \eqref{eq:BC-T-nk}, where $T(n,k=K)=0$ for $n\in[0:N-1]$; and $(b)$ and $(c)$ follow by simple changing of the summation indexes.
Taking the limit $L\rightarrow \infty$, we obtain the bound on $\frac{D}{L}$, which is also a valid bound on the optimal normalized download cost, $D^*(\mu)$, as defined in \eqref{optimal-download-cost} for $\mu \in[\frac{1}{N},1]$,
since the bound in \eqref{eq:bound-step3} is valid for any achievable pair $(D,L)$. Therefore, we obtain the following bound on $D^*(\mu)$,
\begin{align}
\label{eq:bound-thm1}
D^*(\mu)\geq 1 + \sum_{n_1=1}^{N} \frac{\lambda_{(N-n_1,1)}}{n_1}+ \sum_{n_1=1}^{N}\sum_{n_2=n_1}^{N} \frac{\lambda_{(N-n_1,2)}}{n_1n_2}  +\cdots+ \sum_{n_1=1}^{N}\ldots\hspace{-6pt}\sum_{n_{K-1}=n_{K-2}}^{N}\hspace{-4pt} \frac{\lambda_{(N-n_1,K-1)}}{n_1 \times\cdots\times n_{K-1}},
\end{align}
which completes the proof of Theorem~\ref{theorem1}.

\section{Proof of Theorem \ref{theorem2}: Lower Bounds for Uncoded Storage Constrained  Databases and General $(N,K,\mu)$}
\label{sec:converse}

We now specialize the lower bound in (\ref{eq:bound-step1}) for the case of uncoded storage placement as defined in Section~\ref{re:uncoded-assump}. Using the uncoded storage content given in \eqref{eq:storage-content2}, and since the messages are i.i.d., the following relation is already satisfied,
\begin{align}
\label{eq:const-uncoded}
H(W_j/\Z_{\mathcal{N}},\W_{\mathcal{K}}) =H(W_j/\Z_{\mathcal{N}}) =\sum_{\substack{\mathcal{S} \subseteq [1:N]\setminus \mathcal{N}\\ \vert \mathcal{S}\vert \geq 1 }} \vert W_{j,\mathcal{S}}\vert L,
\end{align}
for every $\mathcal{K}\subseteq [1:K]$, $\mathcal{N}\subseteq [1:N]$, and $j\in[1:K]\setminus \mathcal{K}$. Therefore, the term $\lambda_{(n,k)}$ as defined in \eqref{eq:def-L_nk} can be expressed as,
\begin{align}
\lambda_{(n,k)} &= \frac{1}{K\binom{K-1}{k}\binom{N}{n}}\sum_{\substack{\mathcal{K}\subseteq[1:K]\\ \vert \mathcal{K}\vert =k}} \sum_{\substack{\mathcal{N}\subseteq[1:N]\\ \vert \mathcal{N}\vert =n}} \sum_{j\in[1:K]\setminus \mathcal{K}} H(W_j/\Z_{\mathcal{N}})\nonumber\\
 &=\frac{1}{K\binom{K-1}{k}\binom{N}{n}} \sum_{\substack{\mathcal{N}\subseteq[1:N]\\ \vert \mathcal{N}\vert =n}} \sum_{j=1}^K \sum_{\substack{\mathcal{K}\subseteq[1:K]\setminus \{j\}\\ \vert \mathcal{K}\vert =k}}
 H(W_j/\Z_{\mathcal{N}})=\frac{1}{K\binom{N}{n}} \sum_{\substack{\mathcal{N}\subseteq[1:N]\\ \vert \mathcal{N}\vert =n}} \sum_{j=1}^K H(W_j/\Z_{\mathcal{N}}).
\end{align}
%where $(a)$ follows since every message index, $j\in[1:K]$, in LHS appears whenever $j\not \in\mathcal{K}$ in total number of $\binom{K-1}{k}$ times in the summation over $\mathcal{K}$.
We notice that $\lambda_{(n,k)}$ is independent of $k$, and hence we can define $\lambda_{n} \overset{\Delta}{=} \lambda_{(n,k)}$, for all $k\in [1:K]$. Therefore, we can write the bound in \eqref{eq:bound-step1} as follows,
\begin{align}
\label{eq:bound-step2}
 D&\geq L + \sum_{n_1=1}^{N} \left(\frac{1}{n_1}+ \sum_{n_2=n_1}^{N} \frac{1}{n_1n_2}  +\cdots+ \sum_{n_2=n_1}^{N}\ldots\sum_{n_{K-1}=n_{K-2}}^{N} \frac{1}{n_1 \times\cdots\times n_{K-1}} \right)\lambda_{N-n_1}  +o(L)\nonumber\\
 & = L + \sum_{n_1=1}^{N}  S(n_1,K) \ \lambda_{N-n_1} +o(L),
\end{align}
where $S(n,k)$, for $n\in[1:N]$ and $k\in[1:K]$, is defined as follows,
\begin{align}
\label{eq:def-S_nk}
S(n,k) \overset{\Delta}{=} \frac{1}{n}+ \sum_{n_2=n}^{N} \frac{1}{n n_2}  +\cdots+ \sum_{n_2=n}^{N}\ldots\sum_{n_{k-1}=n_{k-2}}^{N} \frac{1}{n n_2 \times\cdots\times n_{k-1}}.
\end{align}
It is important to notice the following boundary conditions and properties of $S(n,k)$:
\begin{align}\label{eq:BC-S_nk}
&\textbf{Property 1: }
S(n,k=1) = 0, \qquad\qquad\qquad \textbf{Property 2: }S(n,k=2) = \frac{1}{n},\nonumber\\  
&\textbf{Property 3: } nS(n=N,k) = S(n=N,k-1) +1,\nonumber\\  
&\textbf{Property 4: } nS(n,k)-(n+1)S(n+1,k)=S(n,k-1).
\end{align}
The first $3$ properties are straight forward to prove from the definition of $S(n,k)$ in \eqref{eq:def-S_nk}.
The forth property of $S(n,k)$ provides a useful recursive relation and can be proved as follows:
\begin{align}
nS(n,k)&-(n+1)S(n+1,k)\nonumber\\
&\overset{(a)}{=}\left(1+\sum_{n_2=n}^N\frac{1}{n_2}+\sum_{n_2=n}^N\sum_{n_3=n_2}^N\frac{1}{n_2n_3}+\cdots+\sum_{n_2=n}^N\cdots\sum_{n_{k-1}=n_{k-2}}\frac{1}{n_2\times\cdots\times n_{k-1}}\right)\nonumber\\
&\qquad- \left(1+\sum_{n_2=n+1}^N\frac{1}{n_2}+\sum_{n_2=n+1}^N\sum_{n_3=n_2}^N\frac{1}{n_2n_3}+\cdots+\sum_{n_2=n+1}^N\cdots\sum_{n_{k-1}=n_{k-2}}\frac{1}{n_2\times\cdots\times n_{k-1}}\right)\nonumber\\
&=\frac{1}{n}+\sum_{n_3=n}^N\frac{1}{nn_3}+\sum_{n_3=n}^N\sum_{n_4=n_3}^N\frac{1}{nn_3n_4}+\cdots+\sum_{n_3=n}^N\cdots\sum_{n_{k-1}=n_{k-2}}^N\frac{1}{nn_3\times\cdots\times n_{k-1}}\nonumber\\
&\overset{(b)}{=}\frac{1}{n}+\sum_{n_2=n}^N\frac{1}{nn_2}+\sum_{n_2=n}^N\sum_{n_3=n_2}^N\frac{1}{nn_2n_3}+\cdots+\sum_{n_2=n}^N\cdots\sum_{n_{k-2}=n_{k-3}}^N\frac{1}{nn_2\times\cdots\times n_{k-2}}\nonumber\\
&\overset{(c)}{=}S(n,k-1),
\end{align}
where $(a)$ and $(c)$ follow from the definition of $S(n,k)$ in \eqref{eq:def-S_nk}; and $(b)$ follows by simple relabeling of the summation indexes.

Next, we express the $\lambda_n$ term that appears in \eqref{eq:bound-step2} in terms of $x_{\ell}$ as defined in  \eqref{eq:def-x_ell} as follows,
\begin{align}
\label{eq:bound-step3}
\lambda_{n} &=\frac{1}{K\binom{N}{n}} \sum_{\substack{\mathcal{N}\subseteq[1:N]\\ \vert \mathcal{N}\vert =n}} \sum_{k=1}^K  H(W_k/\Z_{\mathcal{N}})\overset{(a)}{=} \frac{1}{K\binom{N}{n}} \sum_{\substack{\mathcal{N}\subseteq[1:N]\\ \vert \mathcal{N}\vert =n}} \sum_{k=1}^K   \sum_{\substack{\mathcal{S} \subseteq [1:N]\setminus \mathcal{N}\\ \vert \mathcal{S}\vert \geq 1 }} \vert W_{k,\mathcal{S}}\vert L\nonumber \\
&=\frac{1}{K\binom{N}{n}} \sum_{\ell = 1}^{N-n}\sum_{k=1}^K   \sum_{\substack{\mathcal{N}\subseteq[1:N]\\ \vert \mathcal{N}\vert =n}}  \sum_{\substack{\mathcal{S} \subseteq [1:N]\setminus \mathcal{N}\\ \vert \mathcal{S}\vert =\ell }} \vert W_{k,\mathcal{S}}\vert L 
=\frac{1}{K\binom{N}{n}} \sum_{\ell = 1}^{N-n}\sum_{k=1}^K \sum_{\substack{\mathcal{S} \subseteq [1:N]\\ \vert \mathcal{S}\vert =\ell }} \sum_{\substack{\mathcal{N}\subseteq[1:N] \setminus \mathcal{S}\\ \vert \mathcal{N}\vert =n}} \vert W_{k,\mathcal{S}}\vert L \nonumber\\
&=\frac{1}{K\binom{N}{n}} \sum_{\ell = 1}^{N-n}\sum_{k=1}^K    \sum_{\substack{\mathcal{S} \subseteq [1:N]\\ \vert \mathcal{S}\vert =\ell }}\binom{N-\ell}{n}  \vert W_{k,\mathcal{S}}\vert L\nonumber\\
&=  \sum_{\ell = 1}^{N-n}   \binom{N-n}{\ell}  \frac{1}{K\binom{N}{\ell}} \sum_{k=1}^K \sum_{\substack{\mathcal{S} \subseteq [1:N]\\ \vert \mathcal{S}\vert =\ell }} \vert W_{k,\mathcal{S}}\vert L = \sum_{\ell = 1}^{N-n}   \binom{N-n}{\ell}  x_{\ell} L,
\end{align}
where $(a)$ follows from \eqref{eq:const-uncoded}.
%; and $(b)$ follows since the term $\vert W_k^{\mathcal{S}} \vert$ appears on the LHS of $(b)$, for some $k\in [1:K]$ and $\mathcal{S} \subseteq [1:N]$ where $\vert {\mathcal{S}} \vert =\ell$, only if ${\mathcal{S}} \cap {\mathcal{N}} = \phi$
Substituting with \eqref{eq:bound-step3}  in \eqref{eq:bound-step2} and taking the limit $L\rightarrow \infty$, we obtain the bound on $\frac{D}{L}$ in terms of $x_{\ell}$ as follows,
\begin{align}
\label{eq:bound-step4}
 \frac{D}{L}&\geq  1 + \sum_{n_1=1}^{N}\sum_{\ell = 1}^{n_1}  \binom{n_1}{\ell}  S(n_1,K) \  x_{\ell} =1 + \sum_{\ell = 1}^{N} \sum_{n_1=\ell}^{N}  \binom{n_1}{\ell}  S(n_1,K) \  x_{\ell} =1 + \sum_{\ell = 1}^{N} \alpha(\ell,K) \  x_{\ell},
\end{align}
where $\alpha(\ell,k)$ for $\ell\in[1:N]$ and $k\in[1:K]$ is defined as follows,
\begin{align}
\label{eq:def-a_lk}
\alpha(\ell,k)\ \overset{\Delta}{=} \ \sum_{n=\ell}^N\binom{n}{\ell} \  S(n,k).
\end{align}
%The following boundary conditions on $\alpha(\ell,k)$ readily follow:
%\begin{align}\label{eq:BC-a_nk}
%\alpha(\ell ,k=1) = 0, \qquad \alpha(\ell=N ,k) = S(n=N,k).
%\end{align}

Since the bound in \eqref{eq:bound-step3} is valid for any achievable pair $(D,L)$, it is also a valid bound on the optimal normalized download cost, $D^*(\mu)$, as defined in \eqref{optimal-download-cost},  where $\mu \in[\frac{1}{N},1]$. Therefore, we obtain the following bound on $D^*(\mu)$,
\begin{align}
\label{eq:bound-step5}
D^*(\mu)\geq 1 + \sum_{\ell = 1}^{N} \alpha(\ell,K) \  x_{\ell}.
\end{align}

Next, we use the properties of $S(n,k)$ in \eqref{eq:BC-S_nk} to obtain a recursion relation for $\alpha(\ell,k)$ as introduced in the following Lemma:
\begin{lemma}\label{lem5}
The function $\alpha(\ell,k)$ satisfies the  following recursion relation:
\begin{align}
\alpha(\ell,k)=\frac{1}{\ell}\left[\alpha(\ell,k-1)+\binom{N}{\ell}\right].
\end{align}
\end{lemma}
\begin{proof}
\begin{align}
\alpha(\ell,k)&\overset{(a)}{=}\sum_{n=\ell}^N\binom{n}{\ell}S(n,k)=\frac{1}{\ell}\sum_{n=\ell}^N\binom{n-1}{\ell-1}nS(n,k)\nonumber\\
&=\frac{1}{\ell}\left[\sum_{n=\ell}^N\binom{n-1}{\ell-1}nS(n,k)+\sum_{n=\ell+1}^N\binom{n-1}{\ell}nS(n,k)-\sum_{n=\ell+1}^N\binom{n-1}{\ell}nS(n,k)\right]\nonumber\\
%&=\frac{1}{\ell}\left[\binom{\ell-1}{\ell-1}\ell S(\ell,k)+\sum_{n=\ell+1}^N\binom{n}{\ell}nS(n,k)-\sum_{n=\ell}^{N-1}\binom{n}{\ell}(n+1)S(n+1,k)\right]\nonumber\\
&=\frac{1}{\ell}\left[\sum_{n=\ell}^{N}\binom{n}{\ell}nS(n,k)-\sum_{n=\ell}^{N-1}\binom{n}{\ell}(n+1)S(n+1,k)\right]\nonumber\\
&=\frac{1}{\ell}\left[\sum_{n=\ell}^{N-1}\binom{n}{\ell}[nS(n,k)-(n+1)S(n,k)]+\binom{N}{\ell}NS(N,k)\right]\nonumber\\
&\overset{(b)}{=}\frac{1}{\ell}\left[\sum_{n=\ell}^{N-1}\binom{n}{\ell}S(n,k-1)+\binom{N}{\ell}NS(N,k-1)+\binom{N}{\ell}\right]\nonumber\\
&=\frac{1}{\ell}\left[\sum_{n=\ell}^{N}\binom{n}{\ell}S(n,k-1)+\binom{N}{\ell}\right]\nonumber\\
&\overset{(c)}{=}\frac{1}{\ell}\left[\alpha(\ell,k-1)+\binom{N}{\ell}\right],
\end{align}
where $(a)$ and $(c)$ follow from the definition of $\alpha(\ell,k)$ in \eqref{eq:def-a_lk}; and $(b)$ follows from properties $3$ and $4$  in \eqref{eq:BC-S_nk}.
\end{proof}

Next, we use the recursion relation for $\alpha(\ell,k)$ given in Lemma~\ref{lem5}  to obtain a closed form expression for the coefficients $\alpha(\ell,K)$, for $\ell\in[1:N]$, in terms of the system parameters as follows:
\begin{align}
\label{eq:a_lk}
\alpha(\ell,K) &= \frac{1}{\ell} \left( \alpha(\ell,K-1) +\binom{N}{\ell} \right)\nonumber\\
& = \frac{1}{\ell} \binom{N}{\ell} + \frac{1}{\ell^2} \left( \alpha(\ell,K-2) +\binom{N}{\ell} \right)\nonumber\\
&\hspace{7pt} \vdots \nonumber\\
& = \binom{N}{\ell} \left( \frac{1}{\ell}+\frac{1}{\ell^2}+\cdots +\frac{1}{\ell^{K-1}} \right)= \binom{N}{\ell} \left( \tilde{D}(\ell)-1 \right),
\end{align}
which follows by applying the boundary condition on $\alpha(\ell,k)$ where $\alpha(\ell,k=1)=0$, and $\tilde{D}(\ell) =\sum_{k=0}^{K-1} \frac{1}{\ell^k}$ as defined in \eqref{eq:def-D_l}.
Therefore, the bound in \eqref{eq:bound-step5} can be written as
\begin{align}
\label{eq:objective}
D^*(\mu)\geq 1 + \sum_{\ell = 1}^{N} \binom{N}{\ell} \left( \tilde{D}(\ell)-1 \right)\  x_{\ell}.
\end{align}

Next, we  obtain $N-1$ different lower bounds on $D^*(\mu)$, by eliminating the pairs $(x_{j-1},x_{j})$, for each ${j \in \left[1:N-1\right]}$, in the equation \eqref{eq:objective} using the message size, and the storage constraints for uncoded storage placement given in \eqref{eq:size-const}, and \eqref{eq:storage-const}, respectively. 
We use \eqref{eq:size-const} to write $x_{j}$ as follows:
\begin{align}
\label{eq:x_j}
&x_{j}=\frac{1}{\binom{N}{j}}\left(1-\sum_{\ell\in\left[1:N\right]\setminus j}\binom{N}{\ell}x_{\ell}\right).
\end{align}
We first apply \eqref{eq:x_j} in \eqref{eq:objective} to obtain
\begin{align}
\label{eq:objective-step2}
D^*(\mu)&\geq 1 + \sum_{\ell \in [1:N]\setminus j} \binom{N}{\ell} \left( \tilde{D}(\ell)-1 \right)\  x_{\ell} + \left(1-\sum_{\ell\in\left[1:N\right]\setminus j}\binom{N}{\ell}x_{\ell}\right)\left( \tilde{D}(j)-1 \right)\nonumber\\
& =\tilde{D}(j) + \sum_{\ell \in [1:N]\setminus j} \binom{N}{\ell} \left( \tilde{D}(\ell)-\tilde{D}(j) \right)\  x_{\ell}.
\end{align}
We next apply \eqref{eq:x_j} in the storage constraint \eqref{eq:storage-const} to obtain
\begin{align}
\mu N &\geq \sum_{\ell\in\left[1:N\right]\setminus j}\ell \binom{N}{\ell} x_{\ell}+j \left(1-\sum_{\ell\in\left[1:N\right]\setminus j}\binom{N}{\ell}x_{\ell}\right)=j + \sum_{\ell\in\left[1:N\right]\setminus j} \binom{N}{\ell} (\ell-j) x_{\ell}. \label{eq:objective-step3} 
\end{align}
In order to eliminate $x_{j+1}$ from \eqref{eq:objective-step2}, we first use \eqref{eq:objective-step3}  to bound $x_{j+1}$ as
\begin{align}
&x_{j+1}\leq \frac{1}{\binom{N}{j+1}}\left(\mu N-j-\sum_{\ell\in\left[1:N\right]\setminus \{j,j+1\}}\binom{N}{\ell}(\ell-j)x_{\ell}\right), \label{eq:x_j+1}
\end{align}
which can be applied in \eqref{eq:objective-step2} to obtain the following bound on $D^*(\mu)$,
\begin{align}
\label{eq:objective-step4} 
D^*(\mu)&\geq \tilde{D}(j) + \sum_{\ell \in [1:N]\setminus j} \binom{N}{\ell} \left( \tilde{D}(\ell)-\tilde{D}(j) \right)\  x_{\ell}\nonumber\\
& \overset{(a)}{\geq} \tilde{D}(j) + \sum_{\ell \in [1:N]\setminus \{j,j+1\}}\binom{N}{\ell} \left( \tilde{D}(\ell)-\tilde{D}(j) \right) x_{\ell} \nonumber\\
&\hspace{15pt}+ \left( \tilde{D}(j+1)-\tilde{D}(j) \right)\left(\mu N-j-\sum_{\ell\in\left[1:N\right]\setminus \{j,j+1\}}\binom{N}{\ell}(\ell-j)x_{\ell}\right)\nonumber\\
& \overset{(b)}{=} (\mu N -j)\tilde{D}(j+1) -(\mu N -j-1)\tilde{D}(j) +\sum_{\ell\in\left[1:N\right]\setminus \{j,j+1\}} \binom{N}{\ell} \Gamma_{\ell}^{(j)}\ x_{\ell}\nonumber\\
&\overset{(c)}{\geq}(\mu N -j)\tilde{D}(j+1) -(\mu N -j-1)\tilde{D}(j),
\end{align} 
where $(a)$ follows from \eqref{eq:x_j+1} where the coefficient $\tilde{D}(j+1)-\tilde{D}(j)$  is negative for all ${j\in\left[1:N-1\right]}$; $\Gamma_{\ell}^{(j)}$ for $\ell\in[1:N]\setminus \{j,j+1\}$ in $(b)$ is defined as
\begin{align}
\label{eq:def-L_l}
\Gamma_{\ell}^{(j)}\overset{\Delta}{=} \tilde{D}(\ell) +(\ell-j-1)\tilde{D}(j) -(\ell-j)\tilde{D}(j+1);
\end{align}
 and $(c)$ since $x_{\ell}$ and $\Gamma_{\ell}^{(j)}$ are non-negative for ${\ell\in\left[1:N\right]\setminus \{j,j+1\}}$, which can be shown in the following discussion.
 In order to prove that $\Gamma_{\ell}^{(j)}$ is non-negative for ${\ell\in\left[1:N\right]\setminus \{j,j+1\}}$, we first need to prove an important property for $\tilde{D}(\ell)$ in the following Lemma:
\begin{lemma}\label{lem6}
$\tilde{D}(\ell)-\tilde{D}(\ell+1)$ is non increasing with respect to $\ell$, i.e., $\tilde{D}(\ell')-\tilde{D}(\ell'+1)\geq \tilde{D}(\ell)-\tilde{D}(\ell+1)$, for any $\ell' \leq \ell$.
\end{lemma}
\begin{proof}
In order to prove Lemma~\ref{lem6}, it is sufficient to prove that $\tilde{D}(\ell')-\tilde{D}(\ell'+1)\geq \tilde{D}(\ell)-\tilde{D}(\ell+1)$ for $\ell' = \ell -1$, or $\tilde{D}(\ell -1)-2\tilde{D}(\ell)+\tilde{D}(\ell+1) \geq 0$. The proof for any $\ell' \leq \ell$ follows by induction.
\begin{align}
\tilde{D}(\ell-1)-2\tilde{D}(\ell)+\tilde{D}(\ell+1)&=\sum_{j=0}^{K-1}\frac{1}{(\ell-1)^j}+\frac{1}{(\ell+1)^j}-\frac{2}{(\ell)^j}\nonumber\\
&\overset{(a)}{\geq} 2\sum_{j=0}^{K-1} \frac{1}{(\ell^2-1)^{j/2}}-\frac{1}{(\ell)^j}\nonumber\\
&=2\sum_{j=0}^{K-1}\frac{(\ell^2)^{j/2}-(\ell^2-1)^{j/2}}{(\ell)^j(\ell^2-1)^{j/2}}\geq 0,
\end{align}
where, $(a)$ follows from the AM-GM inequality, i.e., arithmetic mean is larger than geometric mean, that is $\frac{x_1+x_2}{2}\geq {(x_1x_2)^{1/2}},\ \forall x_1,x_2\geq 0$. 
Therefore, we obtain $\tilde{D}(\ell-1)-\tilde{D}(\ell)\geq \tilde{D}(\ell)-\tilde{D}(\ell+1)$, which completes the proof of the Lemma.
\end{proof}

 \noindent $\bullet$ {Case $\ell < j$:} We prove that $\Gamma_{\ell}^{(j)}\geq 0$ for ${\ell < j}$ as follows,
\begin{align}
\Gamma_{\ell}^{(j)}&= \tilde{D}(\ell) +(\ell-j-1)\tilde{D}(j) -(\ell-j)\tilde{D}(j+1)\nonumber\\
& = \left[\tilde{D}(\ell)- \tilde{D}(j)\right] - (j-\ell) \left[\tilde{D}(j)- \tilde{D}(j+1)\right] \nonumber\\
& = \sum_{i=\ell}^{j-1}\left[\tilde{D}(i)- \tilde{D}(i+1)\right] - (j-\ell) \left[\tilde{D}(j)- \tilde{D}(j+1)\right] \nonumber\\
& = \sum_{i=\ell}^{j-1}\left(\left[\tilde{D}(i)- \tilde{D}(i+1)\right] -  \left[\tilde{D}(j)- \tilde{D}(j+1)\right]\right) \nonumber\\
& \overset{(a)}{\geq} \sum_{i=\ell}^{j-1}\left(\left[\tilde{D}(j)- \tilde{D}(j+1)\right] -  \left[\tilde{D}(j)- \tilde{D}(j+1)\right]\right)  =0,
\end{align}
where $(a)$ follows from Lemma~\ref{lem6}.

  \noindent $\bullet$ {Case $\ell >j+1$:}
Similar to the case ${\ell < j}$, we prove  $\Gamma_{\ell}^{(j)}\geq 0$ for ${\ell > j+1}$ as follows,
\begin{align}
\Gamma_{\ell}^{(j)}&= \tilde{D}(\ell) +(\ell-j-1)\tilde{D}(j) -(\ell-j)\tilde{D}(j+1)\nonumber\\
& = (\ell-j-1)\left[\tilde{D}(j)- \tilde{D}(j+1)\right] - \left[\tilde{D}(j+1)- \tilde{D}(\ell)\right] \nonumber\\
& = (\ell-j-1)\left[\tilde{D}(j)- \tilde{D}(j+1)\right] - \sum_{i=j+2}^{\ell}\left[\tilde{D}(i-1)- \tilde{D}(i)\right] \nonumber\\
& = \sum_{i=j+2}^{\ell}\left(\left[\tilde{D}(j)- \tilde{D}(j+1)\right] -  \left[\tilde{D}(i-1)- \tilde{D}(i)\right]\right) \nonumber\\
& \overset{(a)}{\geq} \sum_{i=j+2}^{\ell}\left(\left[\tilde{D}(j)- \tilde{D}(j+1)\right] -  \left[\tilde{D}(j)- \tilde{D}(j+1)\right]\right)  =0,
\end{align}
where $(a)$ follows from Lemma~\ref{lem6}.

From \eqref{eq:objective-step4} we arrive to the following lower bound on $D^*(\mu)$:
\begin{align}
D^*(\mu)\geq (\mu N -j)\tilde{D}(j+1) -(\mu N -j-1)\tilde{D}(j),
\end{align} 
which is a linear function of $\mu$ for a fixed value of $j\in[1:N-1]$ passing through the two points: $(\mu_1 = \frac{j}{N}, \frac{D}{L} = \tilde{D}(j))$ and $(\mu_2 = \frac{j+1}{N}, \frac{D}{L} = \tilde{D}(j+1))$. 
We obtain $N-1$ such lower bounds for every $j\in[1:N-1]$, which eventually give the lower bound on
the optimal download cost $D^*(\mu)$ as the lower convex envelope of the following $N$ points: 
\begin{align}
\left( \mu = \frac{t}{N}, \ \frac{D}{L}  =\tilde{D}(t) = 1+\frac{1}{t}+\frac{1}{t^2}+\cdots+ \frac{1}{t^{K-1}} \right), \quad \forall t\in[1:N],
\end{align}
which completes the converse proof of Theorem~\ref{theorem2}.

\section{Proof of Theorem \ref{theorem2}: Achievability for General $(N,K,\mu)$}
\label{sec:achiev}
The achievability proof of the optimal download cost given in  Theorem \ref{theorem2} has two main parts: a) the storage design (i.e., what to store across $N$ databases) subject to storage constraints; and b) the design of the PIR scheme from storage constrained databases. We next describe our placement scheme while satisfying the storage constraint at each database. In particular, we focus on the storage points $\mu KL$ for $\mu = t/N$ and $t\in[1:N]$. Once we achieve a scheme for these storage points, the lower convex envelope is also achieved using memory sharing as discussed in Claim~\ref{claim1}. 

\subsection{Storage Placement Scheme for $\mu = t/N$ and $t\in[1:N]$}
This storage placement scheme is inspired by the placement strategy proposed in the work on coded caching \cite{Maddah}.
For a fixed parameter $t\in [1:N]$, we take each message $W_k$ and sub-divide it into $\binom{N}{t}$ equal sized sub-messages. We then label each sub-message with a unique subset $\mathcal{S}\subseteq[1:N]$ of size $t$. Therefore, the message $W_k$ can be expressed as:
\begin{align}
W_k = \underset{\substack{\mathcal{S}\subseteq [1:N]\\ \vert \mathcal{S} \vert =t}}{\cup} W_{k,\mathcal{S}}.
\end{align}
Using this message splitting scheme, we propose the databases storage placement scheme as follows: for every message, each database stores all sub-messages which contain its index.
%Furthermore, we define $W_{\mathcal{S}} =\cup_{k\in[1:K]} W_{k,\mathcal{S}}$ as the set of all sub-messages stored by the set of databases in the set $\mathcal{S}$. 
 For instance, consider an example of $t=2$, $N=3$ databases, and $K=2$ messages (say $W_1$, and $W_2$). We first sub-divide the two messages into $\binom{3}{2}= 3$  sub-messages as $W_1=\{W_{1,\{1,2\}}, W_{1,\{2,3\}}, W_{1,\{1,3\}}\}$, and $W_2=\{W_{2,\{1,2\}}, W_{2,\{2,3\}}, W_{2,\{1,3\}}\}$. Therefore, the storage scheme is as follows:
\begin{itemize}
\item DB$_1$ stores $Z_1 % = \{W_{\{1,2\}}, W_{\{1,3\}}\}
=\{W_{1,\{1,2\}}, W_{1,\{1,3\}}, W_{2,\{1,2\}}, W_{2,\{1,3\}}\}$;
\item DB$_2$ stores $Z_2 % = \{W_{\{1,2\}}, W_{\{2,3\}}\}
= \{W_{1,\{1,2\}}, W_{1,\{2,3\}}, W_{2,\{1,2\}}, W_{2,\{2,3\}}\}$;
\item DB$_3$ stores $Z_3 % = \{W_{\{1,3\}}, W_{\{2,3\}}\}
=\{W_{1,\{1,3\}}, W_{1,\{2,3\}}, W_{2,\{1,3\}}, W_{2,\{2,3\}}\}$.
\end{itemize}

Furthermore, we assume that each sub-message is of size $t^{K}$ bits. Hence the total size of each message, $L$, is given as $L= \binom{N}{t}t^{K}$. 
We next verify that the above scheme satisfies the storage constraint. To this end, we note that for every message, each database stores $\binom{N-1}{t-1}$ sub-messages (this corresponds to the number of sub-sets of databases of size $t$ in which the given database is present). 
Hence, the total storage necessary for any database is given as:
\begin{align}
\underbrace{K}_{\substack{\text{Total number}\\ \text{of messages}}}\times \underbrace{\binom{N-1}{t-1}}_{\substack{\text{Number of submessages}\\ \text{per message per database}}}\times \underbrace{t^{K}}_{\substack{\text{Size of each}\\ \text{submessage}}}
&= \frac{t}{N}\times K\times N\times \binom{N-1}{t-1}\times t^{K-1}\nonumber\\
&= \frac{t}{N}\times K\times \left(\binom{N}{t}t^{K}\right)= \frac{t}{N}\times K\times L= \mu KL.\nonumber
\end{align}
This shows that the proposed scheme satisfies the storage constraints for every database.

\begin{table*}[t]
	\renewcommand{\arraystretch} {1.3}
	\caption{Fixed query structure to DB$_n$: Total vs. desired downloaded bits}
	\label{table4}
	\centering 
	\begin{tabular}{|c|c|c|c|c|}
\specialrule{.15em}{0em}{0em} 
		Blocks & Stages & Tuple&Total bits &Useful bits\\
\specialrule{.15em}{0em}{0em} 
		\multirow{6}{*}{\shortstack{Block $1$ \\ \\ $\mathcal{S}= \{n\}\cup\{1,2,\ldots,t-1\}$ }} & Stage $1$ & Single&$\binom{K}{1}(t-1)^0$ & $\binom{K-1}{0}(t-1)^0$ \\
		\cline{2-5}
		 &Stage $2$ & Pair&$\binom{K}{2}(t-1)^1$&$\binom{K-1}{1}(t-1)^1$\\
		\cline{2-5}
		&\vdots&\vdots&\vdots&\vdots\\
		\cline{2-5}
		&Stage $k$ &$k$-tuple&$\binom{K}{k}(t-1)^{k-1}$&$\binom{K-1}{k-1}(t-1)^{k-1}$\\
		\cline{2-5}
		&\vdots&\vdots&\vdots&\vdots\\
		\cline{2-5}
		&Stage $K$ &$K$-tuple&$\binom{K}{K}(t-1)^{K-1}$&$\binom{K-1}{K-1}(t-1)^{K-1}$\\
\specialrule{.15em}{0em}{0em} 
        		\multirow{5}{*}{\shortstack{Block $2$ \\ \\  $\mathcal{S}= \{n\}\cup\{1,2,\ldots,t-2,t\}$}} & Stage $1$ & Single&$\binom{K}{1}(t-1)^0$ & $\binom{K-1}{0}(t-1)^0$ \\
		\cline{2-5}
		&\vdots&\vdots&\vdots&\vdots\\
		\cline{2-5}
		&Stage $k$ &$k$-tuple&$\binom{K}{k}(t-1)^{k-1}$&$\binom{K-1}{k-1}(t-1)^{k-1}$\\
		\cline{2-5}
		&\vdots&\vdots&\vdots&\vdots\\
		\cline{2-5}
		&Stage $K$ &$K$-tuple&$\binom{K}{K}(t-1)^{K-1}$&$\binom{K-1}{K-1}(t-1)^{K-1}$\\
\specialrule{.15em}{0em}{0em} 
      \shortstack{ \vdots\\ \vdots} &\shortstack{ \vdots\\ \vdots} &\shortstack{ \vdots\\ \vdots}&\shortstack{ \vdots\\ \vdots}&\shortstack{ \vdots\\ \vdots}\\
\specialrule{.15em}{0em}{0em} 
           \multirow{5}{*}{\shortstack{Block $\binom{N-1}{t-1}$ \\ \\  $\mathcal{S}= \{n\}\cup \{N-t+2,\ldots, N\}$}} & Stage $1$ & Single&$\binom{K}{1}(t-1)^0$ & $\binom{K-1}{0}(t-1)^0$ \\
		\cline{2-5}
		&\vdots&\vdots&\vdots&\vdots\\
		\cline{2-5}
		&Stage $k$ &$k$-tuple&$\binom{K}{k}(t-1)^{k-1}$&$\binom{K-1}{k-1}(t-1)^{k-1}$\\
		\cline{2-5}
		&\vdots&\vdots&\vdots&\vdots\\
		\cline{2-5}
		&Stage $K$ &$K$-tuple&$\binom{K}{K}(t-1)^{K-1}$&$\binom{K-1}{K-1}(t-1)^{K-1}$\\
\specialrule{.15em}{0em}{0em} 
	\end{tabular}{}
\end{table*}

\subsection{Storage Constrained PIR Scheme for $\mu = t/N$ and $t\in[1:N]$}

We now present the storage constrained PIR scheme for any $(N,K)$ introduced in our previous work \cite{ICC2017}.
We focus on the storage parameter $\mu=t/N$ for any $t\in [1:N]$.  The general ideas of the scheme are described in the following points:

\noindent $\bullet$ \hspace{5pt} \textbf{\underline{Fixed Query Structure:}} We assume a fixed query structure at each database independent of the desired message $W_i$, where the query set $Q_n^{[i]}$ to each DB$_n$ is structured as follows: The query  $Q_n^{[i]}$ is composed of $\binom{N-1}{t-1}$ blocks, where every block is labeled by a set $\mathcal{S}\in[1:N]$ of size $t$, where $n\in\mathcal{S}$. The query block labeled with $\mathcal{S}$ only involves the sub-messages stored at the databases DB$_n$ where $n\in\mathcal{S}$, i.e., $W_{k,\mathcal{S}}$ for $k\in[1:K]$. Subsequently, since $\vert \mathcal{S}\vert =t$, there are $t$ databases sharing each block labeled with $\mathcal{S}$.
 Every query block is composed of a sequence $K$ ordered stages, where in Stage $k$, the user requests $k$-tuple coded bits composed of $k$ bits from $k$ different files, $W_{\mathcal{K}}$ for $\mathcal{K} \in[1:K]$ and $\vert\mathcal{K}\vert =k$. There are $\binom{K}{k}$ such types of $k$-tuples, and each type is designed to contain $(t-1)^{k-1}$ coded bits. The summary of the query structure per database is shown in Table~\ref{table4}, where the query set to each database is composed of $\binom{N-1}{t-1}$ blocks, each contains $K$ ordered stages,  each Stage $k$ contains $\binom{K}{k}$ types of $k$-tuples, and each type is a set of $(t-1)^{k-1}$ coded bits. In total, the query involves a total number of $\binom{N-1}{t-1} \sum_{k=1}^K \binom{K}{k} (t-1)^{k-1}$ requested coded bits.

\noindent $\bullet$ \hspace{5pt} \textbf{\underline{Privacy Guarantees:}}
The query structure for the Stage $k$, Block $S$ and DB$_n$ is shown in Table~\ref{table5}.
 In order to maintain privacy, we introduce message symmetry such that any query stage must be symmetric with respect to all the messages. Looking at Stage $k$ within any block in Table~\ref{table5}, out of the $\binom{K}{k} (t-1)^{k-1}$ coded bits requested, the desired message $W_i$ is participating in $\binom{K-1}{k-1} (k-1)^{k-1}$ coded bits, which are referred to as \textit{desired symbols.}
The user must download some extra undesired bits of the other messages in order to maintain symmetry, which are referred to as \textit{side-information}, given by the remaining $\binom{K-1}{k} (k-1)^{k-1}$ coded bits.
Furthermore, we restrict our scheme not to download the same bit multiple times from the same database, otherwise the privacy constraint may be violated. 
In addition to the symmetric structure of the queries, before initializing the PIR scheme, the user assumes a random permutation of the bits for each sub-message $W_{k,\mathcal{S}}$ according to $\bdelta^{k,\mathcal{S}}:(1,2,\ldots,L)\rightarrow(\delta^{k,\mathcal{S}}_{1},\delta^{k,\mathcal{S}}_{2}, \ldots, \delta^{k,\mathcal{S}}_{t^K})$. This procedure will become clear when we discuss the stages of the PIR scheme.

\begin{table*}[t]
	\renewcommand{\arraystretch} {1.3}
	\caption{Fixed query structure to DB$_n$ during Block $\mathcal{S}$, Stage $k$: Total vs. desired downloaded bits}
	\label{table5}
	\centering 
	\begin{tabular}{|c|c|c|c|}
\Cline{1pt}{1-4}
		Types & Instance No. & $k$-tuple Coded Symbols & Total bits \\
\Cline{1pt}{1-4}
\multirow{4}{*}{\shortstack{Type $1$ \\ \\ $\mathcal{K}=\{i\} \cup \{1,2,\ldots,k-1\}$}} &$1$ & $\mathbf{n}(W_{i,\mathcal{S}})+\mathbf{s}(\sum_{j\in \mathcal{K}\setminus  \{i\}}W_{j,\mathcal{S}})$& \multirow{9}{*}{\shortstack{(Desired bits) \\ \\  $\binom{K-1}{k-1}(t-1)^{k-1}$ }} \\
\cline{2-3}
&$2$ & $\mathbf{n}(W_{i,\mathcal{S}})+\mathbf{s}(\sum_{j\in \mathcal{K}\setminus  \{i\}}W_{j,\mathcal{S}})$& \\
\cline{2-3}
&\vdots  &\vdots &\\
\cline{2-3}
&$(t-1)^{k-1}$ & $\mathbf{n}(W_{i,\mathcal{S}})+\mathbf{s}(\sum_{j\in \mathcal{K}\setminus \{i\}}W_{j,\mathcal{S}})$& \\
\Cline{1pt}{1-3}
\vdots &\vdots  &\vdots&\\
\Cline{1pt}{1-3}
\multirow{4}{*}{\shortstack{Type $\binom{K-1}{k-1}$ \\ \\ $\mathcal{K}=\{i\} \cup \{K-k+2,\ldots,K\}$}} &$1$ & $\mathbf{n}(W_{i,\mathcal{S}})+\mathbf{s}(\sum_{j\in \mathcal{K}\setminus  \{i\}}W_{j,\mathcal{S}})$& \\
\cline{2-3}
&$2$ & $\mathbf{n}(W_{i,\mathcal{S}})+\mathbf{s}(\sum_{j\in \mathcal{K}\setminus  \{i\}}W_{j,\mathcal{S}})$& \\
\cline{2-3}
&\vdots  &\vdots&\\
\cline{2-3}
&$(t-1)^{k-1}$ & $\mathbf{n}(W_{i,\mathcal{S}})+\mathbf{s}(\sum_{j\in \mathcal{K}\setminus  \{i\}}W_{j,\mathcal{S}})$& \\
\Cline{2pt}{1-4}
\multirow{4}{*}{\shortstack{Type $\binom{K-1}{k-1}+1$ \\ \\ $\mathcal{K}= \{1,2,\ldots,k\}$\\$ (i\not\in \mathcal{K})$}} &$1$ & $\sum_{j\in \mathcal{K}}\mathbf{n}(W_{j,\mathcal{S}})$& \multirow{9}{*}{\shortstack{(Side Info.) \\ \\  $\binom{K-1}{k}(t-1)^{k-1}$ }}   \\
\cline{2-3}
&$2$ & $\sum_{j\in \mathcal{K}}\mathbf{n}(W_{j,\mathcal{S}})$& \\
\cline{2-3}
&\vdots  &\vdots&\\
\cline{2-3}
&$(t-1)^{k-1}$ & $\sum_{j\in \mathcal{K}}\mathbf{n}(W_{j,\mathcal{S}})$& \\
\Cline{1pt}{1-3}
\vdots &\vdots  &\vdots &\\
\Cline{1pt}{1-3}
\multirow{4}{*}{\shortstack{Type $\binom{K}{k}$ \\ \\ $\mathcal{K}= \{K-k+1,\ldots,K\}$\\$ (i\not\in \mathcal{K})$}} &$1$ & $\sum_{j\in \mathcal{K}}\mathbf{n}(W_{j,\mathcal{S}})$ & \\
\cline{2-3}
&$2$ & $\sum_{j\in \mathcal{K}}\mathbf{n}(W_{j,\mathcal{S}})$ & \\
\cline{2-3}
&\vdots  &\vdots &\\
\cline{2-3}
&$(t-1)^{k-1}$ & $\sum_{j\in \mathcal{K}}\mathbf{n}(W_{j,\mathcal{S}})$ & \\
\specialrule{.1em}{0em}{0em} 
	\end{tabular}{}
\end{table*}

\noindent $\bullet$ \hspace{5pt} \textbf{\underline{Exploiting Side Information:}}
At each stage $k$, the user uses the side information downloaded during the previous stage $k-1$ to pair with new desired bits.
However, the side information downloaded from a database at Stage $k-1$ cannot be used to download new desired bits from the same database in Stage $k$ in order to maintain privacy. 
In other words, the number of desired bits that can be downloaded from a database, say DB$_n$, at Stage $k$ is given by the side information downloaded from the remaining databases at the Stage $k-1$, and also stored within DB$_n$.
The exploitation of side-information is carefully designed to account for the limited storage capabilities of the databases. At stage $k$, when the user downloads a $k$-tuple coded bit, he only pairs bits that belong to the same set of database $\mathcal{S}$, such that the resultant coded bit is available at all the databases who indexes are in the set $\mathcal{S}$. The resultant coded bits belong to the query block labeled by $\mathcal{S}$ as discussed before in Table~\ref{table4}, and henceforth we treat the sub-messages that are stored in the same set of databases $\mathcal{S}$ independently in the PIR scheme.
In Table~\ref{table5}, we demonstrate the types of bits requested from DB$_n$ during Stage $k$ in a block characterized by $\mathcal{S}$, where $n\in \mathcal{S}$. We use the notation $\mathbf{n}(W_{j,\mathcal{S}})$ to denote  a new bit of $W_{j,\mathcal{S}}$  that has not been downloaded from any database before; while $\mathbf{s}(\sum_{j\in \mathcal{K}\setminus  \{i\}}W_{j,\mathcal{S}})$ where $\vert \mathcal{K}\vert=k$ denotes a side information coded bit the user already downloaded from the previous Stage $k-1$ from other database. It is important to note here that the side information bits are considered fresh from DB$_n$ perspective.

\noindent $\bullet$ \hspace{5pt} \textbf{\underline{Detailed Description of the stages:}}
As we discussed before, for every DB$_n$ and a desired message $W_i$, the query set $Q^{[i]}_n$ is divided into $\binom{N-1}{t-1}$ blocks labeled with sets $\mathcal{S}\subseteq[1:N]$ and $n\in\mathcal{S}$.
Each  block $\mathcal{S}$ is composed only of the sub-messages $W_{k,\mathcal{S}}$, where $k\in[1:K]$, and the databases in the set $\mathcal{S}$.
The user assumes a random permutation of the bits for each sub-message $W_{k,\mathcal{S}}$ according to $\bdelta^{k,\mathcal{S}}:(1,2,\ldots,L)\rightarrow(\delta^{k,\mathcal{S}}_{1},\delta^{k,\mathcal{S}}_{2}, \ldots, \delta^{k,\mathcal{S}}_{t^K})$, i.e., considering the sub-message bits in the order $\left(w_{k,\mathcal{S}}^{\left(\delta^{k,\mathcal{S}}_{1}\right)},w_{k,\mathcal{S}}^{\left(\delta^{k,\mathcal{S}}_{2}\right)},\ldots,w_{k,\mathcal{S}}^{\left(\delta^{k,\mathcal{S}}_{t^K}\right)}\right)$.
As demonstrated in Table \ref{table4}, the user chooses the new bits of the sub-message $\mathbf{n}(W_{j,\mathcal{S}})$ according to the permutation $\bdelta^{j,\mathcal{S}}$, e.g., the following sequence of calls $\left(\mathbf{n}(W_{1,\mathcal{S}}),\mathbf{n}(W_{2,\mathcal{S}}),\mathbf{n}(W_{1,\mathcal{S}})+\mathbf{n}(W_{2,\mathcal{S}})\right)$ correspond to $\left(w_{1,\mathcal{S}}^{\left(\delta^{1,\mathcal{S}}_1\right)},w_{2,\mathcal{S}}^{\left(\delta^{2,\mathcal{S}}_1\right)},w_{1,\mathcal{S}}^{\left(\delta^{1,\mathcal{S}}_2\right)}+w_{1,\mathcal{S}}^{\left(\delta^{2,\mathcal{S}}_2\right)}\right)$. 
Considering  Block $\mathcal{S}$, we present the $K$ stages of the general PIR scheme for DB$_n$ in details, and assume that the user is interested in privately retrieving message $W_i$.

\noindent \textbf{Stage $1$}: In the first stage, the user starts by downloading single uncoded bits from DB$_n$. 
The user first downloads one bit from the desire sub-message $W_{i,\mathcal{S}}$. 
In order to maintain privacy, we introduce message symmetry, and the user downloads  one bit from remaining $K-1$ sub-messages, $W_{j,\mathcal{S}}$ where $j\in[1:K]\setminus i$.
As a summary as shown in the first row of each Block in Table \ref{table4}, the user downloads in total $\binom{K}{1}$ bits, out of which the number of desired bits are $\binom{K-1}{0}=1$. The total number of bits downloaded of each sub-message $W_{j,\mathcal{S}}$ across all databases DB$_{n'}$, where $n'\in\mathcal{S}\setminus \{n\}$, is given as $\binom{K-1}{0}(t-1)$, and for all undesired sub-messages $W_{j,\mathcal{S}}$, where $j\in[1:K]\setminus \{i\}$, we obtain $\binom{K-1}{1}(t-1)$ undesired bits as side information to be used to get desired bits from DB$_n$ in the next stage, Stage $2$.

\noindent \textbf{Stage $2$}: In the second stage, the user downloads pairs of bits, where the desired bits are coded with the undesired bits previously downloaded as side information from Stage $1$. 
As previously discussed, the number of desired bits that can be downloaded from DB$_n$ is given by the side information downloaded from the remaining databases at Stage $1$, which is given by $\binom{K-1}{1}(t-1)$ bits.
 Furthermore, in order to maintain message symmetry, the user downloads extra coded pairs of undesired bits, which will form the side information of the next stage.  For every one of the $\binom{K-1}{2}$ undesired pair types, i.e., $\mathcal{K}\subseteq [1:K]\setminus \{i\}$ and $\vert \mathcal{K}\vert =2$, the user downloads $(t-1)$ bits from DB$_n$, or total $\binom{K-1}{2}(t-1)$  undesired bits.
 As a summary as shown in the second row of each Block in Table \ref{table4}, the user downloads from DB$_n$ in total $\binom{K}{2}(t-1)$ bits, out of which the number of desired bits are $\binom{K-1}{1}(t-1)$.
 The total number of undesired pairs  downloaded across all databases DB$_{n'}$, where $n'\in\mathcal{S}\setminus \{n\}$, is given as $\binom{K-1}{2}(t-1)^2$ as side information to be used to get desired bits from DB$_n$ in the next stage, Stage $3$.

\noindent \textbf{Stage $k$}: In the general Stage $k$, 
 we download $k$-tuples of bits composed of $k$ different messages. The total number of undesired side information obtained from the previous stage $k$, which can be used to obtain desired bits from DB$_n$, is $\binom{K-1}{k-1}(t-1)^{k-1}$ $k-1$-tuples, which also gives the number of desired bits that can be obtained. By symmetrizing the stage, we obtain a total number of $\binom{K}{k}(t-1)^{k-1}$ $k$-tuples, out of which only $\binom{K-1}{k-1}(t-1)^{k-1}$ desired bits are obtained.
 
 The following two lemmas proves that the achievable scheme is correct and private, and completes the proof of achievability of the capacity expression in Theorem~\ref{theorem2}.
\begin{lemma}
The achievable scheme is correct and achieves the trade-off stated in Theorem~\ref{theorem2}.
\end{lemma}
\begin{proof}
The correctness of the scheme follows from a counting argument. %For each one of the $\binom{N-1}{t-1}$ blocks of the scheme, a number of $t$ databases participate in the retrieval process. Each one of the $t$ databases is sending in the $K$ stages a number of $\sum{k=1}^K \binom{K-1}{k-1} (t-1)^{k-1}$  coded symbols including distinct desired bits. Therefore, the total number of coded symbols including distinct desired bits across all databases and all blocks is given by
%\begin{align}
%\underbrace{\binom{N-1}{t-1}}_{\text{number of blocks}} \underbrace{t}_{\substack{\text{number of databases}\\ \text{ participating in the block}}} \underbrace{\sum_{k=1}^K \binom{K-1}{k-1} (t-1)^{k-1}}_{\substack{\text{number of coded desired bits}\\ \text{ per database in a block}}}
%\end{align}
Let us calculate the total number of desired bits and the total number of  downloaded bits. For each database, there are $\binom{N-1}{t-1}$ blocks. For each block, there are $K$ ordered stages, where in Stage $k$, a total number of $\binom{K}{k}(t-1)^{k-1}$ coded $k$-tuple bits are downloaded, out of which $\binom{K-1}{k-1}(t-1)^{k-1}$ are desired bits.
\begin{align}
\text{Desired}&\text{ bits (per DB)} \nonumber\\
 &=\binom{N-1}{t-1}\sum_{k=1}^{K}\binom{K-1}{k-1}(t-1)^{k-1} = \binom{N-1}{t-1}\sum_{k'=0}^{K-1}\binom{K-1}{k'}(t-1)^{k'}\nonumber\\
&= \binom{N-1}{t-1}\left(1 +(t-1) \right)^{K-1} =\binom{N-1}{t-1} t^{K-1} = \frac{1}{N}\times\binom{N}{t}t^{K}.\nonumber\\
\text{Total do}&\text{wnloaded bits (per DB)} \label{eq:total-desired}\\
&=\binom{N-1}{t-1}\sum_{k=1}^{K}\binom{K}{k}(t-1)^{k-1} = \binom{N-1}{t-1}\frac{1}{t-1}\sum_{k=1}^{K-1}\binom{K}{k}(t-1)^{k}\nonumber\\
&=  \binom{N-1}{t-1}\frac{1}{t-1}(t^K -1) =\binom{N-1}{t-1}(t^{K-1}+t^{K-2}+\cdots+t+1).\label{eq:total-download}
\end{align}

From all the $N$ databases, the user is able to recover $L=\binom{N}{t}t^{K}$ coded bits of the desired message. 
Furthermore, since all these desired bits are distinct and only paired with side information previously downloaded, the user can decode all the $L$ bits of the desired message and hence the scheme is proven correct.
The download cost $D(\mu)$ of the proposed storage constrained PIR scheme when $\mu=t/N$ and $t\in[1:N]$ is given as
\begin{align}
D(\mu)&= \frac{N\times \text{Total Downloaded bits (per DB)}}{N \times \text{Desired bits (per DB)}}\nonumber\\
&= \frac{\binom{N-1}{t-1}(t^{K-1}+t^{K-2}+\cdots+t+1)}{\binom{N-1}{t-1}t^{K-1}}\nonumber\\
&= 1 + \frac{1}{t} + \frac{1}{t^{2}}+ \cdots + \frac{1}{t^{K-1}} =\tilde{D}(t).
\end{align}

Using the memory sharing concept in Claim~\ref{claim1}, we can achieve the lower convex envelope of the
following $N$ achievable points:
\begin{align}
\left( \mu = \frac{t}{N}, \ D(\mu)  = \tilde{D}(t)=1+\frac{1}{t}+\frac{1}{t^2}+\cdots+ \frac{1}{t^{K-1}} \right), \quad \forall t\in[1:N],
\end{align}
which matches the trade-off given in Theorem \ref{theorem2}.
\end{proof}

\begin{lemma}
The achievable scheme is private.
\end{lemma}
\begin{proof}
The privacy of the scheme is guaranteed using three factors:
\begin{enumerate}
\item The query structure to each database is symmetric with respect to all the $K$ messages.
 \item All the bits downloaded from a database are distinct for all messages, i.e., no bit is downloaded more than once.
\item The user randomly permutes the bits indexes of each message before designing the queries. 
\end{enumerate}
\noindent These three steps make the query to any database independent from the index of the desired message $W_i$, i.e., $I(i;Q_n^{[i]}) = 0$ for all $n\in[1:N]$.

To prove this formally, we calculate the probability of an arbitrary query realization. In each Block $\mathcal{S}$, the user downloads a number of $t^{K-1}$ bits of the desired sub-message $W_{i,\mathcal{S}}$ from each database DB$_n$, where $n\in\mathcal{S}$. Due to the symmetric query structure, the number of bits of every undesired sub-message, $W_{j,\mathcal{S}}$ where $j\in [1:K]\setminus i$, involved in the block per database is $t^{K-1}$ bits as well. According to the permutation $\bdelta^{j,\mathcal{S}}$, the user permutes the bits indexes of each sub-message $W_{j,\mathcal{S}}$, and therefore, the $t^{K-1}$ bits involved in the query block $\mathcal{S}$ for DB$_n$ are given as: $(w_{j,\mathcal{S}}^{({\ell^{i,j,n,\mathcal{S}}_{1}})},w_{j,\mathcal{S}}^{({\ell^{i,j,n,\mathcal{S}}_{2}})},\ldots,w_{j,\mathcal{S}}^{({\ell^{i,j,n,\mathcal{S}}_{t^{K-1}}})})$, where $(\ell^{i,j,n,\mathcal{S}}_1,\ldots,\ell^{i,j,n,\mathcal{S}}_{t^{K-1}})\subset \bdelta^{j,\mathcal{S}}$ and can be found according to the order of the databases to which the user sends the queries as well as the desired message index $i$. Without loss of generality, an arbitrary realization of $(w_{j,\mathcal{S}}^{({\ell^{i,j,n,\mathcal{S}}_{1}})},w_{j,\mathcal{S}}^{({\ell^{i,j,n,\mathcal{S}}_{2}})},\ldots,w_{j,\mathcal{S}}^{({\ell^{i,j,n,\mathcal{S}}_{t^{K-1}}})})$ can be considered to be $(w_{j,\mathcal{S}}^1,w_{j,\mathcal{S}}^2,\ldots,w_{j,\mathcal{S}}^{t^{K-1}})$.
The probability of such realization of the query block $\mathcal{S}$ at DB$_n$ is then given as:
\begin{align}
\Pr&\left[\cup_{j=1}^K \left((w_{j,\mathcal{S}}^{({\ell^{i,j,n,\mathcal{S}}_{1}})},w_{j,\mathcal{S}}^{({\ell^{i,j,n,\mathcal{S}}_{2}})},\ldots,w_{j,\mathcal{S}}^{({\ell^{i,j,n,\mathcal{S}}_{t^{K-1}}})}) =(w_{j,\mathcal{S}}^1,w_{j,\mathcal{S}}^2,\ldots,w_{j,\mathcal{S}}^{t^{K-1}})\right) \right]\nonumber \\
&= \prod_{j=1}^K \Pr\left[(w_{j,\mathcal{S}}^{({\ell^{i,j,n,\mathcal{S}}_{1}})},w_{j,\mathcal{S}}^{({\ell^{i,j,n,\mathcal{S}}_{2}})},\ldots,w_{j,\mathcal{S}}^{({\ell^{i,j,n,\mathcal{S}}_{t^{K-1}}})}) =(w_{j,\mathcal{S}}^1,w_{j,\mathcal{S}}^2,\ldots,w_{j,\mathcal{S}}^{t^{K-1}})\right]\nonumber\\
&= \prod_{j=1}^K \Pr\left[\ell^{i,j,n,\mathcal{S}}_1=1,\ \ell^{i,j,n,\mathcal{S}}_2=2,\ \ldots,\ \ell^{i,j,n,\mathcal{S}}_{t^{K-1}}=t^{K-1} \right] \nonumber\\
&= \left(\left(\frac{1}{t^K}\right)\left(\frac{1}{t^K-1}\right)\times\cdots\times\left(\frac{1}{t^K-t^{K-1}+1}\right)\right)^K,
\end{align} 
which is independent of the desired message index $i$, i.e., $I(i;Q_n^{[i]}) = 0$. 

Next, we show that the privacy condition in \eqref{eq:privacy-const} is satisfied through the following:
 \begin{align}
 I(i;Q_n^{[i]},A_n^{[i]},\W_{[1:K]},\Z_{[1:N]})\hspace{-1pt} = \hspace{-1pt}  I(i;Q_n^{[i]}) \hspace{-1pt} + \hspace{-1pt} I(i;\W_{[1:K]},\Z_{[1:N]}|Q_n^{[i]})\hspace{-1pt}  + \hspace{-1pt} I(i;A_n^{[i]}|\W_{[1:K]},\Z_{[1:N]},Q_n^{[i]}) \hspace{-1pt} =\hspace{-1pt} 0,
 \end{align}
 where the second term, $I(i;\W_{[1:K]},\Z_{[1:N]}|Q_n^{[i]})$ is zero since the user does not know the messages or the actual contents of the databases when he chooses the index $i$ and designs the query $Q_n^{[i]}$, while the third term, $I(i;A_n^{[i]}|\W_{[1:K]},\Z_{[1:N]},Q_n^{[i]})$ is zero since the answer $A_n^{[i]}$ is fully determined by DB$_n$ knowing the query $Q_n^{[i]}$ and the part of the storage content $Z_n$.
\end{proof}

\section{Conclusions}
In this paper, we characterized the optimal download cost of PIR for uncoded storage constrained databases. In particular, for any $(N,K)$,  we show that the optimal trade-off between the storage parameter, $\mu\in[1,N]$, and the download cost, $D(\mu)$, is given by the lower convex hull of the pairs $(\frac{t}{N}, \left(1+ \frac{1}{t}+ \frac{1}{t^{2}}+ \cdots + \frac{1}{t^{K-1}}\right))$ for $t\in[1:N]$. 
The main technical contribution of this paper is obtaining lower bounds on the download cost for PIR as a function of storage, which matches the achievable scheme in \cite{ICC2017}, and hence characterizes the optimal trade-off.
We first arrived to a lower bound on the download cost, which is valid for both coded and uncoded placement strategies. We then specialized the obtained bound for uncoded placement strategies, which helps in obtaining a linear program subject to message size and storage constraints. Solving this linear program, we arrive at a set of $N-1$ lower bounds, where each bound is tight in a certain range of storage.
There are several interesting future directions  on this important variation of storage-constrained PIR such as a) settling the tradeoff with coded storage allowed at databases, b) colluding databases and c) introducing additional reliability constraints on storage, such that data must be recoverable from any $N$ out of $M$ databases.

\bibliographystyle{unsrt}
\bibliography{./paper-arxiv.bib}

\end{document}